\documentclass[12pt,reqno,english]{article}
\usepackage[T1]{fontenc}
\usepackage[latin9]{inputenc}
\usepackage[letterpaper]{geometry}
\usepackage{color}
\usepackage{babel}
\usepackage{amsmath}
\usepackage{amsthm}
\usepackage{amssymb}
\usepackage{setspace}
\usepackage[authoryear]{natbib}
\usepackage{tikz}
\usepackage{tikz-3dplot}
\usepackage[unicode=true,pdfusetitle,
 bookmarks=true,bookmarksnumbered=false,bookmarksopen=false,
 breaklinks=false,pdfborder={0 0 1},backref=false,colorlinks=true]
 {hyperref}
\hypersetup{ pdfborderstyle=,linkcolor=magenta,citecolor=blue}

\makeatletter
\usepackage{caption}
\usepackage{subcaption}
\usepackage{bm}

\usepackage{color}
\usepackage[plain]{fullpage}%Pepe added this for wider margins
\usepackage{babel}

\usepackage{array}
\usepackage{booktabs}
\usepackage{multirow}

\usepackage{psfrag}\usepackage{epsf}\usepackage{enumerate}
\usepackage{url}% not crucial - just used below for the URL 

\usepackage{amsfonts}
\usepackage{pdflscape}
\usepackage{threeparttable}

\theoremstyle{plain}
\newtheorem{lem}{\protect\lemmaname}[section]
\theoremstyle{plain}

\newtheorem{theorem}{\textbf{Theorem}}%[section]
\newtheorem{lemma}{Lemma}%[section]
\newtheorem{assumption}{Assumption}%[section]
\theoremstyle{remark}

\providecommand{\remarkname}{Remark}

\usepackage{verbatim}

\makeatother

\providecommand{\lemmaname}{Lemma}
\providecommand{\propositionname}{Proposition}

\definecolor{myred}{RGB}{220, 80, 60}
\definecolor{mygreen}{RGB}{0, 150, 135}

\begin{document}
\title{Local Asymptotics for Treatment Choice \\ with Partial Identification\thanks{We thank Jack Porter for encouragement at an early stage of this project and Karun Adusumilli, Tim Armstrong, Xinyue Bei, Tim Christensen, Kei Hirano, Toru Kitagawa, Brendan Kline, Han Xu, and Kohei Yata as well as audiences at Bonn, TAMU, UT Austin, as well as participants at EcoSta 2026 and the Iowa Econometrics Conference for feedback. William Pan, Elliott Serna and Sara Yoo provided excellent research assistance. We gratefully acknowledge financial
support from the NSF under grant SES-2315600.}}

\author{Jos\'e Luis Montiel Olea\thanks{Email: jlo67@cornell.edu}\and Chen Qiu\thanks{Email: cq62@cornell.edu}  \and J\"{o}rg Stoye\thanks{Email: stoye@cornell.edu}}

\date{August 2026}

\maketitle

\vspace{-0.5cm}

\begin{abstract}
We provide a new  asymptotic framework to
derive approximately optimal treatment assignments when sampling noise from data is compounded by fundamental uncertainty due to partial identification. We recenter the reduced-form parameter around its \emph{least-favorable} configuration and consider drifting parameter sequences that yield both diminishing levels of sampling uncertainty and of partial identification. We characterize the limiting decision problem as a normal location shift model with a suitable limiting identified set.
We apply our results to treatment choice problems with contaminated outcomes, to robust welfare analyses with partially identified consumer surplus, and to the problem of aggregating experimental estimates for policy adoption.

\end{abstract}

\newpage{}

\onehalfspacing

\section{Introduction}

A policy maker must decide between implementing a new policy or preserving the status quo. Her data provide information about the potential benefits of these two options. Unfortunately, these data only \emph{partially identify} payoff-relevant parameters and may therefore not reveal, even in large samples, the correct course of action. The question of how to best use the available data in such \emph{treatment choice problems with partial identification} has received recent attention in the literature; see, for example, \cite{ishihara2021}, \cite{yata2021}, \cite*{christensen2022optimal}, and \cite{kido2023locallyasymptoticallyminimaxstatistical}. Several interesting problems that arise in empirical research can be recast  using this framework; see \cite{MQS2023decision}. 

In this paper, we use Le Cam's limits of experiments framework \citep{le2000asymptotics,le2012asymptotic} to propose a simple local asymptotic approximation for a general class of treatment choice problems with partial identification.  In the problems that we study, the decision maker observes a random sample of size $n$ of a vector-valued outcome variable $Y_i$. There is a point-identified parameter $\gamma \in \mathbb{R}^{d_m}$ that flexibly determines the distribution of the data. There is also a payoff-relevant parameter, $U^{*} \in \mathbb{R}$, that is partially identified. As in \cite{christensen2022optimal}, we assume that the decision maker can use $\gamma$ to deduce restrictions on $U^*$. The loss function for the decision problem is given by the usual regret loss $L(a,U^*,\gamma) := U^* ( \mathbf{1}\{U^* \geq 0\}-a)$, where we interpret the action $a \in [0,1]$ as the fraction of the population exposed to the new policy.   

We show that the treatment choice problems that we have just described above can be conveniently approximated by a simpler problem. In the limiting decision problem, the decision maker observes a $d_m$-valued multivariate normal vector $\Delta$ with unknown mean, $h$, but known variance. There is a payoff-relevant parameter $\mu^* \in \mathbb{R}$ that, for each $h$, is known to belong to a nonempty interval $I_{\infty}(h):=[\underline{I}_{\infty}(h),\overline{I}_{\infty}(h)]$. We refer to $I_{\infty}(h)$ as the \emph{limiting (local) identified set}. The endpoints of the interval are possibly nonlinear transformations of $h$. The limiting loss function takes the form $L_{\infty}(a,\mu^*,h):= \mu^*(\mathbf{1}\{\mu^* \geq 0\}-a)$. 

Our Theorem \ref{thm: general} then shows that if we find a minimax rule for the limiting decision problem, we can  construct a sequence of decision rules in the original problem that are \emph{locally asymptotically minimax regret} optimal (in a sense we make precise). We then provide specific solutions to the limiting problem when the endpoints of $I_{\infty}(h)$ are affine (Theorems \ref{thm:full-diff} and \ref{thm:full-diff-asymptotic}), which occurs when the bounds of the original identified set for $U^*$ are differentiable. Since in many partially identified models the bounds of the original identified set are only directionally differentiable, we also consider the limiting problem without imposing the affinity of the bounds of the set $I_{\infty}(h)$. In this case, we show that when the parameter space for the original problem is i) convex and centrosymmetric (in a sense we make precise) and ii) the original bounds for the identified set are directionally differentiable, we can also solve the limiting problem (Theorems \ref{thm:non-diff} and \ref{thm:non-diff-asymptotic}).\footnote{In contrast, it is well known that directional differentiability causes irregularity for inference  problems \citep{hirano2012impossibility}.}   

We use two key ideas to justify our asymptotic approximations. The first key idea is to consider drifting parameter sequences over which the identified set for the payoff-relevant parameter shrinks to a point (at a root-$n$ rate) as the sample size grows large. We argue that local asymptotic approximations based on sequences that lead to \emph{near point-identification} ensure a nontrivial trade-off between identification and estimation in the limit as the sample size grows large. The idea of shrinking the identified set as the sample size grows large is based on the seminal work of \cite{armstrong2021sensitivity}, who use this device to conduct sensitivity analysis in point and partially-identified models defined by moment conditions which are only required to hold in an approximate sense. Drifting parameter sequences that lead to near point-identification are also discussed in \cite{song2014point} and \cite{Stoye2009ecma}. 

The second key idea is that we make an explicit recommendation about which reduced-form parameter value to localize about.  In point-identified treatment choice problems, one often localizes at a point such that the payoff-relevant parameter equals zero  \citep{HiranoPorter2009,HiranoPorter2020}. This choice is usually motivated as being the hardest (or least-favorable) for the decision maker. In partially-identified treatment choice problems, choosing where to localize is a more delicate matter. Indeed, any reduced-form parameter value at which the identified set for the payoff-relevant parameter contains zero makes the treatment choice problem hard,  but different values of the reduced-form parameter imply different degrees of partial identification. We argue that in our partially-identified treatment choice problem, there also exists a natural notion of a least-favorable or hardest case, which we formalize and advocate to localize around. 

We apply our results to three concrete problems. First, we consider treatment choice problems with binary outcomes that may be contaminated or corrupted \citep{horowitz1995identification}. Second, we present a stylized example in robust welfare analyses \citep{kang2025robustness}. Third, we revisit an example in the evidence aggregation framework in \cite{ishihara2021}.

{\scshape Related Literature.} \cite{Manski2000,manski2004statistical} and \cite{Dehejia2005} advocated using a decision-theoretic framework to study treatment choice problems. \cite{Manski2000,manski2005social,manski2007identification} and \citet{Stoye07} provide optimal treatment rules assuming the true distribution of the data is known. 
\cite{stoye2012minimax,stoye2012new}, \cite{yata2021}, \cite{ishihara2021} and \cite{MQS2023decision} focus on finite-sample minimax regret optimal rules; \citet{aradillas2024robust} on multiple prior minimax regret rules; and \cite{kitagawa2023treatment} on mean-squared regret optimal rules. 

In important recent work, \cite{christensen2020robust,christensen2022optimal} pioneered a local asymptotic approach for discrete choice problems when payoffs depend on a partially-identified parameter (binary treatment choice problems with partial identification being a special case). Their main innovation is to \emph{profile out} the partially-identified parameter in the loss function and recenter a decision rule's risk around that of an oracle who knows a reduced-form parameter that restricts the payoffs for the different choices. While they focus on a Bayes criterion based on profiled loss, \cite{xu2026asymptotic} extends their methodology to a minimax criterion and also allows for randomized decisions. Our paper shares a similar motivation to these two papers and continues the research agenda of \cite{christensen2022optimal}. An important difference is that we do not consider a loss function that profiles out the partially-identified parameter. Our asymptotic analysis with a shrinking identified set enables us to respect the problem's original loss function and retain a more traditional minimax regret perspective on the finite-sample problem. However, to achieve this, we must be specific about where to localize the reduced-form parameter, whereas \citet{christensen2022optimal} and \citet{xu2026asymptotic} can be very flexible in this choice. Another difference from \citet{xu2026asymptotic} is that we find that approximately minimax regret optimal decisions even when the bounds of the identified set are only directionally differentiable. Thus, we view our work as a natural complement to the work of \cite{christensen2020robust,christensen2022optimal,xu2026asymptotic}. \cite{song2014point} studied local asymptotic minimax decisions for interval-identified parameters by profiling out the mean-squared error risk function; see also \cite{kido2023locallyasymptoticallyminimaxstatistical} for analogous treatments in   treatment choice problems. We think all of these papers (and ours) contribute to creating a bridge to existing results concerning limit experiments in point-identified settings \citep{HiranoPorter2009,HiranoPorter2020}. 

Drifting parameter sequences of parameter values analogous to the near point-identification asymptotics used in this paper have a long history in econometrics: for example, the local-to-zero asymptotics in the linear instrumental variables model of \cite{staiger1997instrumental}, local-to-unit-root asymptotics in the study of nearly integrated autoregressive processes in \cite{phillips1988regression} and \cite{dou2021generalized}, local-to-identification-failure analysis for Generalized Method of Moments models in \cite{andrews2022optimal}. Local asymptotics are also used in the work of \cite{hirano2017forecasting} to compare the performance of different forecasting procedures under model uncertainty. See also \cite{powell2017identification} for further examples. 

An important motivation for our paper is the work of \citet{stoye2012minimax}, \citet{ishihara2021}, \citet{yata2021}, and \citet{MQS2023decision}. All of these papers consider treatment choice problems with partial identification where a point-identified (reduced-form) parameter and  a payoff-relevant but partially-identified parameter.  Crucially, the analysis of minimax-regret optimal decision rules is conducted under the assumption that the statistical information for the point-identified parameter is revealed only through the mean of a Gaussian location shift model. \cite{yata2021} makes significant progress characterizing exact finite-sample minimax-regret optimal treatment assignment rules by assuming that the parameter space is convex and centrosymmetric. As noted by \cite{hirano25}, \emph{``Given the shifted normal form for the reduced-form model, it is tempting to think of this statistical setup
as corresponding to the limit experiment in the locally asymptotically normal case, so that these results
can be interpreted as simultaneously providing asymptotic optimality results when the reduced-form
model is regular and parametric. However, making this connection is somewhat complicated.''} Our approach justifies the aforementioned Gaussian construction as a \emph{bona fide} limit experiment (subject to a suitable adjustment of the limit identified set). While our results assume that the statistical model is parametric (in the sense of being indexed by a finite-dimensional parameter), the results in \citet[][Appendix C and part of Section 4]{armstrong2021sensitivity} can be used to provide local asymptotic justifications  for treatment choice problems in semiparametric models where the identified set for the payoff-relevant parameter is characterized via  moment conditions. Beyond the use of a parametric statistical model, we note that a key difference is that our asymptotic analysis  simplifies the geometry of the identified set considerably in the limit. As a result, we are able to provide solutions to the limiting problem even beyond a centrosymmetric parameter space. 

Finally, our main results are most helpful to a researcher who believes that identification and  estimation-induced uncertainty are roughly equally important for their finite-sample problem.  If the dataset at hand is so large that they believe the lack of identification is the sole issue, one is free to set a slower drifting rate for the identification-induced uncertainty. In this case, we simply recommend the  plug-in or ``as-if'' \citep{manski2021econometrics} approach that just solves population-level problems after plugging in consistent estimators for identifiable quantities. Alternatively, if a researcher thinks that the model-induced uncertainty vanishes faster than estimation-induced uncertainty (because data is so scarce and the degree of partial identification is tiny in the model), our results equally apply by simply setting the bounds of the identified set in the limit experiment as approaching zero.

{\scshape Outline.} The rest of the paper is organized as follows. In Section \ref{sec:illustrate}, we explain the key idea of our paper via a simple and  stylized example. Section \ref{sec:frame} introduces our limit treatment choice problem in a general setup and formal notion of asymptotic minimax regret optimality. Section \ref{sec:results} 
presents our results on solving the limit problem and constructing asymptotically optimal rules. 
We apply our results to three applications in Section \ref{sec:applications}. Section \ref{sec:conclude} concludes.  All proofs and additional technical results are reserved in the Appendix and Online Appendix.

\section{Illustrative Example}\label{sec:illustrate}

In this section, we illustrate the key idea of the paper via a stylized example that also demonstrates some general difficulties of asymptotics with partially identified payoff-relevant parameters. Suppose a decision maker (DM) must decide what fraction $a\in[0,1]$ of a target population to assign to a new policy. Denote by $\mu^{*}\in\mathbb{R}$ the true treatment effect on the target population. Applying action $a$ yields welfare $\mu^{*}a$ (status quo welfare is assumed to be known and normalized to zero). Thus, the infeasible optimal rule is $\mathbf{1}\left\{ \mu^{*}\geq0\right\} $. DM observes one realization of $Y\sim\mathcal{N}(\mu,\sigma^{2})$, where $\sigma>0$ is known and $\mu$ is an $\textit{identified}$ treatment effect; it is related to $\mu^{*}$ by 
\begin{equation*}
\mu^{*}\in I(\mu):=\left[\mu-k,\mu+k\right]~~\text{ for some known }k>0.\label{eq:id.illustrative.eg}
\end{equation*}
Here, the parameter space is
\begin{equation}\label{eq:stylized}
 \Theta:=\{(\mu,\mu^{*})\mid\mu\in\mathbb{R},\mu^{*}\in[\mu-k,\mu+k]\}.   
\end{equation}
A statistical decision rule $d$ maps $Y$ to the unit interval, and the goal is to find a decision that minimizes worst-case expected regret defined as
\begin{equation}
\sup_{\mu^{*}\in I(\mu),\mu\in\mathbb{R}}{\mu^{*}\mathbf{1}\{\mu^{*}\geq0\}-\mu^{*}\mathbb{E}_{\mu}d(Y)}.\label{eq:illustrative.regret}
\end{equation}
This minimax regret (MMR) problem is tractable due to (i) normal likelihood of $Y$ and (ii) convexity and centrosymmetry of $\Theta$ \citep{stoye2012minimax,yata2021}. The qualitative and quantitative features of the optimal decision depend on the tradeoff between $\sigma$ (indicating estimation uncertainty) and $k$ (indicating severity of partial identification): if $k\leq\sigma\sqrt{\pi/2}$, then $\mathbf{1}\{Y\geq0\}$ is MMR optimal, uniquely so if the inequality is strict; if $k>\sigma\sqrt{\pi/2}$, then infinitely many MMR optimal rules exist, one of which is $\Phi\left(Y/\sqrt{2k^{2}/\pi-\sigma^{2}}\right)$ and all of which in a large regularity class involve randomized treatment assignment.

These findings are instructive, but in practice, data are rarely exactly normal.  As noted in  \cite{hirano25}, rigorously formalizing the idea that the above Gaussian construction represents a limit experiment result is challenging. 
\begin{figure}[http]
 \centering
\begin{tikzpicture}[
    axis/.style={->, thick, black},
    grayline/.style={very thick, gray},
    redline/.style={very thick, red},
    myredline/.style={very thick, myred},
    blueline/.style={very thick, blue},
    greenline/.style={very thick, mygreen},
    tangent/.style={blue, ultra thick}
]
\def\xmin{-2.8}
\def\xmax{2.8}
\def\ymin{-2.8}
\def\ymax{2.8}
\def\k{1.5}
\def\sep{7.5}

% --------------------------------------------------------
% PANEL 1 (left)
% --------------------------------------------------------
\def\mured{1.0}
\def\mugreen{-0.7}

\begin{scope}[xshift=0cm]
  \begin{scope}
    \clip (\xmin,\ymin) rectangle (\xmax,\ymax);
    \fill[orange!30]
      (\xmin, {\xmin-\k})
      -- (\xmin, {\xmin+\k})
      -- (\xmax, {\xmax+\k})
      -- (\xmax, {\xmax-\k})
      -- cycle;
    \draw[grayline]  (\xmin, {\xmin+\k}) -- (\xmax, {\xmax+\k});
    \draw[grayline] (\xmin, {\xmin-\k}) -- (\xmax, {\xmax-\k});
  \end{scope}

  % Horizontal axis full width, vertical axis shortened at bottom
  \draw[axis] (\xmin,0) -- (\xmax,0) node[right] {$\mu$};
  \draw[axis] (0,-2.2) -- (0,\ymax) node[above] {$\mu^*$};

  \draw[redline] (\mured, {\mured-\k}) -- (\mured, {\mured+\k});
  \filldraw[red] (\mured, {\mured+\k}) circle (1.5pt);
  \filldraw[red] (\mured, 0)           circle (1.5pt);
  \filldraw[red] (\mured, {\mured-\k}) circle (1.5pt);
  \node[red, font=\large\bfseries] (mu0label) at (\mured+0.7, 0.7) {$\mu_0$};
    \draw[->, red, thick] (mu0label.south west) -- (\mured+0.07, 0.07);

  \filldraw[black] (0, \k) node[left=3pt] {$k$};
  \filldraw[black] (0,-\k) node[right=3pt] {$-k$};

  \draw[greenline] (\mugreen, {\mugreen-\k}) -- (\mugreen, {\mugreen+\k});
  \filldraw[mygreen] (\mugreen, {\mugreen+\k}) circle (1.5pt);
  \filldraw[mygreen] (\mugreen, 0)             circle (1.5pt);
  \filldraw[mygreen] (\mugreen, {\mugreen-\k}) circle (1.5pt);
  \node[mygreen, font=\large] (mulabel) at (\mugreen-1, 1.4)
    {$\mu_0 + \dfrac{h_1}{\sqrt{n}}$};
    \draw[->, mygreen, thick] (mulabel.south) -- (\mugreen-0.08, 0.07);

  \draw[->, very thick, black]
      (\mugreen+0.15, -2.4) -- (\mured-0.15, -2.4)
      node[midway, below, font=\normalsize] {as $n \to \infty$};
\end{scope}

% --------------------------------------------------------
% PANEL 2 (right)
% --------------------------------------------------------
\def\extn{0.9}
\def\off{0.35}

\begin{scope}[xshift=\sep cm]
  \begin{scope}
    \clip (\xmin,\ymin) rectangle (\xmax,\ymax);

    \fill[orange!30]
      (\xmin, {\xmin-\k})
      -- (\xmin, {\xmin+\k})
      -- (\xmax, {\xmax+\k})
      -- (\xmax, {\xmax-\k})
      -- cycle;

    \fill[blue!30]
      ({-\extn}, {-\extn-\off})
      -- ({\extn}, {\extn-\off})
      -- ({\extn}, {\extn+\off})
      -- ({-\extn}, {-\extn+\off})
      -- cycle;

    \begin{scope}
      \clip
        ({-\extn}, {-\extn-\off})
        -- ({\extn}, {\extn-\off})
        -- ({\extn}, {\extn+\off})
        -- ({-\extn}, {-\extn+\off})
        -- cycle;
      \draw[tangent] (\xmin, {\xmin+\off}) -- (\xmax, {\xmax+\off});
      \draw[tangent] (\xmin, {\xmin-\off}) -- (\xmax, {\xmax-\off});
    \end{scope}

    \draw[very thick, gray] (\xmin, {\xmin+\k}) -- (\xmax, {\xmax+\k});
    \draw[very thick, gray] (\xmin, {\xmin-\k}) -- (\xmax, {\xmax-\k});
  \end{scope}

  \draw[axis] (\xmin,0) -- (\xmax,0) node[right] {$\mu$};
  \draw[axis] (0,\ymin) -- (0,\ymax) node[above] {$\mu^*$};

  \draw[redline] (0,-\k) -- (0,\k);
  \filldraw[red] (0, \k) circle (1.5pt) node[left=3pt, black] {$k$};
  \filldraw[red] (0,-\k) circle (1.5pt) node[right=3pt, black] {$-k$};
  \filldraw[red] (0, 0)  circle (1.5pt);
  \node[red, font=\large\bfseries] at (0.45, 0.43) {$\mu_0$};

  \node[blue, font=\large] at (1.75, 1.25) {$\Theta_n$};
  \draw[blue, ->, thick] (1.4, 1.2) to[out=210, in=30] (1, 1);
\end{scope}

\end{tikzpicture}
\caption{Two different asymptotics in the stylized example. Left: we localize at any $\mu_0$. As a result, $I\left(\mu_0+h/\sqrt{n}\right)\rightarrow I(\mu_0)$ as $n\rightarrow\infty$. Right: we localize at  $\mu_0=0$ and consider  a local parameter space $\Theta_n$ implying a shrinking identified set.}
\label{fig:illustrate}
\end{figure}
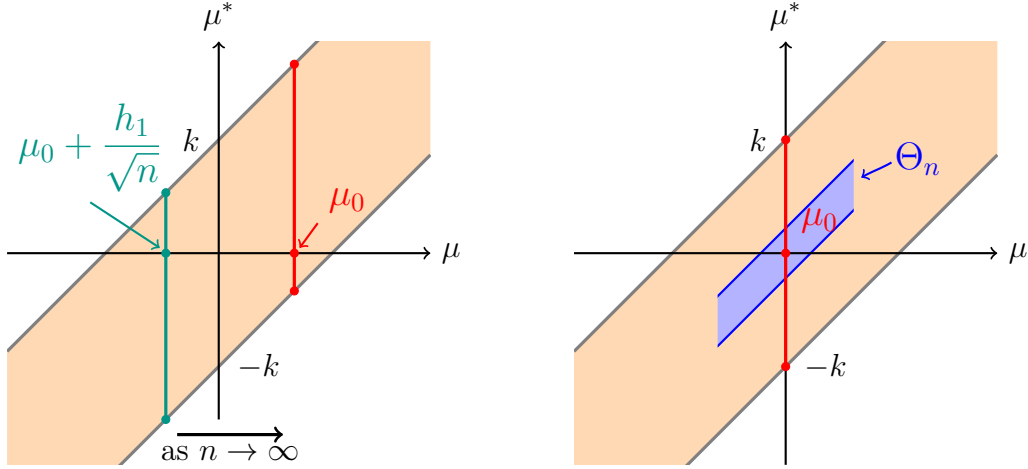

\subsection{The Challenge}\label{sec:eg.challenge}

Suppose now the DM observes a random sample $\{Y_i\}_{i=1}^{n}$, where $Y_i\sim\mathcal{N}(\mu,\sigma^{2})$ with both $\mu$ and $\sigma$ unknown. The parameter space becomes 
\[
\Theta=\{(\mu,\sigma^{2},\mu^{*})\mid\mu\in\mathbb{R},\sigma^{2}>0,\mu^{*}\in I(\mu)\}.
\]
As $n$ becomes large, the location of $\mu$ and $\sigma^{2}$ becomes easier to learn but that of $\mu^*$ may still remain ambiguous. Suppose that, in the spirit of \citet{HiranoPorter2009,HiranoPorter2020} and work that they build on, we recenter the reduced-form parameter around some $\mu=\mu_{0}$ and $\sigma^{2}=\sigma_{0}^{2}$:\footnote{The typical notion of localization, including that in  \cite{HiranoPorter2009,HiranoPorter2020},  focuses on point-identified problems and does not apply to partially identified problems.}
\begin{align*}
\mu & =\mu_{0}+\frac{h_{1}}{\sqrt{n}},~~\sigma^{2}=\sigma_{0}^{2}+\frac{h_{2}}{\sqrt{n}},
\end{align*}
where $h=(h_{1},h_{2})^{\top}\in\mathbb{R}^2$.
Then, the identified set for $\mu^{*}$ is $I\left(\mu_{0}+h_1/\sqrt{n}\right)$. As $n\rightarrow\infty$, the location  $h$ will be signaled by an \emph{asymptotically sufficient statistic} $\Delta\sim\mathcal{N}(h,\mathbf{I}_0^{-1})$ with Fisher Information $\mathbf{I}_0$, while it is clear that (see also the left panel of Figure \ref{fig:illustrate})
\begin{align*}
\mu_{0}+\frac{h_{1}}{\sqrt{n}} & \rightarrow\mu_{0},
~~I\left(\mu_{0}+\frac{h_{1}}{\sqrt{n}}\right)\rightarrow I\left(\mu_{0}\right).
\end{align*}
In other words, in the limit, the DM would observe $\Delta\sim\mathcal{N}(h,\mathbf{I}_{0}^{-1})$
and \emph{know} $\mu^{*}\in[\mu_{0}-k,\mu_{0}+k]$. In this limit, a rule  $d_\infty$  maps $\Delta$ to $[0,1]$ and the MMR problem becomes
\begin{equation}\label{eq:default.asymptotics.limit.game}
\min_{d_{\infty}}\sup_{\mu^*\in [\mu_{0}-k,\mu_{0}+k],h\in\mathbb{R}^{2}}\mu^{*}\left[\mathbf{1}\left\{ \mu^{*}\geq0\right\} -\mathbb{E}d_\infty(\Delta)\right].
\end{equation}
The signal $\Delta$ has become pure noise, revealing no information regarding $\mu_0$ or $\mu^*$. It follows that \eqref{eq:default.asymptotics.limit.game} is solved by the no-data rule  \citep{manski2007identification} 
\[
\left[\frac{\mu_{0}+k}{2k}\right]_{0}^{1}:=\begin{cases}
0, & \mu_{0}<-k\\
\tfrac{\mu_{0}+k}{2k}, & -k\leq\mu_{0}\leq k\\
1, & \mu_{0}>k
\end{cases}.
\]
%where $[x]_0^{1}$ denotes clamping $x$ to the unit interval $[0,1]$. 
This finding reflects partial identification overwhelming sampling noise in the limit. Indeed, under these asymptotics, sampling uncertainty vanishes at the standard $1/\sqrt{n}$ rate whereas the degree of partial identification is constant even in the limit. As a result,  there is no meaningful trade-off between estimation uncertainty and the severity of partial identification, as opposed to what we saw earlier in the ``finite-sample'' version of the problem. Moreover, these asymptotics cannot (sufficiently) inform the selection of good rules because too many rules will be asymptotically MMR optimal in the sense of worst-case regret converging to the same limit. For example, this applies to $\left[(\hat{\mu}+k)/(2k)\right]_0^1$ for any consistent estimator $\hat{\mu}$ of $\mu$.

\subsection{Introducing Local Asymptotics}\label{sec:eg.key.idea}
Our asymptotic approach embodies two major innovations; see also the right panel of Figure \ref{fig:illustrate} for a visualization. First, we are more explicit about which parameter values to localize about. In point-identified settings, one often picks a point at which the payoff-relevant parameter equals zero, representing a \emph{hardest case} where it is difficult to determine the best
treatment even with large sample sizes \citep[][p.1687]{HiranoPorter2009}.  In our example,  that would correspond to $\mu^*=0$. However, as $\mu^*$ is partially identified, any point of $(\mu,\sigma^2)$ such that $I(\mu)$ contains $0$ is consistent with $\mu^*=0$, and it is not immediately clear which point in the space of   $(\mu,\sigma^2)$ we should pick. We argue that, even in the partially identified setting, there is an analogous notion of the  \emph{hardest case}  that  motivates our choice of localization for the reduced-form parameter. In this stylized example, that corresponds to $\mu=0$ (and the point of localization for $\sigma^2$ can be arbitrary). Intuitively, if the sample size is sufficiently large to inform the location of $\mu$, then $\mu=0$ is the most challenging case as it involves the greatest   ambiguity regarding  $\mu^*$. See Section \ref{sec:local.point} for the definition of our ``most difficult case'' in a general setting. 

Second, to ensure a nontrivial trade-off between the severity of partial identification and estimation difficulty, we let the former vanish at an appropriate rate \citep{armstrong2021sensitivity}. That is, we look at drifting sequences that imply ``near point identification.''  Specifically,  we consider a local parameter space with a shrinking identified set: 
\begin{align*}\Theta_{n} :=\left\{ (\mu,\sigma^{2},\mu^{*})\mid \mu=0+\frac{h_{1}}{\sqrt{n}},\sigma^{2}=\sigma_{0}^{2}+\frac{h_{2}}{\sqrt{n}},\mu^{*}\in I_n\left(\frac{h_{1}}{\sqrt{n}}\right)\right\},
\end{align*}
where 
${I_{n}}\left(\frac{h_{1}}{\sqrt{n}}\right)  =\left[\frac{h_{1}}{\sqrt{n}}-k_{n},\frac{h_{1}}{\sqrt{n}}+{k_n}\right]$
for some $k_{n}$ such that $\sqrt{n}k_{n}=C>0$.
Now, consider the scaled  risk function of a  decision rule $d$ based on data $\{Y_i\}_{i=1}^{n}$, which can be written as 
\[
\sqrt{n}\mu^{*}\left(\mathbf{1}\left\{ \mu^{*}\geq0\right\} -\mathbb{E}_{(\mu,\sigma^{2})}[d((Y_i)_{i=1}^n)]\right).
\]
It follows that
\begin{align*}\sqrt{n}\mu^{*}  \in\sqrt{n}{I_{n}}\left(\frac{h_{1}}{\sqrt{n}}\right)=\sqrt{n}\left[\frac{h_{1}}{\sqrt{n}}-k_{n},\frac{h_{1}}{\sqrt{n}}+k_{n}\right] =\left[h_{1}-C,h_{1}+C\right].
\end{align*}
Furthermore, standard asymptotic arguments (e.g.,  Proposition 3.1 in \citealt{HiranoPorter2009}) yield that for any converging rule $d$, there is a limiting rule $d_{\infty}$ such that
\[
\mathbb{E}_{(\mu,\sigma^{2})}[d((Y_i)_{i=1}^n)]\rightarrow\mathbb{E}_{h}[d_{\infty}(\Delta)],
\]
where $\Delta\sim\mathcal{N}(h,\mathbf{I}_{0}^{-1})$. Thus, we have a stable limit game in which the researcher observes $\Delta\sim\mathcal{N}(h,\mathbf{I}_{0}^{-1})$, and the partially identified welfare contrast $\mu^{*}_{\infty}$ lies in a limit identified set $\left[h_{1}-C,h_{1}+C\right]$. The limit MMR problem becomes 
\begin{equation}\label{eq:intro.limit.mmr}
\min_{d_{\infty}:\mathbb{R}^{2}\rightarrow[0,1]}\sup_{\mu_{\infty}^{*}\in[h_{1}-C,h_{1}+C],h\in\mathbb{R}^{2}}\mu_{\infty}^{*}\left[\mathbf{1}\left\{ \mu_{\infty}^{*}\geq0\right\} -\mathbb{E}_{h}[d_{\infty}(\Delta)]\right],   
\end{equation}
recovering the stylized example; compare \eqref{eq:stylized}-\eqref{eq:illustrative.regret}. Letting  $w=(1,0)^{\top}$ and $\Sigma_{0}=(w)^{\top}\mathbf{I}_{0}^{-1}w$, we can apply existing results and find the solution of \eqref{eq:intro.limit.mmr} as follows:
If $C\leq\sqrt{\Sigma_{0}\pi/2}$, then $\mathbf{1}\{w^{\top}\Delta\geq0\}$ is MMR optimal; if $C>\sqrt{\Sigma_{0}\pi/2}$, one optimal rule is
$\Phi\left(w^{\top}\Delta\left(2C^{2}/\pi-\Sigma_{0}\right)^{-1/2}\right)$.
With this limit optimal rule, we may further find a matching feasible rule via the natural ``plug-in'' principle, i.e., replacing (i) $w^{\top}\Delta$ with $\sqrt{n}\hat{\mu}$, where $\hat{\mu}$ is a best regular estimator of $\mu$ (e.g., the MLE) and (ii) $\Sigma_{0}$ with a consistent estimator $\hat{\Sigma}$.  This is exactly what a researcher would use in practice based on the ``finite-sample'' results above if we view  $\hat{\mu}$ as $Y$ in \eqref{eq:illustrative.regret} and calibrate $C$  with  $\sqrt{n}k$. 

Next, we formalize the above discussions in a general framework. We show that our idea of applying ``near point identification'' at a suitably defined hardest case allows one to find asymptotically MMR optimal rules (in a sense we also precisely define) for a large class of treatment decision problems with partially identified parameters, even beyond convexity and centrosymmetry.

\section{The Limiting Decision Problem} \label{sec:frame}

\subsection{Setup}

Suppose a DM's payoff when taking action $a\in[0,1]$ is $aW_{1}+(1-a)W_{0}$
for some scalars $W_{1}$ and $W_{0}$. Denote by $U^{*}:=W_{1}-W_{0}\in\mathbb{R}$
the \emph{welfare contrast} between actions $a=1$ and $a=0$. 
If $U^{*}$ were known, the optimal action would be $\mathbf{1}\left\{ U^{*}\geq0\right\}$. The DM observes a random sample
$Y^{n}:=\left\{ Y_{i}\right\} _{i=1}^{n}\in\mathbf{Y}^{n}$ of size
$n$, where $Y_{i}\in\mathbf{Y}\subseteq\mathbb{R}^{d_{Y}}$, to learn
about $U^{*}$. We assume that $Y_{i}$ follows a parametric
model with a distribution $P_{\gamma}$ indexed by $\gamma\in M\subseteq\mathbb{R}^{d_{m}}$.
We require $\gamma$
to be point-identified in the usual sense, i.e., if $\gamma,\gamma^{\prime}\in M$ are such
that $\gamma\neq\gamma^{\prime}$, we have $P_{\gamma}\neq P_{\gamma^{\prime}}$.
We therefore treat the values of $\gamma\in M$ as the point-identified
reduced-form parameter.

As in \cite{christensen2022optimal}, our primary focus  is when $U^{*}$---the payoff-relevant parameter---is only partially identified, but the DM can use $\gamma$ to deduce restrictions on $U^*$. That is, if $\gamma$ were known, the most the DM could infer is that $U^{*}\in\mathcal{I}(\gamma)\subseteq\mathbb{R}$,
where $\mathcal{I}(\cdotp)$ is a set-valued mapping from $M$ to
$\mathbb{R}$ that may contain both positive and negative values. We refer to  $\mathcal{I}(\gamma)$ as the identified
set of $U^{*}$ given $\gamma\in M$. 

\begin{assumption}\label{asm:1} The identified set $\mathcal{I}(\gamma)$
depends on $\gamma$ only via a known function $\mu(\cdot):\mathbb{R}^{d_{m}}\rightarrow\mathbb{R}^{d_{\mu}}$,
with $d_{\mu}\leq d_{m}$. \end{assumption}

We may think of $\mu(\cdot)$ as a \emph{decision-relevant} reduced-form
parameter and view the rest of $\gamma$ as some nuisance parameter
not directly relevant for the identification of $U^{*}$. For instance,
in the example of Section \ref{sec:eg.challenge}, $\gamma$ refers
to the mean and variance of a normal distribution, while the identified
set of $U^{*}$ only depends on the mean.

Under Assumption \ref{asm:1}, we can write $\mathcal{I}(\gamma):=I(\mu(\gamma))$,
where 
\[
I(t):=\left\{ u\in\mathbb{R}\mid u\in\mathcal{I}(\gamma),\mu(\gamma)=t,\gamma\in M\right\} 
\]
is the identified set of $U^{*}$ given a fixed value of $\mu(\cdotp)$.
Let 
\[
M_{\mu}:=\left\{ t \in\mathbb{R}^{d_{\mu}}\mid\mu(\gamma)=t,\gamma\in M\right\} 
\]
collect all the values that the function $\mu(\cdotp)$ can take as
$\gamma$ ranges over  $M$. For any $t\in M_{\mu}$, let 
\[
\overline{I}(t):=\sup I(t);\quad\underline{I}(t):=\inf I(t)
\]
be the upper and lower bounds of $I(t)\subseteq\mathbb{R}$. Abusing  notation, sometimes we write $I(\mu)$, and pretend that
$\mu$ denotes a particular value of the function $\mu(\cdotp)$,
instead of the whole function. In the problems that we are interested
in, we assume there must be at least one point $t^{*}\in M_{\mu}$
for which $0\in I(t^{*})$ and $\underline{I}(t^{*})<0<\overline{I}(t^{*})$.
This means that at values of the decision-relevant reduced-form parameter
that equal $t^{*}$, the welfare contrast could be strictly positive
or strictly negative.

After observing data $Y^{n}$, the decision maker chooses a statistical
decision rule $d_{n}:\mathbf{Y}^{n}\rightarrow[0,1]$, interpreted
as the probability of treatment assignment or the fraction of the
target population to be treated. To  organize notation, it will be convenient to define
$\theta:=\left(\gamma^{\top},U^{*}\right)^{\top}$ and 
\[
\Theta:=\left\{ \left(a^{\top},b\right)^{\top}\in\mathbb{R}^{d_{m}+1}\mid a\in M,b\in\mathcal{I}(a)\right\}. 
\]
Note that the set $\Theta$ collects the values that $\gamma$ and $U^{*}$ can take jointly. We endow $\Theta$ with the standard subspace topology of $\mathbb{R}^{d_m+1}$. Throughout the rest of the paper we require that its topological interior, denoted as $\operatorname{int}(\Theta)$, is nonempty.   

In order to connect with previous work in the literature---in particular, with the work of \cite{yata2021}---for each $\theta\in\Theta$, denote by $m(\theta)$ the mapping that
selects the first $d_{m}$ elements of $\theta$ (i.e., the point-identified $\gamma$), and by $U(\theta)$ the mapping
that selects the last element of $\theta$ (i.e., the partially identified payoff-relevant parameter $U^{*}$). By construction, both $m(\cdotp)$
and $U(\cdotp)$ are linear. By definition,  $m(\cdot)$ is clearly not injective: $m(\theta) = m(\theta')$ does not imply $\theta=\theta'$.\footnote{Even though the framework considered by \cite{yata2021} appears to be significantly more general than ours (because he allows for the possibility of an underlying infinite-dimensional parameter $\theta$), we argue that the treatment choice problems that he considers can be recast using our framework. To see this,
given some underlying $\tilde{\theta}\in\tilde{\Theta}$ that may
be infinitely dimensional, let $\tilde{m}(\cdotp):\tilde{\Theta}\rightarrow\mathbb{R}^{d_{\tilde{m}}}$
map $\theta$ to a finite-dimensional reduced-form parameter and $\tilde{U}(\cdotp):\tilde{\Theta}\rightarrow\mathbb{R}$
map $\theta$ to a partially identified payoff-relevant parameter. Since the
risk function depends on $\theta$ only via $\tilde{m}(\cdotp)$
and $\tilde{U}(\cdotp)$, one can always reparametrize so that
$\theta:=(\tilde{m}(\tilde{\theta})^{\top},\tilde{U}(\tilde{\theta}))^{\top}\in\Theta\subseteq\mathbb{R}^{1+d_{\tilde{m}}}$, where $\Theta:=\{(\tilde{m}(\tilde{\theta})^{\top},\tilde{U}(\tilde{\theta}))\in\mathbb{R}^{1+d_{\tilde{m}}}\mid\tilde{\theta}\in\tilde{\Theta}\}$
is a set in $\mathbb{R}^{1+d_{\tilde{m}}}$. 
}  

We evaluate the performance of a decision rule $d_{n}$ via expected
regret:
\begin{align*}
R(d_{n},\theta) =U(\theta)\left(\mathbf{1}\left\{ U(\theta)\geq0\right\} -\mathbb{E}_{m(\theta)}\left[d_{n}\left(Y^{n}\right)\right]\right),
\end{align*}
where $\mathbb{E}_{m(\theta)}\left[\cdotp\right]$ denotes expectation with respect
to $Y^{n}$, whose distribution only depends on $m(\theta)$. We omit the subscript whenever there is no risk of confusion. Note that we also omit the explicit dependence of the risk on the sample size, for the sake of notational simplicity. A rule is finite-sample MMR optimal if it solves 
\begin{equation}
\min_{d_{n}}\sup_{\theta\in\Theta}R(d_{n},\theta).\label{eq:finite.sample.mmr}
\end{equation}

\textsc{Local Asymptotics.} To provide a simpler, large-sample representation of the
finite-sample MMR problem \eqref{eq:finite.sample.mmr}, we fix a parameter $\theta_0:=(\gamma_0^{\top},0)^{\top}\in \operatorname{int}(\Theta)$, such that $\gamma_0\in \operatorname{int}(M)$, $0\in \mathcal{I}(\gamma_0)$ and $\operatorname{int}(\mathcal{I}(\gamma_0))\neq \emptyset$. Then, we  consider a sequence of parameter values of the form:
\begin{equation}\label{eq:local.parameter.theta}
\theta_n=\theta_0+\frac{\theta_h}{\sqrt{n}},   
\end{equation}
where $\theta_h:=(h^{\top},\mu^*)^{\top} \in \mathbb{R}^{d_m+1}$. Since $m(\cdot)$ and $U(\cdot)$ are linear,  we also have
\begin{align}
m(\theta_{n})=\gamma_{0}+\frac{h}{\sqrt{n}},\quad U(\theta_{n})=\frac{\mu^{*}}{\sqrt{n}}.\label{eq:local.parameter.1}
\end{align}
In words, we recentered the reduced-form parameter at $\gamma_{0}$
and the payoff-relevant but partially identified parameter at $0$.

\subsection{Choice of Localization for the Point-Identified Parameter}\label{sec:local.point}
In point-identified settings, a standard choice for the localization parameter $\gamma_0$ is any point such that the corresponding payoff-relevant parameter equals zero, which corresponds to a \emph{hardest case} where it is difficult to determine the best
treatment even with large sample sizes \citep[][p.1687]{HiranoPorter2009}.  In our setup, the payoff-relevant parameter $U(\theta)$ is partially identified. This means that the choice of $\gamma_0$ is more delicate,  as any  values of $\gamma_0$ such that $0\in I(\mu(\gamma_0))$  would be consistent with a situation in which $U(\theta)=0$, and hence it might be difficult to determine the best treatment even in large samples. 

Let $\mu_0:=\mu(\gamma_0)$. We propose a notion of \emph{hardest case} that considers the population level MMR problem 
\begin{equation}\label{eq:population.mmr}
\min_{a\in[0,1]}\sup_{u\in I(\mu_0)}u\left(\mathbf{1}\left\{  u\geq0\right\} -a\right)
\end{equation}
corresponding to oracle knowledge of  the decision-relevant reduced-form parameter. 
Denote by $\overline{R}(\mu_0)$ the value of \eqref{eq:population.mmr}. Note $\overline{I}(\mu_0)=\underline{I}(\mu_0)$ corresponds to a point-identified problem, for which $\overline{R}(\mu_0)=0$. Otherwise, a MMR optimal  action for \eqref{eq:population.mmr} is \citep{manski2007identification}
\begin{equation}\label{eq:manski.rule.general}
\left[\frac{\overline{I}(\mu_0)}{\overline{I}(\mu_0)-\underline{I}(\mu_0)}\right]_{0}^{1},   
\end{equation}
achieving a value of
\[
\overline{R}(\mu_0)=\max\left\{ \frac{-\overline{I}(\mu_0)\underline{I}(\mu_0)}{\overline{I}(\mu_0)-\underline{I}(\mu_0)},0\right\}.
\]
We already assumed the existence of a decision-relevant reduced-form parameter, $\mu_0$, such that
$0\in I(\mu_0)$ and $\overline{I}(\mu_0) > 0 > \underline{I}(\mu_0)$. We strengthen this requirement by assuming that there is
a point $\mu_0$ in the interior of $M_{\mu}$ that satisfies this property, and that yields the highest possible value for the
population level MMR problem.
\begin{assumption}\label{asm:2}
There exists $\theta_0\in \operatorname{int}(\Theta)$  such that $m(\theta_0)=\gamma_0\in \operatorname{int}(M)$, $\mu_{0}=\mu(\gamma_0)\in\operatorname{int}\left(M_{\mu}\right)$, where $\overline{I}(\mu_0)>0>\underline{I}(\mu_0)$, and $\overline{R}(\mu_{0})\geq\overline{R}(\mu^{\prime})$
for all $\mu^{\prime}\in M_{\mu}$. \end{assumption} 
The value of the problem \eqref{eq:population.mmr} is maximized and strictly positive at $\mu_0$ satisfying Assumption \ref{asm:2}, which we interpret as $\mu_0$ being most difficult from the DM's perspective. We refer to such $\mu_0$ as a \emph{global
least favorable (decision-relevant reduced-form) point}. This means that $m(\theta_n)$ in \eqref{eq:local.parameter.1} is centered at any $\gamma_{0}\in \operatorname{int(M)}$ 
such that  $\mu(\gamma_{0})=\mu_{0}$. See Figure \ref{fig:lfp} for an illustrative example  of a global least favorable point.

We note that different values of the reduced-form parameter $\gamma_0$ can be associated with the
same decision-relevant reduced-form parameter $\mu_0$. This is  allowed in our theory,  as long as we have efficient estimators for $\gamma_0$ (as well as consistent estimators of the asymptotic variance). See analogous treatments in  \cite{HiranoPorter2009} with asymmetric welfare regret and \cite{kitagawa2026treatment} with nonlinear regret.

\begin{figure}
\centering
\begin{tikzpicture}[
    axis/.style={->, thick, black},
    grayline/.style={very thick, gray},
    redline/.style={very thick, red},
    myredline/.style={very thick, myred},
    blueline/.style={very thick, blue},
    greenline/.style={very thick, mygreen},
    tangent/.style={gray, ultra thick}
]
\def\xmin{-2.8}
\def\xmax{2.8}
\def\ymin{-2.8}
\def\ymax{2.8}
\def\k{1.5}

\begin{scope}
  \clip (\xmin,\ymin) rectangle (\xmax,\ymax);
  \fill[orange!30]
    (\xmin, {\xmin-\k})
    -- (\xmin, {\xmin+\k})
    -- (\xmax, {\xmax+\k})
    -- (\xmax, {\xmax-\k})
    -- cycle;
  \draw[grayline]  (\xmin, {\xmin+\k}) -- (\xmax, {\xmax+\k});
  \draw[grayline] (\xmin, {\xmin-\k}) -- (\xmax, {\xmax-\k});
\end{scope}

% Axes
\draw[axis] (\xmin,0) -- (\xmax,0) node[right] {$\mu$};
\draw[axis] (0,\ymin) -- (0,\ymax) node[above] {$\mu^*$};

% Red vertical segment from -k to k on u* axis
\draw[redline] (0,-\k) -- (0,\k);

% Red dots
\filldraw[orange] (-\k, 0)  circle (1.5pt);
\filldraw[blue] (\k,0)  circle (1.5pt);  % right of axis, lower-right quadrant
\filldraw[red] (0, 0)   circle (1.5pt);

% u0 label at origin (red, slightly offset)
\node[red, font=\large\bfseries] at (0.4, 0.25) {$\mu_0$};

% k on u axis (right side)
\node[below] at (\k, 0) {$k$};
% -k on u axis: upper-left quadrant, just above the axis
\node[above] at (-\k - 0.2, 0) {$-k$};

\end{tikzpicture}
\caption{A global least favorable point in the example of Section \ref{sec:eg.challenge}. On $\mu$-axis, every point to the left of the orange dot and to the right of the blue dot leads to $\overline{R}(\mu)=0$. Between these two points, $\overline{R}(\mu)$ is strictly positive and maximized at the red dot with a value of $k/2$.}
\label{fig:lfp}
\end{figure}
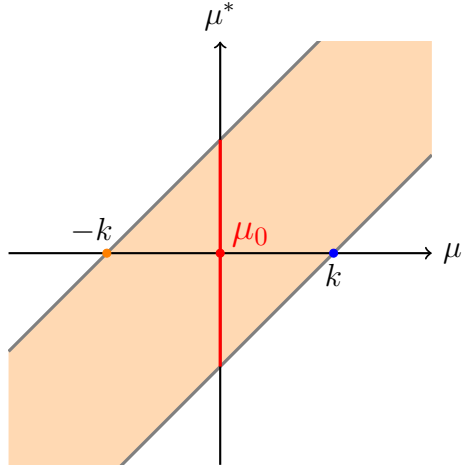

\subsection{Choice of Local Identified Set}

Let $\gamma_0$ and $\mu_0$ be defined as in Assumption \ref{asm:2}. We impose the following smoothness and regularity condition. 
\begin{assumption}\label{asm:diff}
\begin{itemize}
\item[(i)] $\mu(\cdotp)$ is (fully) differentiable in a local neighborhood of $\gamma_{0}$.
\item[(ii)] In a local neighborhood of $\mu_{0}$, $\underline{I}(\cdotp)$ and $\overline{I}(\cdotp)$ are attained and (Hadamard) directionally
differentiable for all directions $\mathbf{v}\in\mathbb{R}^{d_{\mu}}$,  and $I(\cdot)$ is connected. 
\end{itemize}
\end{assumption} 
Under Assumption \ref{asm:diff}, let 
\[
G_{0}:=\underset{\mathbb{R}^{d_{\mu}}\times\mathbb{R}^{d_{m}}}{\underbrace{\mu^{(1)}(\gamma_{0})}}
\]
be the gradient of $\mu(\cdotp)$ at $\gamma_{0}$.  Denote by $\overline{I}^{(1)}(\mu;\mathbf{v})$ the directional derivative
of $\overline{I}(\cdotp)$ at $\mu$ with respect to vector $\mathbf{v}\in\mathbb{R}^{d_{\mu}}$,
i.e., 
\[
\overline{I}^{(1)}(\mu;\mathbf{v}):=\lim_{\lambda\downarrow0,\mathbf{v'\rightarrow\mathbf{v}}}\frac{\overline{I}(\mu+\lambda\mathbf{v}')-\overline{I}(\mu)}{\lambda}.
\] 
If $\overline{I}(\cdotp)$
is (fully) differentiable at $\mu$, we have 
\[
\overline{I}^{(1)}(\mu;\mathbf{v})=\left(\overline{I}^{(1)}(\mu)\right)^{\top}\mathbf{v},
\]
where $\overline{I}^{(1)}(\mu)$ is the gradient of $\overline{I}(\cdotp)$
at $\mu$. The directional derivative and gradient of $\underline{I}(\cdotp)$
at $\mu$, written as $\underline{I}^{(1)}(\mu;\mathbf{v})$ and $\underline{I}^{(1)}(\mu)$
respectively, are defined analogously. 

For each $h\in\mathbb{R}^{d_m}$, the point-identified parameter follows a sequence 
$m(\theta_n)=\gamma_0+\frac{h}{\sqrt{n}}$. As $n\rightarrow\infty$, 
%$m(\theta_n)$ will be more and more concentrated around $\gamma_{0}$, and 
the identified set of $U(\theta_n)$ will be 
\begin{align}
I\left(\mu\left(\gamma_{0}+\frac{h}{\sqrt{n}}\right)\right) & \approx\left[\underline{I}(\mu_{0})+\frac{\underline{I}^{(1)}(\mu_{0};G_{0}h)}{\sqrt{n}},\overline{I}(\mu_{0})+\frac{\overline{I}^{(1)}(\mu_{0};G_{0}h)}{\sqrt{n}}\right]\nonumber\\
&\rightarrow\left[\underline{I}(\mu_{0}),\overline{I}(\mu_{0})\right],\text{ as }n\rightarrow\infty.\label{eq:fixed.id}
\end{align}
The limit identified set \eqref{eq:fixed.id} reveals two important  observations. First, in the limit $U(\theta_n)$ converges to zero (implying point identification); however, in the limit, \eqref{eq:fixed.id} is a set, which implies partial identification. We think this suggests an inconsistency in the analysis which ought to be reconciled. Second, \eqref{eq:fixed.id} does not depend on the local parameter $h$---in this limit,  the data became pure noise and the issue of partial identification overwhelms sampling uncertainty.

Given these two observations and also motivated by the local behavior of $I(\mu(\gamma_0+\frac{h}{\sqrt{n}}))$
with large $n$,  we will construct a particular sequence of local,  shrinking identified set in the form of \eqref{eq:local.id.U}-\eqref{eq:local.id.U.2} below, leading to a scaled limit identified set \eqref{eq:limit.ID}. We stress that our construction is not the only way in which one could approximate the decision problem of interest. As we explain in Appendix \ref{sec:global.ID}, there are other reasonable ways of constructing a shrinking identified set that may lead to a different geometry of the limit identified set.\footnote{For example, if the  identified set can be defined via a number of moment conditions, the analysis  in \cite{armstrong2021sensitivity} would yield a limiting identified set different from \eqref{eq:limit.ID}, although  the two approaches may coincide in simple cases, e.g., the stylized example in Section \ref{sec:illustrate}.}

Take  $\underline{I}(\mu_0)$, $\overline{I}(\mu_0)$, $\underline{I}^{(1)}(\mu_0;\cdot)$ and  $\overline{I}^{(1)}(\mu_0;\cdot)$ as given from the original identified set, we construct a new, local identified set  as follows: 
\begin{equation}\label{eq:local.id.U}
U(\theta_n)\in\left[\underline{I}_{{n}}(\mu_{0})+\frac{\underline{I}^{(1)}(\mu_{0};G_{0}h)}{\sqrt{n}},\overline{I}_{{n}}(\mu_{0})+\frac{\overline{I}^{(1)}(\mu_{0};G_{0}h)}{\sqrt{n}}\right],
\end{equation}
where $\underline{I}_{{n}}(\mu_{0})\rightarrow0$ and $\overline{I}_{{n}}(\mu_{0})\rightarrow0$ as $n\rightarrow\infty$ and
\begin{equation}\label{eq:local.id.U.2}
\frac{\underline{I}_{{n}}(\mu_{0})}{\overline{I}_{{n}}(\mu_{0})}=\frac{\underline{I}(\mu_0)}{\overline{I}(\mu_0)},\quad\sqrt{n}\underline{I}_{{n}}(\mu_{0})=\underline{C}{},\quad\sqrt{n}\overline{I}_{{n}}(\mu_{0})=\overline{C},   
\end{equation}
for some $\underline{C}<0,\overline{C}>0$. In particular,  the relative magnitude of  $\underline{I}_{{n}}(\mu_{0})$  versus $\overline{I}_{{n}}(\mu_{0})$  respects what we observe in the original identified set.  It follows that
\begin{align*}
\mu^*=\sqrt{n}U(\theta_n) & \in\left[\sqrt{n}\underline{I}_{{n}}(\mu_{0})+\underline{I}^{(1)}(\mu_{0};G_{0}h),\sqrt{n}\overline{I}_{{n}}(\mu_{0})+\overline{I}^{(1)}(\mu_{0};G_{0}h)\right]\\
 & =\left[\underline{C}+\underline{I}^{(1)}(\mu_{0};G_{0}h),\overline{C}+\overline{I}^{(1)}(\mu_{0};G_{0}h)\right].
\end{align*}
Therefore, we define the identified set
for $\mu^*$ in \eqref{eq:local.parameter.1} as
\begin{equation}
I_{\infty}(h):=[\underline{I}_{\infty}(h),\overline{I}_{\infty}(h)],\label{eq:limit.ID}
\end{equation}
where $\underline{I}_{\infty}(h):=\underline{C}+\underline{I}^{(1)}(\mu_{0};G_{0}h)$ and 
$\overline{I}_{\infty}(h):=\overline{C}+\overline{I}^{(1)}(\mu_{0};G_{0}h)$. We refer to \eqref{eq:limit.ID} as a \emph{limiting (local) identified set}, since in the limit it only reflects the local information of the original identified set $I(\cdot)$ around $\mu_0$. This means that our limiting identified set enjoys a simpler geometry than the original identified set and allows for slightly more tractability. As \eqref{eq:limit.ID} should be nonempty, we define
\begin{equation}\label{eq:M.h}
M_h:=\{h\in\mathbb{R}^{d_m}\mid\overline{I}_{\infty}(h)\geq\underline{I}_{\infty}(h)\}.
\end{equation}

It is important to stress that,  by working with such ``near-point-identification'' asymptotics, we do not  mean that the identified set  is literally shrinking as sample size increases. As noted in \cite{armstrong2021sensitivity},  the usefulness of any asymptotic device ``\emph{should be judged by whether it yields accurate approximations to the finite-sample behavior}''. In this regard, we think our approach offers a unifying framework for researchers to judge a suitable course of action in treatment choice problems with partial identification.  

\begin{itemize}
\item[(i)]  Our approach is most helpful if, given finite data and the  decision problem at hand, the researcher  believes that both sampling uncertainty and the severity of partial identification are  important considerations. In practice, they simply apply our theory below by replacing  $\overline{C}$ with $\sqrt{n}\overline{I}(\mu_0)$ and analogously replacing  $\underline{C}$ with $\sqrt{n}\underline{I}(\mu_0)$. 

\item[(ii)] If the data are so abundant that the researcher believes sampling uncertainty is rather tiny compared to the degree of partial identification, then they are free to set $\overline{C}=-\underline{C}=\infty$.  In this case, we would simply recommend to apply \eqref{eq:manski.rule.general} with consistent estimators of the corresponding reduced-form parameters. 

\item[(iii)] Finally, if the data is so scarce that the researcher believes the problem of lack of identification is not as important at all,  they may  simply apply our approach and evaluate the optimal rule by letting    $\overline{C}\downarrow0$ and $\underline{C}\uparrow0$. 
\end{itemize}

\subsection{Local Asymptotic Optimality}

To complete characterizing the limit  game,  we impose a standard differentiability
in quadratic mean (DQM) condition for the statistical model  $P_{m(\theta)}$.
\begin{assumption}\label{asm:DQM} For an open set $\Gamma\subseteq M$ 
that contains $\gamma_{0}$, the model $\left\{ P_{m(\theta)}:m(\theta)\in\Gamma\right\} $
satisfies the DQM assumption when $m(\theta)=\gamma_{0}$, i.e., there
exists a function $s:\mathbf{Y}\rightarrow\mathbb{R}^{d_{m}}$ such
that 
\begin{align*}
  \int\left[dP_{\gamma_{0}+h}^{1/2}(y)-dP_{\gamma_{0}}^{1/2}(y)-\frac{1}{2}h^{\top}s(y)dP_{\gamma_{0}}^{1/2}(y)\right]^{2}= o\left(\left\Vert h\right\Vert ^{2}\right),\text{ as }h\rightarrow0.
\end{align*}
Moreover, the Fisher information $\mathbf{I}_{0}=\mathbb{E}_{\gamma_{0}}\left[ss^{\prime}\right]$
is nonsingular. \end{assumption} \begin{lemma}\label{lem:DQM} Suppose
Assumptions \ref{asm:1}-\ref{asm:DQM} hold. For a sequence of 
rules $d_{n}$, if $\mathbb{E}_{\gamma_{0}+\frac{h}{\sqrt{n}}}[d_{n}(Y^{n})]$
converges as $n\rightarrow\infty$ for every $h$, then there exists
some $d_\infty:\mathbb{R}^{d_m}\rightarrow[0,1]$ such that for every $h$,
\begin{equation}\label{eq:matching.rule}
\mathbb{E}_{\gamma_{0}+\frac{h}{\sqrt{n}}}[d_{n}(Y^{n})]\rightarrow\mathbb{E}[d_\infty(\Delta)]   
\end{equation}
as $n\rightarrow\infty$, where $\Delta\sim\mathcal{N}(h,\mathbf{I}_{0}^{-1})$
follows a Gaussian distribution with mean $h$ and a covariance matrix
$\mathbf{I}_{0}^{-1}$. \end{lemma} 

Lemma \ref{lem:DQM} simply restates Proposition 3.1 of \cite{HiranoPorter2009} in the context of our statistical model $P_{m(\theta)}$ and point of localization $\gamma_0$. Let
\begin{equation}\label{eq:theta.n}
\Theta_n: = \{ \theta \in \Theta \: | \: \theta = \theta_0 + \theta_h/\sqrt{n}, \quad \textrm{where } \theta_h=(h^{\top},\mu^*)^{\top}  \textrm{satisfies } h \in M_h, \mu^* \in I_{\infty}(h) \}.    
\end{equation}
Together with 
our choice of $\gamma_0$, $M_h$ and $I_\infty(h)$, observe that for each $\theta_n\in\Theta_n$, we have $m(\theta_n)=\gamma_0+\frac{h}{\sqrt{n}}$ and $U(\theta_n)=\frac{\mu^*}{\sqrt{n}}$, where $h\in M_h$ and $\mu^*\in I_\infty(h)$. 
Therefore, for any $d_n$ that converges in the sense of Lemma \ref{lem:DQM}, we have
\begin{align*}
\sqrt{n}R(d_{n},\theta_n) & =\sqrt{n}\frac{\mu^{*}}{\sqrt{n}}\left(\mathbf{1}\left\{ \frac{\mu^{*}}{\sqrt{n}}\geq0\right\} -\mathbb{E}_{\gamma_{0}+\frac{h}{\sqrt{n}}}[d_{n}(Y^n)]\right)\\
 &\rightarrow \mu^{*}\left(\mathbf{1}\left\{ \mu^{*}\geq0\right\} -\mathbb{E}[d_{\infty}(\Delta)]\right),\text{as }n\rightarrow\infty.
\end{align*}
Then, the limit decision problem becomes the following:
The researcher observes $\Delta\sim\mathcal{N}(h,\mathbf{I}_{0}^{-1})$
and needs to come up with a rule $d_{\infty}:\mathbb{R}^{d_m}\rightarrow[0,1]$
that maps $\Delta$ to the unit interval. The limit payoff-relevant parameter is $\mu^{*}\in\mathbb{R}$,  partially identified with an
identified set $I_{\infty}(h)$.
In this limit game, a decision rule is MMR optimal if it achieves the value  
\begin{equation}
R^*:=\min_{d_{\infty}}\sup_{h\in M_h,\mu^{*}\in I_{\infty}(h)}\mu^{*}\left(\mathbf{1}\left\{ \mu^{*}\geq0\right\} -\mathbb{E}[d_{\infty}(\Delta)]\right).\label{eq:limit.game.1}
\end{equation}
We are now ready to define our notion of local asymptotic MMR optimality.  For $\Theta_n$ that has the local parametrization \eqref{eq:theta.n},  let 
\begin{equation}\label{eq:formal.local.space}
\Theta_{n}(J):=\left\{ \theta_n\in\Theta_n\mid \theta_n=\theta_0+\theta_{h}/\sqrt{n},\theta_h=(h^{\top},\mu^*)^{\top}, h\in J,\mu^{*}\in I_{\infty}(h)\right\},
\end{equation}
where $J$ is a finite subset of  $M_{h}$. A rule
$d_{n}$ is \emph{locally asymptotically MMR optimal} if\footnote{Let $\mathcal{D}$ be the set of all sequences of rules that converge in the sense of Lemma \ref{lem:DQM}. Our definition implies that $d_n$ is asymptotically MMR optimal in $\mathcal{D}$, i.e., the RHS of \eqref{eq:def.asymptotic.mmr} may be equivalently replaced with $\inf_{d'_n\in\mathcal{D}}\sup_{J}\liminf_{n\rightarrow\infty}\sup_{\theta_n\in\Theta_{n}(J)}\sqrt{n}R(d'_{n},\theta_n)$.}
\begin{equation}\label{eq:def.asymptotic.mmr}
\sup_{J}\liminf_{n\rightarrow\infty}\sup_{\theta_n\in\Theta_{n}(J)}\sqrt{n}R(d_{n},\theta_n)=R^{*}.    
\end{equation}

\begin{theorem}\label{thm: general}
Suppose Assumptions \ref{asm:1}-\ref{asm:DQM} hold and let $d^*_{\infty}$ be a solution of  \eqref{eq:limit.game.1}. If a sequence of statistical decision rules $d_n:\mathbf{Y}^n\rightarrow[0,1]$ matches $d^*_{\infty}$ in the sense of \eqref{eq:matching.rule}, $d_n$ is locally asymptotically MMR optimal. 
\end{theorem}
Theorem \ref{thm: general} is our first main result. It justifies a Gaussian experiment with a suitable local limit identified set  as a limit experiment for a large class of treatment choice problems with partial identification. 
By finding a solution $d^*_\infty$ of the simpler Gaussian limit game \eqref{eq:limit.game.1}, one can discover an asymptotically optimal rule for the finite-sample problem \eqref{eq:finite.sample.mmr}. Specifically, Theorem \ref{thm: general} shows that any finite-sample rule $d_n$ such that $\mathbb{E}_{\gamma_{0}+h/\sqrt{n}}[d_{n}(Y^{n})]\rightarrow\mathbb{E}[d^*_\infty(\Delta)]$ for every $h$  as $n\rightarrow\infty$ is asymptotically MMR optimal in the sense of \eqref{eq:def.asymptotic.mmr}. In practice, once $d^*_\infty$ is found, $d_n$ can often be constructed by plugging in a best regular estimator for $\gamma_0$ as well as consistent estimators of other nuisance parameters in $d^*_\infty$.

\section{Solving the Limit Problem}\label{sec:results}

The illustrative example in Section \ref{sec:illustrate} corresponds to fully differentiable bounds and centrosymmetric parameter space. In general, however, the limit identified set in our characterization need not be centrosymmetric and may also display directionally differentiable bounds. In the next two results, we discover $d^*_\infty$ for two non-nested cases, namely by either relaxing centrosymmetry or full differentiability. These findings are also of interest as stand-alone results in the spirit of \citet{yata2021} or \citet{stoye2012minimax}.

\subsection{Fully Differentiable Bounds}
When $\underline{I}(\cdot)$
and $\overline{I}(\cdot)$ are fully differentiable at $\mu_{0}$, the identified set \eqref{eq:limit.ID} simplifies to
\[
I_{\infty}(h)=\left[\underline{C}+\left(\underline{I}^{(1)}(\mu_{0})\right)^{\top}G_{0}h,\overline{C}+\left(\overline{I}^{(1)}(\mu_{0})\right)^{\top}G_{0}h\right].
\]
In this case, the Gaussian limit game \eqref{eq:limit.game.1} can be solved by considering an even simpler, one-dimensional limit
game. Let $\kappa_{0}:=(\underline{I}(\mu_{0}))^{2}/(\overline{I}(\mu_{0}))^{2}>0$.  Lemma \ref{lem:CoMo} in Appendix \ref{sec:app.1} 
establishes that  
\begin{equation}\label{eq:co-mono}
\underline{I}^{(1)}(\mu_{0})=\overline{I}^{(1)}(\mu_{0})\cdotp\kappa_{0},
\end{equation}
i.e., the gradients of the lower and upper bounds of the original identified set at $\mu_0$ must share the same direction.\footnote{Intuitively, since $\mu_0$ is a global least favorable point in the sense of Assumption \ref{asm:2}, moving away from it cannot increase the value of $\overline{R}(\cdot)$. In the case of full differentiability, that pins down a first-order condition implying \eqref{eq:co-mono}.} It follows that $I_\infty(h)$ depends on $h$ only via $\upsilon:=\left(\overline{I}^{(1)}(\mu_{0})\right)^{\top}G_{0}h\in\mathbb{R}$. Therefore, consider the following one-dimensional game: a researcher observes a scalar-valued  signal
\[\hat{\upsilon}:=\left(\overline{I}^{(1)}(\mu_{0})\right)^{\top}G_0\Delta\sim\mathcal{N}\left(\upsilon,\sigma^{2}\right),\]
where $\sigma^{2}:=\left(\overline{I}^{(1)}(\mu_{0})\right)^{\top}G_{0}\mathbf{I}_{0}^{-1}G_{0}^{\top}\overline{I}^{(1)}(\mu_{0})$
is known and strictly positive;  a limit welfare contrast is $\upsilon^{*}\in\mathbb{R}$,
partially identified with an identified set 
\[
I(\upsilon)=\left[\underline{C}+\kappa_{0}\upsilon,\overline{C}+\upsilon\right],
\]
for values of $\upsilon$ such that $\left(\kappa_{0}-1\right)\upsilon\leq\overline{C}-\underline{C}$.
In this simple, one-dimensional game, a rule $d:\mathbb{R}\rightarrow[0,1]$
is MMR optimal if it solves 
\begin{equation}
\min_{d}\sup_{\upsilon\in\mathbb{R},\left(\kappa_{0}-1\right)\upsilon\leq\overline{C}-\underline{C},\upsilon^{*}\in I(\upsilon)}\upsilon^{*}\left(\mathbf{1}\left\{ \upsilon^{*}\geq0\right\} -\mathbb{E}[d(\hat{\upsilon})]\right)\label{eq:one.dim.limit.game.smooth}
\end{equation}

The shape of the parameter space in game \eqref{eq:one.dim.limit.game.smooth} depends on the value of $\kappa_0$ and can be one of the three cases illustrated in Figure \ref{fig:parameter.space.smooth}. The case of $\kappa_0=1$ nests a convex and centrosymmetric parameter space with fully differentiable bounds analyzed in the existing literature. Otherwise, the parameter space is still convex but not centrosymmetric in general.  We characterize analytically the solution of  \eqref{eq:one.dim.limit.game.smooth}, which we show is in fact a MMR solution of the limit game  \eqref{eq:limit.game.1}. 

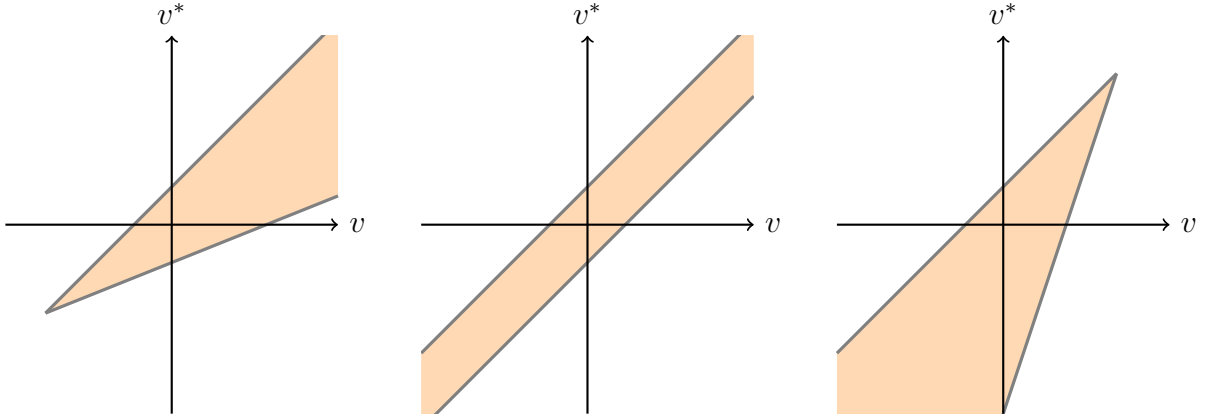
\begin{figure}[http]
\centering
\begin{tikzpicture}[
    axis/.style={->, thick, black},
    redline/.style={very thick, gray},
    blueline/.style={very thick, gray}
]
\def\xmin{-2.2}
\def\xmax{2.2}
\def\ymin{-2.5}
\def\ymax{2.5}

% ----------------------------------------------------------
% PANEL 1 (left): red v*=v+0.5, blue v*=0.4v-0.5
%   Intersection: v+0.5 = 0.4v-0.5 => v=-1.67, v*=-1.17
%   Lines START at intersection, go right only
% ----------------------------------------------------------
\begin{scope}[xshift=0cm]
  \begin{scope}
    \clip (\xmin,\ymin) rectangle (\xmax,\ymax);
    % Triangle: apex at intersection, two corners at right edge
    \fill[orange!30]
      (-1.67, -1.17)
      -- (\xmax, {\xmax+0.5})
      -- (\xmax, {0.4*\xmax-0.5})
      -- cycle;
    % Lines only from intersection point rightward
    \draw[redline]  (-1.67,-1.17) -- (\xmax, {\xmax+0.5});
    \draw[blueline] (-1.67,-1.17) -- (\xmax, {0.4*\xmax-0.5});
  \end{scope}
  \draw[axis] (\xmin,0) -- (\xmax,0) node[right] {$v$};
  \draw[axis] (0,\ymin) -- (0,\ymax) node[above] {$v^*$};
\end{scope}

% ----------------------------------------------------------
% PANEL 2 (middle): red v*=v+0.5, blue v*=v-0.5 (parallel)
%   No intersection — full band across window
% ----------------------------------------------------------
\begin{scope}[xshift=5.5cm]
  \begin{scope}
    \clip (\xmin,\ymin) rectangle (\xmax,\ymax);
    \fill[orange!30]
      (\xmin, {\xmin-0.5})
      -- (\xmin, {\xmin+0.5})
      -- (\xmax, {\xmax+0.5})
      -- (\xmax, {\xmax-0.5})
      -- cycle;
    \draw[redline]  (\xmin, {\xmin+0.5}) -- (\xmax, {\xmax+0.5});
    \draw[blueline] (\xmin, {\xmin-0.5}) -- (\xmax, {\xmax-0.5});
  \end{scope}
  \draw[axis] (\xmin,0) -- (\xmax,0) node[right] {$v$};
  \draw[axis] (0,\ymin) -- (0,\ymax) node[above] {$v^*$};
\end{scope}

% ----------------------------------------------------------
% PANEL 3 (right): red v*=v+0.5, blue v*=3v-2.5
%   Intersection: v+0.5 = 3v-2.5 => 2v=3 => v=1.5, v*=2.0
%   Lines END at intersection, come from left only
%   Blue shifted down by 2 (was 3v-0.5, now 3v-2.5) for more gap
% ----------------------------------------------------------
\begin{scope}[xshift=11cm]
  \begin{scope}
    \clip (\xmin,\ymin) rectangle (\xmax,\ymax);
    % Triangle: apex at intersection, two corners at left edge
    \fill[orange!30]
      (1.5, 2.0)
      -- (\xmin, {\xmin+0.5})
      -- (\xmin, {3*\xmin-2.5})
      -- cycle;
    % Lines only from left edge up to intersection point
    \draw[redline]  (\xmin, {\xmin+0.5})   -- (1.5, 2.0);
    \draw[blueline] (\xmin, {3*\xmin-2.5}) -- (1.5, 2.0);
  \end{scope}
  \draw[axis] (\xmin,0) -- (\xmax,0) node[right] {$v$};
  \draw[axis] (0,\ymin) -- (0,\ymax) node[above] {$v^*$};
\end{scope}

\end{tikzpicture}
\caption{Parameter space in limit game \eqref{eq:one.dim.limit.game.smooth} (from left to right: $\kappa_0<1$, $\kappa_0=1$ and $\kappa_0>1$)}
\label{fig:parameter.space.smooth}
\end{figure}
Let $\phi(\cdotp)$ and $\Phi(\cdotp)$ be the standard normal pdf and
cdf and  $\Phi^{-1}(\cdotp)$ the inverse of $\Phi(\cdotp)$. Define  
\begin{align*}
\overline{\sigma} & :=-\frac{\overline{C}}{\underline{C}}\left(\overline{C}-\underline{C}\right)\phi\left(\Phi^{-1}\left(\frac{\overline{C}}{\overline{C}-\underline{C}}\right)\right),\quad 
t^{*} :=-\Phi^{-1}\left(\frac{\overline{C}}{\overline{C}-\underline{C}}\right)\overline{\sigma}.
\end{align*}
Also, for each $t,\upsilon\in\mathbb{R}$, and each $\sigma>0$, write 
\begin{align*}
R_{2}(t,\upsilon;\sigma)  :=(\upsilon+\overline{C})\Phi\left(\frac{t-\upsilon}{\sigma}\right),\quad
R_{1}(t,\upsilon;\sigma)  :=-(\kappa_{0}\upsilon+\underline{C})\Phi\left(\frac{\upsilon-t}{\sigma}\right).
\end{align*}

\begin{theorem}\label{thm:full-diff} Suppose
Assumptions \ref{asm:1}-\ref{asm:DQM} hold,  $\underline{I}(\cdot)$
and $\overline{I}(\cdot)$ are fully differentiable at $\mu_{0}$,  $G_{0}\mathbf{I}_{0}^{-1}G_{0}^{\top}$ is positive definite and $\sigma^{2}>0$.  
The following results hold true for the limit game \eqref{eq:limit.game.1}.
 
\begin{itemize}
\item[(i)] If $\sigma<\overline{\sigma}$, 
$
d_{\infty,RT}=\Phi\left(\frac{\hat{\upsilon}-t^{*}}{\sqrt{\overline{\sigma}^{2}-\sigma^{2}}}\right)
$
is MMR optimal.
\item[(ii)] If $\sigma\geq\overline{\sigma}$, then 
$\mathbf{1}\left\{ \hat{\upsilon}\geq t_{0}\right\}$
is MMR optimal, where $t_{0}$ is the unique $t\in\mathbb{R}$ such
that 
\begin{equation}\label{eq:fixed.point.program.simple}
\max_{\left\{ \upsilon\in\mathbb{R}:\left(\kappa_{0}-1\right)\upsilon\leq\overline{C}-\underline{C},\upsilon+\overline{C}\geq0\right\} }R_{2}(t,\upsilon;\sigma)=\max_{\left\{ \upsilon\in\mathbb{R}:\left(\kappa_{0}-1\right)\upsilon\leq\overline{C}-\underline{C},\kappa_{0}\upsilon+\underline{C}\leq0\right\} }R_{1}(t,\upsilon;\sigma).    
\end{equation}
\end{itemize}
\end{theorem}

Theorem \ref{thm:full-diff} demonstrates two interesting features of the multivariate-signal game \eqref{eq:limit.game.1} with fully differentiable bounds at $\mu_0$. First, there exists a MMR optimal rule that depends on data $\Delta$ only via an efficient  linear combination $\hat{\upsilon}$. Along this direction, the limit MMR problem becomes effectively one-dimensional. Second, the key insights of the MMR optimal rule for the symmetric case \citep{stoye2012minimax,yata2021,MQS2023decision} extend to a non-centrosymmetric parameter space. In particular, $\overline{\sigma}$ depends only on $\overline{C}$ and $\underline{C}$ and can be interpreted as a measure of the severity of the identification problem, while $\sigma$ signals the effective estimation difficulty. If estimation difficulty dominates the severity of partial identification (case (ii)), an MMR optimal rule is a threshold rule, although the threshold is no longer zero in general and must be found via a (simple one-dimensional) fixed point  program \eqref{eq:fixed.point.program.simple}.\footnote{The idea of the fixed point program is analogous to that in \cite{tetenov2012statistical}, who focuses on MMR optimal rules in point-identified situations with asymmetric welfare regret. In this regard, we uncover an interesting connection between a point-identified problem with asymmetric regret and a partially identified problem with symmetric regret but non-centrosymmetric parameter space.}  If, however, the severity of partial identification  dominates estimation difficulty (case (i)), the  MMR optimal rule randomizes, in which case we find one randomized rule $d_{\infty,RT}$. We are confident that, in case (i), there are  infinitely many MMR optimal rules, although pinning down a piecewise linear rule in the spirit of \cite{MQS2023decision}  appears to be extremely algebraically involved. Theorem \ref{thm:full-diff} informs the following asymptotic results.  

\begin{theorem}\label{thm:full-diff-asymptotic}

Suppose the conditions of Theorem \ref{thm:full-diff} hold.  Moreover, let $\hat{\gamma}$
be a best regular estimator of $\gamma_{0}$ such that 
\[
\sqrt{n}(\hat{\gamma}-\gamma_{0}-\frac{h}{\sqrt{n}})\overset{h}{\rightsquigarrow}N(\mathbf{0},\mathbf{I}_{0}^{-1}),\text{ for all }h\in\mathbb{R}^{d_m}.
\]
Let $\hat{\sigma}$ be a consistent estimator of $\sigma$ when $h=\mathbf{0}$,
$\hat{t}$ be the unique $t\in\mathbb{R}$ such that 
\eqref{eq:fixed.point.program.simple} holds with $\sigma$ replaced by $\hat{\sigma}$.  Then, the following feasible  rule is locally asymptotically
MMR optimal: 
\[
d_{F}:=\mathbf{1}\left\{ \hat{\sigma}\geq\overline{\sigma}\right\} \cdot d_{F,\hat{t}}+\mathbf{1}\left\{ \hat{\sigma}<\overline{\sigma}\right\} \cdot d_{F,RT},
\]
where
\begin{align*}
d_{F,\hat{t}} & :=\mathbf{1}\left\{ \sqrt{n}\left(\overline{I}^{(1)}(\mu_{0})\right)^{\top}(\mu(\hat{\gamma})-\mu_0)\geq\hat{t}\right\} ,\quad
d_{F,RT}:  =\Phi\left(\frac{\sqrt{n}\left(\overline{I}^{(1)}(\mu_{0})\right)^{\top}(\mu(\hat{\gamma})-\mu_0)-t^{*}}{\sqrt{\overline{\sigma}^{2}-\hat{\sigma}^{2}}}\right).
\end{align*}
\end{theorem}
In our partially identified setting, 
the shape of the limit optimal rule found in Theorem \ref{thm:full-diff} depends on the value of $\sigma$, which is unknown and must be estimated. As a result, establishing asymptotic optimality of $d_F$
 is more challenging than in \cite{HiranoPorter2009}. We show the discontinuities caused by the regime change do not affect the asymptotic validity of $d_F$. In particular, irrespective of the true value of $\sigma$, $d_F$ is always matched  with the right limiting optimal rule found in Theorem \ref{thm:full-diff}.

% \begin{rem}\label{rem:multiple}
The form of our feasible  rule  explicitly depends on the location of $\mu_0$. Our asymptotic theory applies  to any least favorable point that satisfies Assumption \ref{asm:2} and does not require uniqueness of $\mu_0$. However, if there are multiple least favorable points, they may lead to different decision rules. When this happens, we think researchers should have the freedom to judge judiciously  which point of localization is more pertinent to their analyses. Otherwise, we recommend to report the one associated with a higher calculated value of $R^*$ defined in \eqref{eq:limit.game.1}. 
% \end{rem}

\subsection{Directionally Differentiable Bounds}\label{sec:centro}
Many partially identified models feature only directionally differentiable bounds. When $\underline{I}(\cdot)$ and $\overline{I}(\cdot)$
 are only directionally differentiable at $\mu_{0}$, the geometry of our limit game becomes more complicated. 
To gain tractability, let 
\[
\Theta_{\mathcal{M}}:=\left\{ \left(\mu(m(\theta))^{\top},U(\theta)\right)^{\top}\in\mathbb{R}^{d_{\mu}+1}\mid\theta\in\Theta\right\} 
\]
be the set of values the decision-relevant parameter $\mu(m(\theta))$
and the payoff-relevant parameter $U(\theta)$ can take for the original problem \eqref{eq:finite.sample.mmr}. 
We maintain
the following:
\begin{assumption}\label{asm:centro} \begin{itemize}
\item[(i)] $\Theta_{\mathcal{M}}$ is
convex and centrosymmetric.
\item[(ii)] There exists some $\mu\in M_{\mu}$
such that $\overline{I}(\mu)>\overline{I}(\mathbf{0})$. 
\end{itemize}
\end{assumption}
With the additional Assumption \ref{asm:centro}(i), $\overline{I}(\cdot)$ is concave.\footnote{For example, one may apply \citet[][Lemma B.4]{yata2021} or \citet[][Lemma C.2]{MQS2023decision} to $\Theta_{\mathcal{M}}$.} Moreover, Lemma \ref{lem:global.lfp} shows that  $\mathbf{0}\in\mathbb{R}^{d_{\mu}}$ is a global least favorable point, at which we recenter our reduced-form parameter. As $h$ only affects the limit identified set via $G_{0}h$, we may still simplify the problem by  redefining a limit decision-relevant reduced-form parameter $\mu:=G_{0}h$\footnote{We slightly abuse notation here because $\mu$ is not the same as $\mu\in M_{\mu}$ in the original parameter space.} and 
searching among rules that depend on $\Delta$ only via 
\[
G_0\Delta\sim\mathcal{N}(\mu,\Sigma_{0}),
\]
where $\Sigma_{0}:=G_{0}\mathbf{I}_{0}^{-1}G_{0}^{\top}$ is assumed to be positive definite. The identified set for the limit payoff-relevant parameter then simplifies to 
\begin{equation}
 I_{\infty}(\mu)=[\underline{I}_{\infty}(\mu),\overline{I}_{\infty}(\mu)],\label{eq:limit.id.set.simple}
\end{equation}
where
$\underline{I}_{\infty}(\mu)=\underline{C}+\underline{I}^{(1)}(\mathbf{0};\mu)$, 
$\overline{I}_{\infty}(\mu)=\overline{C}+\overline{I}^{(1)}(\mathbf{0};\mu)$, and $\mu$ can take values in the set 
\[
M_{\mu,\infty}:=\{\mu\in\mathbb{R}^{d_{\mu}}\mid \overline{I}_{\infty}(\mu)\geq\underline{I}_{\infty}(\mu)\}. 
\]
Thus, one may consider finding a rule $d:\mathbb{R}^{d_\mu}\rightarrow[0,1]$ that solves  the following problem
\begin{equation}
\min_{d}\sup_{\mu^{*}\in[\underline{I}_{\infty}(\mu),\overline{I}_{\infty}(\mu)],\mu\in M_{\mu,\infty}}\mu^{*}\left[\mathbf{1}\left\{ \mu^{*}\geq0\right\} -\mathbb{E}[d(G_0\Delta)]\right].\label{pf:R.star}
\end{equation}
Note the parameter space of $(\mu,\mu^*)$ in \eqref{pf:R.star} is still convex and centrosymmetric, but data $G_0\Delta$ becomes Gaussian. This characterization endorses  the finite-sample results considered by \citet{yata2021} as in fact a \emph{bona fide} limiting decision problem, with a suitable adjustment of the  original identified set to its localized version. Therefore,  the main results in \cite{yata2021} apply to \eqref{pf:R.star}; in particular, an optimal MMR rule utilizes an efficient linear combination of $G_0\Delta$ and can be found by applying \citet[][Theorem 1]{yata2021}. Moreover, when $\overline{C}$ is sufficiently large, the results in \cite{MQS2023decision} imply there are infinitely many MMR optimal rules. Next, we contribute to the literature by showing that the efficient linear combination can  be alternatively found by solving either a simple quadratic program \eqref{eq:quadratic.program}  or a nonlinear optimization program \eqref{eq:non.benign} below.

Denote by $\partial\overline{I}(\mathbf{0})$ the superdifferential
of $\overline{I}(\cdotp)$ at $\mathbf{0}$.  By \citet[][p.217-218, Section 23]{rockafellarconvex}, we have 
\[
\overline{I}^{(1)}(\mathbf{0};\mu)=\min_{p\in\partial\overline{I}(\mathbf{0})}p^{\top}\mu,\quad \forall  \mu \in \mathbb{R}^{d_\mu},
\]
where $\partial\overline{I}(\mathbf{0})$ is nonempty, compact, and
convex and $\overline{I}^{(1)}(\mathbf{0};\cdotp)$ is a finite positively
homogeneous concave function (i.e., $\overline{I}^{(1)}(\mathbf{0};\alpha\mu)=\alpha\overline{I}^{(1)}(\mathbf{0};\mu)$ for all $\alpha>0$). Moreover, Assumption \ref{asm:centro}(ii) implies
that $\mathbf{0}\notin\partial\overline{I}(\mathbf{0})$. 
Therefore, denote  by  $\overline{p}:=\overline{p}(\Sigma_0)$ the
unique solution of 
\begin{equation}\label{eq:quadratic.program}
\min_{p\in\partial\overline{I}(\mathbf{0})}\left\{ p^{\top}\Sigma_{0}p\right\}^{1/2},  
\end{equation}
and write 
\begin{align*}
\overline{t} & :=\frac{2\overline{C}}{\underline{I}^{(1)}(\mathbf{0};\Sigma_{0}\overline{p})-\overline{I}^{(1)}(\mathbf{0};\Sigma_{0}\overline{p})}\in(0,\infty],
\end{align*}
with the understanding that $\overline{t}=\infty$ if $\underline{I}^{(1)}(\mathbf{0};\Sigma_{0}\overline{p})-\overline{I}^{(1)}(\mathbf{0};\Sigma_{0}\overline{p})=0$. 

\begin{theorem}\label{thm:non-diff}

Suppose Assumptions \ref{asm:1}-\ref{asm:centro} hold and $\Sigma_0$ is positive definite. Consider the limit game \eqref{eq:limit.game.1} with $\mu_0=\mathbf{0}$. The following statements hold true. 
\begin{itemize}
\item[(i)] If $\left( \frac{\pi}{2}\overline{p}^{\top}\Sigma_{0}\overline{p}\right) ^{1/2}<\overline{C}$,
\[
d_{\infty,RT}:=\Phi\left(\frac{\overline{p}^{\top}G_0\Delta}{\sqrt{\frac{2}{\pi}\overline{C}^2-\overline{p}^{\top}\Sigma_{0}\overline{p}}}\right)
\]
is MMR optimal.
\item[(ii)]  If $\left( \frac{\pi}{2}\overline{p}^{\top}\Sigma_{0}\overline{p}\right) ^{1/2}\geq\overline{C}$,
and 
\begin{equation}\label{eq:benign.non.benign.boundary}
\overline{x}:=\overline{x}(\Sigma_0) :=\arg\max_{x\in[0,\infty)}(\overline{C}+x)\Phi\left(-\frac{x}{\left(\overline{p}^{\top}\Sigma_{0}\overline{p}\right)^{1/2}}\right)\leq\overline{t}\overline{I}^{(1)}(\mathbf{0};\Sigma_{0}\overline{p}),   
\end{equation}
$d_{\infty,\overline{p}}:=\mathbf{1}\left\{ \overline{p}^{\top}G_0\Delta\geq0\right\}$
is MMR optimal. 
\item[(iii)] Otherwise, $d_{\infty,\mu^{o}}:=\mathbf{1}\left\{ \left(\mathbf{\mu}^{o}\right)^{\top}\Sigma_{0}^{-1}G_0\Delta\geq0\right\}$ is MMR optimal, where
$\mathbf{\mu}^{o}:=\mathbf{\mu}^{o}(\Sigma_0)$ uniquely solves 
\begin{equation}\label{eq:non.benign}
\max_{\{\mu\in M_{\mu,\infty}:\overline{I}^{(1)}(\mathbf{0};\mu)\geq0\}}r(\mu;\Sigma_0),  
 \end{equation}
and  $r(\mu;\Sigma_0):=\left(\overline{C}+\overline{I}^{(1)}(\mathrm{\mathbf{0}};\mu)\right)\Phi\left(-\sqrt{\mu^{\top}\Sigma_0^{-1}\mu}\right)$.  
\end{itemize}
\end{theorem}

The intuition behind Theorem \ref{thm:non-diff} is as follows. Focus on the limit game among rules that depend on $\Delta$ only via $G_0\Delta$. Consistent with the existing results in the literature, a least favorable prior randomizes evenly between two symmetric points around $\mathbf{0}$, i.e., between  $(-\mu,\underline{I}_\infty(-\mu))$ and $(\mu,\overline{I}_\infty(\mu))$ where $\overline{I}_\infty(\mu)> 0$.  In ``benign cases'' where $\Sigma_0$ is not too large (statements (i) and (ii) above), the direction of $\mu$ is found via \eqref{eq:quadratic.program}, which determines an efficient direction along $\Sigma_0\overline{p}$. Along this direction, the game again essentially becomes one-dimensional. In cases where $\left( \frac{\pi}{2}\overline{p}^{\top}\Sigma_{0}\overline{p}\right) ^{1/2}\leq\overline{C}$,  one has $\mu=0$, and a limit MMR optimal rule randomizes whenever the inequality is strict. In fact, applying \citet[][Theorem 1]{MQS2023decision} to case (i), we can conclude that infinitely many MMR optimal rules exist, among them the piecewise linear rule 
\begin{equation*}\label{eq:d.linear}
d^{*}_{\text{linear}}:=\left[\frac{\overline{p}^{\top}G_0\Delta+\rho^{*}}{2\rho^{*}}\right]_0^1,
\end{equation*}
where $\rho^*$ is the unique strictly positive solution of
\begin{equation*}
\left(\frac{\rho^{*}}{2\overline{C}}\right) -\frac{1}{2} + \Phi\left(-\frac{\rho^{*}}{\left(\overline{p}^{\top}\Sigma_{0}\overline{p}\right)^{1/2}}\right) =0.
\end{equation*} 
In case (ii), $\mu$ has a strictly  positive length (i.e., $\mu=t\Sigma_0\overline{p}$ for some $t>0$) whenever the inequality is strict, and $d_{\infty,\overline{p}}$ is the (a.e.) unique Bayes response. A caveat is that the described  equilibrium can only be sustained if the length of $\mu$ is not too large,  as $\mu$ must also be in  $M_{\mu,\infty}$, which may be bounded in direction $\Sigma_0\overline{p}$. Condition \eqref{eq:benign.non.benign.boundary} gives this boundary condition. If it is  violated, we must search for a possibly alternative optimal direction via the nonlinear program \eqref{eq:non.benign}. In all cases, we also provide a suitable two-point symmetric prior in the space of $(h,\mu^*)$ that supports our decision rules, confirming that, although they depend only on $G_0\Delta$, they are MMR optimal for problem \eqref{eq:limit.game.1}. With Theorem \ref{thm:non-diff}, we next provide  a feasible plug-in rule that will be asymptotically MMR optimal.

\begin{theorem}\label{thm:non-diff-asymptotic}

Suppose the conditions of Theorem \ref{thm:non-diff} hold. Moreover, let $\hat{\gamma}$ be a best regular estimator of $\gamma_{0}$ such that \[ \sqrt{n}(\hat{\gamma}-\gamma_{0}-\frac{h}{\sqrt{n}})\overset{h}{\rightsquigarrow}N(\mathbf{0},\mathbf{I}_{0}^{-1}),\text{ for all }h\in\mathbb{R}^{d_m}. \]
Let $\hat{\Sigma}$ be a positive definite and consistent estimator of $\Sigma_0$ when $h=\mathbf{0}$,  $\hat{p}:=\overline{p}(\hat{\Sigma})$, $\hat{x}:=\overline{x}(\hat{\Sigma})$ and let $\mathbf{\hat{\mu}}^{o}$ solve \eqref{eq:non.benign} with $\Sigma_0$ replaced by $\hat{\Sigma}$. The following feasible rule is locally asymptotically MMR optimal:
\begin{align*}
d_{F} & :=\mathbf{1}\left\{ \left(\frac{\pi}{2}\hat{p}^{\top}\hat{\Sigma}\hat{p}\right)^{1/2}<\overline{C}\right\} d_{F,RT}\\
 & +\mathbf{1}\left\{ \left(\frac{\pi}{2}\hat{p}^{\top}\hat{\Sigma}\hat{p}\right)^{1/2}\geq\overline{C},\hat{x}\leq\frac{2\overline{C}\overline{I}^{(1)}(\mathbf{0};\hat{\Sigma}\hat{p})}{\underline{I}^{(1)}(\mathbf{0};\hat{\Sigma}\hat{p})-\overline{I}^{(1)}(\mathbf{0};\hat{\Sigma}\hat{p})}\right\} d_{F,\hat{p}}\\
 & +\mathbf{1}\left\{ \left(\frac{\pi}{2}\hat{p}^{\top}\hat{\Sigma}\hat{p}\right)^{1/2}\geq\overline{C},\hat{x}>\frac{2\overline{C}\overline{I}^{(1)}(\mathbf{0};\hat{\Sigma}\hat{p})}{\underline{I}^{(1)}(\mathbf{0};\hat{\Sigma}\hat{p})-\overline{I}^{(1)}(\mathbf{0};\hat{\Sigma}\hat{p})}\right\} d_{F,\hat{\mu}^{o}},
\end{align*}
where
\begin{align*}
d_{F,RT} & :=\Phi\left(\frac{\sqrt{n}\hat{p}^{\top}\mu(\hat{\gamma})}{\sqrt{\frac{2}{\pi}\overline{C}^{2}-\hat{p}^{\top}\hat{\Sigma}\hat{p}}}\right),d_{F,\hat{p}}:=\mathbf{1}\left\{ \hat{p}^{\top}\mu(\hat{\gamma})\geq0\right\} ,
d_{F,\hat{\mu}^{o}} :=\mathbf{1}\left\{ \left(\mathbf{\hat{\mu}}^{o}\right)^{\top}\hat{\Sigma}^{-1}\mu(\hat{\gamma})\geq0\right\}.
\end{align*}

\end{theorem}

\section{Applications}\label{sec:applications}

\subsection{Contaminated and Corrupted Data}\label{sec:contaminate}
This example corresponds to a parameter space that can be recast as  convex and centrosymmetric with fully differentiable bounds; see \cite{qiu2026statistical} for a related example in regression discontinuity design.  Suppose a DM must choose between two treatments based on RCT data, but some outcomes are contaminated. For each unit $i$, let  $D_i\in\{0,1\}$ be its treatment status, where $D_i=1$ denotes treatment and $D_i=0$ denotes control. The realized outcome for each unit is  $Y_i=D_iY_i(1)+(1-D_i)Y_i(0)$, where $Y_i(1),Y_i(0)\in\{0,1\}$ are binary potential outcomes (success/no success) under treatment and control, respectively. The DM also observes $S_i\in\{0,1\}$, where $S_i=1$ means the realized outcome $Y_i$ is contaminated or missing and therefore cannot be used for decision making. The payoff-relevant parameter is the average treatment effect $\mu^*=\mathbb{E}[Y(1)]-\mathbb{E}[Y(0)]$. \cite{horowitz1995identification} derive the sharp identified set of $\mu^*$:
\begin{equation}\label{eq:HM-bound}
[(1-p_1)\gamma_1-(1-p_0)\gamma_0-p_0,(1-p_1)\gamma_1-(1-p_0)\gamma_0+p_1],   
\end{equation}
where $\gamma_1=\mathbb{E}[Y_i\mid D_i=1,S_i=0]$, $\gamma_0=\mathbb{E}[Y_i\mid D_i=0,S_i=0]$, $p_1=\mathbb{E}[S_i=1\mid D_i=1]$, and $p_0=\mathbb{E}[S_i=1\mid D_i=0]$. Suppose the contamination rates $p_1,p_0\in(0,1)$ are known, e.g., from RCT experts or through validations from past experiments. This fits  our framework with a convex and centrosymmetric parameter space by writing $\theta=(\mu^*,\gamma_1,\gamma_0)$, $m(\theta)=(\gamma_1,\gamma_0)$, and $U(\theta)=\mu^*$. The identified set depends on $\gamma_1$ and $\gamma_0$ only via $\mu:=(1-p_1)\gamma_1-(1-p_0)\gamma_0-(p_0-p_1)/2$, so that we can rewrite the identified set more compactly as $I(\mu)=[\underline{I}(\mu),\overline{I}(\mu)]$, where 
\begin{equation}\label{eq:HM-bound-centro}
 \underline{I}(\mu)=\mu-k, \overline{I}(\mu)=\mu+k   
\end{equation}
and $k=(p_1+p_0)/2\in(0,1)$. The global least favorable point has $\mu_0=0$,  $\overline{I}^{(1)}(\mu_0)=1$, and $\kappa_0=1$.  Theorem \ref{thm:full-diff-asymptotic} directly applies.

If $p_1$ and $p_0$ are also unknown, they become part of $m(\theta)$. In this case, $k$ in \eqref{eq:HM-bound-centro} is unknown as well. The global least favorable point  occurs on the boundary violating  Assumption \ref{asm:2}, namely at $p_1=p_0=1$, i.e.,  no realized outcome is observed and the identified set becomes completely uninformative (also resembling an issue faced in \cite{song2014point}). Our result does not apply to this irregular case.   In Appendix \ref{sec:non.standard}, we slightly modify our asymptotic framework to accommodate this situation so that one may use  estimated versions of $p_1$ and $p_0$ to construct optimal rules. See also \cite{xu2026asymptotic} for an alternative treatment.  

\subsection{Robust Welfare Analyses}
This example is inspired by \cite{kang2025robustness} and corresponds to a scenario with fully differentiable bounds but non-centrosymmetric parameter space. A policy maker contemplates levying an \emph{ad valorem}  tax $\tau$ on a good and using the tax revenue to generate some known social surplus $G$. Suppose the demand curve does not shift before and after the tax levy.  The status-quo price and quantity, denoted as $p_0$ and $q_0$, are known. The new price after the levy, $p_1=p_0(1+\tau)$, is also known.  The policy maker observes some individual-level experimental data that point-identify the new quantity $q_1:=q_1(\gamma)$, which itself can be a smooth transformation of some finite-dimensional parameter $\gamma$ from   structural or reduced-form modeling. Even if $q_1$ is perfectly observed, we still do not know the true knowledge of the demand curve between $p_1$ and $p_0$. As a result, the loss in consumer surplus   $\Delta CS$ is only partially identified. In this example, $U(\theta)=G-\Delta CS$, while $m(\theta)=\gamma$. 
\begin{figure}[http]
\centering
\begin{tikzpicture}[axis/.style={->, thick, black}]

\pgfmathsetmacro{\xuint}{4.75}              % <-- TURN THIS KNOB
\pgfmathsetmacro{\xlint}{\xuint * 0.2}
\pgfmathsetmacro{\slopeU}{2.0/\xuint}
\pgfmathsetmacro{\slopeL}{4.0/\xuint}

\pgfmathsetmacro{\lowerAtRight}{0.8 - \slopeL*\xuint}
\pgfmathsetmacro{\upperAtLeft} {2.0 - \slopeU*\xlint}

\coordinate (UL) at (\xlint, \upperAtLeft);
\coordinate (UR) at (\xuint, 0);
\coordinate (LR) at (\xuint, \lowerAtRight);
\coordinate (LL) at (\xlint, 0);

\fill[orange!25] (UL) -- (UR) -- (LR) -- (LL) -- cycle;

\draw[very thick, gray]
    (-0.3, {2.0 - \slopeU*(-0.3)}) -- ({\xuint+0.8}, {2.0 - \slopeU*(\xuint+0.8)});
\draw[very thick, gray]
    (-0.3, {0.8 - \slopeL*(-0.3)}) -- ({\xuint+0.8}, {0.8 - \slopeL*(\xuint+0.8)});

\draw[axis] (-0.5, 0) -- ({\xuint+1.2}, 0) node[right] {$q_1$};
\draw[axis] (0, -4.0) -- (0, 3.5)          node[above] {$G - \Delta CS$};

\draw[thick, myred, dashed] (\xlint,  3.2) -- (\xlint, -3.8);
\draw[thick, myred, dashed] (\xuint,  3.2) -- (\xuint, -3.8);

\end{tikzpicture}
    \caption{Identified set of the net social surplus in a stylized example from \cite{kang2025robustness}. The area between the dotted red lines is where the decision problem is nontrivial and where the global least favorable point is located.}
    \label{fig:robust.welfare}
\end{figure}
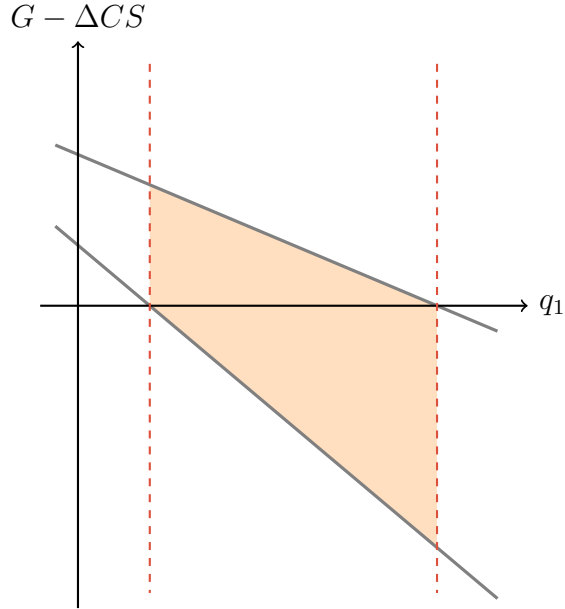

Assuming the gradient of the demand curve at any price between $p_1$ and $p_0$ is bounded between
\[
\left[\frac{q_1-q_0}{p_1-p_0}\frac{1}{1-r},\frac{q_1-q_0}{p_1-p_0}(1-r)\right]
\]
for some given $r\in(0,1)$, \cite{kang2025robustness} find the sharp identified set of $G-\Delta CS$ as $I(q_1):=[\underline{I}(q_1),\overline{I}(q_1)]$, where
\[
\overline{I}(q_1)= G-\frac{(p_{1}-p_{0})(1-r)q_{0}}{2-r}-\frac{p_{1}-p_{0}}{2-r}q_{1},\quad \underline{I}(q_1)= G-\frac{(p_{1}-p_{0})q_{0}}{2-r}-\frac{\left(p_{1}-p_{0}\right)\left(1-r\right)}{2-r}q_{1},
\]
which  depend on $\gamma$ only through $q_1$ and are affine in $q_1$.  See Figure \ref{fig:robust.welfare} for an illustration.  The problem is only interesting when $G$ is in the relevant range so that there exists some $q_1$ such that $\overline{I}(q_1)>0>\underline{I}(q_1)$. In that case, the global least favorable location $q_1^0$ has a closed-form solution. Once  $q_1^0$ is found, we can set $\overline{C}=\sqrt{n}\overline{I}(q_1^0)$ and $\underline{C}=\sqrt{n}\underline{I}(q_1^0)$.  Furthermore, note $\overline{I}^{(1)}(q_1^0)=-\left(p_{1}-p_{0}\right)/(2-r)$ and $\kappa_0=1-r<1$, implying the limit identified set corresponds to a case in the left panel of Figure \ref{fig:parameter.space.smooth}. Then, Theorem \ref{thm:full-diff-asymptotic} can be applied with an asymptotically efficient estimator for $q_1$ (e.g., by applying the transformation $q_1(\hat{\gamma})$ where $\hat{\gamma}$ is the MLE for $\gamma$) and  a consistent estimator of its asymptotic variance.

\subsection{Evidence Aggregation}\label{sec:IK}
We revisit an evidence aggregation example in \citet{ishihara2021}.
Suppose a DM is interested in implementing a new policy in a target
country and observes a random sample $\left\{ Y_{i},S_{i}\right\} _{i=1}^{n}$
collected from (for simplicity) two nearby countries, where $Y_{i}$
denotes the outcome of interest when the policy is applied to unit
$i$ and $S_{i}$ is a binary indicator for the country of origin
of each unit ($S_{i}=0$ means unit $i$ is from country 1 and $S_{i}=1$
means unit $i$ is from country 2). The DM specifies a parametric
model $P_{\vartheta}$ for $\{Y_{i},S_{i}\}$. The policy effects
of the two nearby countries, denoted as $\mu_{1}:=\mu_{1}(\vartheta)$
and $\mu_{2}:=\mu_{2}(\vartheta)$ respectively, are modeled as known
smooth functionals of $\vartheta$. The status quo policy effect is
known and normalized to zero. The DM is willing to extrapolate the
policy effect in the target country based on those from nearby countries:
\[
\left\{ \mu_{0}\in\mathbb{R}:\left|\mu_{0}-\mu_{j}\right|\leq C_{j},j=1,2\right\} 
\]
for some $C_{j}>0,j=1,2$. In this example, we have $\theta=\left(\mu_{0},\vartheta\right)\in\mathbb{R}^{1+d_{\vartheta}}$,
$U(\theta)=\mu_{0}$, $m(\theta)=\vartheta\in\mathbb{R}^{d_{\vartheta}}$,
and the identified set of $\mu_{0}$ depends on $m(\theta)$ only
via $\mu_{1}$ and $\mu_{2}$. % , implying 
% % \[
% % G_{0}=\left(\begin{array}{cccc}
% % 1 & 0 & 0 & 0\\
% % 1 & 1 & 0 & 0
% % \end{array}\right).
% % \]
Furthermore, the identified set for $\mu_{0}$ can be written as intersection
bounds: 
\begin{equation}\label{eq:IK.original.bound}
I(\mu_{1},\mu_{2})=\left[\underline{I}(\mu_{1},\mu_{2}),\overline{I}(\mu_{1},\mu_{2})\right], 
\end{equation}
where $\underline{I}(\mu_{1},\mu_{2})=\max\left\{ \mu_{1}-C_{1},\mu_{2}-C_{2}\right\} $
and $\overline{I}(\mu_{1},\mu_{2})=\min\left\{ \mu_{1}+C_{1},\mu_{2}+C_{2}\right\} $.
The space of $\left(\mu_{0},\mu_{1},\mu_{2}\right)$ is convex and
centrosymmetric, so $\mu_{1}=\mu_{2}=0$ is a global least favorable
point. 

Suppose there is one unique nearest neighbor, e.g., $0<C_{1}<C_{2}$.
Then, note at $\left(0,0\right)$, $\overline{I}(\mu_{1},\mu_{2})$
is fully differentiable with $\overline{I}^{(1)}(0,0)=(1,0)$. As
a result, our Theorems \ref{thm:full-diff} and \ref{thm:full-diff-asymptotic}
apply. In particular, the limit identified set in the sense of (\ref{eq:limit.id.set.simple})
would be
\begin{equation}\label{eq:ik.limit.bound.1}
\bar{I}_{\infty}(\mu_{1},\mu_{2})=\mu_{1}+\overline{C},~~\underline{I}_{\infty}(\mu_{1},\mu_{2})=\mu_{1}-\overline{C},    
\end{equation}
where $\overline{C}=\sqrt{n}\overline{I}(0,0)=\sqrt{n}C_{1}$. Note
how the geometry of \eqref{eq:ik.limit.bound.1} differs from the original identified set \eqref{eq:IK.original.bound},  precisely due to the way we constructed our shrinking identified set in \eqref{eq:local.id.U}-\eqref{eq:local.id.U.2}
that only incorporates the local information of $I(\cdotp)$ around $\left(0,0\right)$.\footnote{This is not the only reasonable way to approximate the decision problem. In Appendix \ref{sec:global.ID}, we discuss an alternative approach to construct the shrinking identified set, which allows incorporating more features of the original identified set. }  Let $\hat{\vartheta}$ be the MLE for $\vartheta$ and define $\hat{\mu}:=(\hat{\mu}_{1},\hat{\mu}_{2}):=(\mu_{1}(\hat{\vartheta}),\mu_{2}(\hat{\vartheta}))$.
Suppose 
\[
\Sigma_{0}:=\left(\begin{array}{cc}
\sigma_{1}^{2} & \rho\sigma_{1}\sigma_{2}\\
\rho\sigma_{1}\sigma_{2} & \sigma_{2}^{2}
\end{array}\right)
\]
is positive definite with $\sigma_{1},\sigma_{2}>0$ and $\rho\in(-1,1)$.
Let 
\[
\hat{\Sigma}:=\left(\begin{array}{cc}
\hat{\sigma}_{1}^{2} & \hat{\rho}\hat{\sigma}_{1}\hat{\sigma}_{2}\\
\hat{\rho}\hat{\sigma}_{1}\hat{\sigma}_{2} & \hat{\sigma}_{2}^{2}
\end{array}\right)
\]
be a positive definite and consistent estimator of $\Sigma_{0}$.
Theorems \ref{thm:full-diff} and \ref{thm:full-diff-asymptotic}
imply that the asymptotically optimal decision rule would be either
$\mathbf{1}\left\{ \hat{\mu}_{1}\geq0\right\} $ or its randomized
version depending on the magnitude of $2\overline{C}_{1}^{2}/\pi$
versus  $\hat{\sigma}_{1}^{2}$.

Now, consider the case when $C_{1}=C_{2}=C>0$. Then, at
$\left(0,0\right)$, $\overline{I}(\mu_{1},\mu_{2})$ is only directionally
differentiable and our Theorems \ref{thm:non-diff} and \ref{thm:non-diff-asymptotic}
apply. We set $\overline{C}=-\underline{C}=\sqrt{n}C$, and can calculate
that 
\[
\partial\overline{I}(0,0)=\left\{ \left(e,1-e\right)^{\top}\in\mathbb{R}^{2}\mid e\in[0,1]\right\} .
\]
For each $\mathbf{v}=(\mathbf{v}_{1},\mathbf{v}_{2})$ with unit length,
we have $\overline{I}^{(1)}((0,0);\mathbf{v})=\min\left\{ \mathbf{v}_{1},\mathbf{v}_{2}\right\} $,
which is nonnegative for all $\mathbf{v}$ in the first quadrant (i.e.,
both $\mathbf{v}_{1}$ and $\mathbf{v}_{2}$ are nonnegative). Analogously,
$\underline{I}^{(1)}((0,0);\mathbf{v})=\max\left\{ \mathbf{v}_{1},\mathbf{v}_{2}\right\} $.
Thus, in this example, the identified set in the limit game becomes
\[
\bar{I}_{\infty}(\mu)=\min\left\{ \mu_{1},\mu_{2}\right\} +\overline{C},~~\underline{I}_{\infty}(\mu)=\max\left\{ \mu_{1},\mu_{2}\right\} -\overline{C}.
\]
Moreover, along each direction $\mathbf{v}$ in the first quadrant,
the maximum possible length that $\mu$ can take is 
\begin{align*}
\overline{t}(\mathbf{v}) & =\frac{2\overline{C}}{\max\left\{ \mathbf{v}_{1},\mathbf{v}_{2}\right\} -\min\left\{ \mathbf{v}_{1},\mathbf{v}_{2}\right\} }=\begin{cases}
\frac{2\overline{C}}{\mathbf{v}_{1}-\mathbf{v}_{2}}, & \mathbf{v}_{1}>\mathbf{v}_{2},\\
\infty, & \mathbf{v}_{1}=\mathbf{v}_{2},\\
\frac{2\overline{C}}{\mathbf{v}_{2}-\mathbf{v}_{1}}, & \mathbf{v}_{1}<\mathbf{v}_{2}.
\end{cases}
\end{align*}
See also Figure \ref{fig:ik.example} for an illustration of $M_{\mu,\infty}$
and $\left\{ \mu\in M_{\mu,\infty}:\overline{I}^{(1)}(\mathbf{0};\mu)\geq0\right\} $
in this example. 
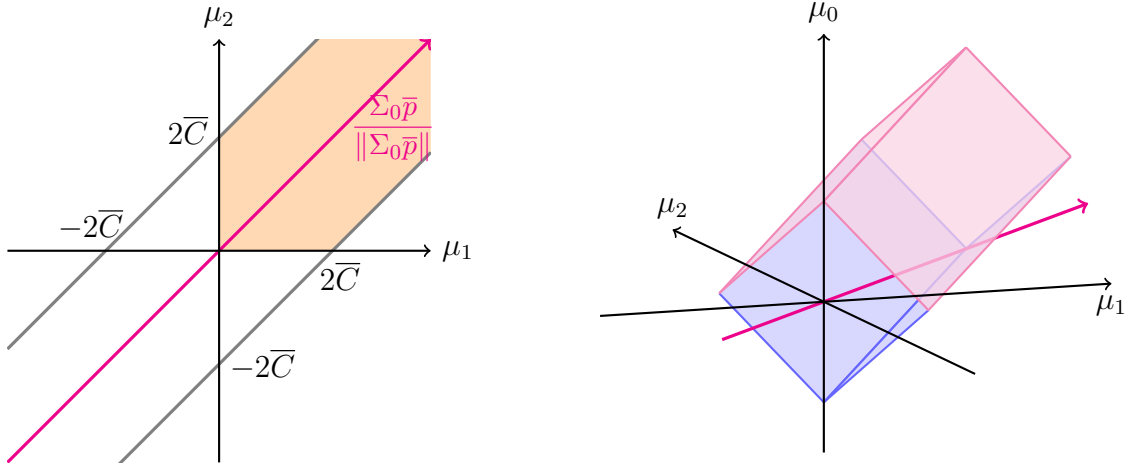
\begin{figure}
\centering 

\begin{tikzpicture}[
    axis/.style={->, thick, black},
    grayline/.style={very thick, gray},
    redline/.style={very thick, red},
    blueline/.style={very thick, blue}
]

% ============================================================
% PANEL 1 (left): 2D flat view
% ============================================================
\begin{scope}[xshift=0cm]
  \def\xmin{-2.8}
  \def\xmax{2.8}
  \def\ymin{-2.8}
  \def\ymax{2.8}
  \def\constant{1.5}

  \begin{scope}
    \clip (\xmin,\ymin) rectangle (\xmax,\ymax);
    \begin{scope}
      \clip (0,0) rectangle (\xmax,\ymax);
      \fill[orange!30]
        (\xmin, {\xmin-\constant})
        -- (\xmin, {\xmin+\constant})
        -- (\xmax, {\xmax+\constant})
        -- (\xmax, {\xmax-\constant})
        -- cycle;
    \end{scope}
    \draw[grayline]  (\xmin, {\xmin+\constant}) -- (\xmax, {\xmax+\constant});
    \draw[grayline] (\xmin, {\xmin-\constant}) -- (\xmax, {\xmax-\constant});
    \begin{scope}
      \clip (0,0) rectangle (\xmax,\ymax);
      \draw[gray, thin]  (\xmin, {\xmin+\constant}) -- (\xmax, {\xmax+\constant});
      \draw[gray, thin] (\xmin, {\xmin-\constant}) -- (\xmax, {\xmax-\constant});
    \end{scope}
    \draw[very thick, magenta] (\xmin,\ymin) -- (\xmax,\ymax);
    \draw[->, very thick, magenta] (\xmax-0.01,\ymax-0.01) -- (\xmax,\ymax);
  \end{scope}

  \draw[axis] (\xmin,0) -- (\xmax,0) node[right] {$\mu_1$};
  \draw[axis] (0,\ymin) -- (0,\ymax) node[above] {$\mu_2$};

  \node[left]  at (0, \constant+0.1)   {$2\overline{C}$};
  \node[right] at (0, -\constant)      {$-2\overline{C}$};
  \node[below] at (\constant+0.1, 0)   {$2\overline{C}$};
  \node[above] at (-\constant-0.2, 0)  {$-2\overline{C}$};

  \node[magenta, font=\small] at (2.3, 1.62)
    {$\dfrac{\Sigma_0 \overline{p}}{\|\Sigma_0 \overline{p}\|}$};
\end{scope}

% ============================================================
% PANEL 2 (right): 3D view
% ============================================================
\begin{scope}[xshift=8cm]
  \tdplotsetmaincoords{80}{-20}
  \begin{scope}[tdplot_main_coords, scale=0.9, xshift=0cm, yshift=-0.75cm]
    \def\tlo{0}
    \def\thi{3.5}
    \def\off{1.2}
    \def\apex{1.5}
    % Tilt factor: z = t * slope lifts the strip ~30 deg from horizontal
    % (arctan(slope * sqrt(2)) ~ 30 deg  =>  slope ~ 0.43)
    \def\slope{0.43}

    % --- BOTTOM HALF (blue, tilted) ---
    \fill[blue!20, opacity=0.75]
      ({\tlo-\off}, {\tlo+\off}, {\tlo*\slope})
      -- ({\thi-\off}, {\thi+\off}, {\thi*\slope})
      -- ({\thi}, {\thi}, {\thi*\slope-\apex})
      -- ({\tlo}, {\tlo}, {\tlo*\slope-\apex})
      -- cycle;
    \draw[blue!60, thick]
      ({\tlo-\off}, {\tlo+\off}, {\tlo*\slope})
      -- ({\tlo}, {\tlo}, {\tlo*\slope-\apex});
    \draw[blue!60, thick]
      ({\thi-\off}, {\thi+\off}, {\thi*\slope})
      -- ({\thi}, {\thi}, {\thi*\slope-\apex});
    \draw[blue!60, thick]
      ({\tlo}, {\tlo}, {\tlo*\slope-\apex})
      -- ({\thi}, {\thi}, {\thi*\slope-\apex});

    \fill[blue!20, opacity=0.75]
      ({\tlo+\off}, {\tlo-\off}, {\tlo*\slope})
      -- ({\thi+\off}, {\thi-\off}, {\thi*\slope})
      -- ({\thi}, {\thi}, {\thi*\slope-\apex})
      -- ({\tlo}, {\tlo}, {\tlo*\slope-\apex})
      -- cycle;
    \draw[blue!60, thick]
      ({\tlo+\off}, {\tlo-\off}, {\tlo*\slope})
      -- ({\tlo}, {\tlo}, {\tlo*\slope-\apex});
    \draw[blue!60, thick]
      ({\thi+\off}, {\thi-\off}, {\thi*\slope})
      -- ({\thi}, {\thi}, {\thi*\slope-\apex});
    \draw[blue!60, thick]
      ({\tlo}, {\tlo}, {\tlo*\slope-\apex})
      -- ({\thi}, {\thi}, {\thi*\slope-\apex});

    % --- Magenta direction arrow (flat, unchanged) ---
    \draw[very thick, magenta, ->]
      (-2.5, -2.5, 0) -- (6.5, 6.5, 0);

    \node[magenta, font=\small] at (-3.5, -3.5, 0){};

    % --- TOP HALF (magenta, tilted) ---
    \fill[magenta!20, opacity=0.75]
      ({\tlo-\off}, {\tlo+\off}, {\tlo*\slope})
      -- ({\thi-\off}, {\thi+\off}, {\thi*\slope})
      -- ({\thi}, {\thi}, {\thi*\slope+\apex})
      -- ({\tlo}, {\tlo}, {\tlo*\slope+\apex})
      -- cycle;
    \draw[magenta!60, thick]
      ({\tlo-\off}, {\tlo+\off}, {\tlo*\slope})
      -- ({\tlo}, {\tlo}, {\tlo*\slope+\apex});
    \draw[magenta!60, thick]
      ({\thi-\off}, {\thi+\off}, {\thi*\slope})
      -- ({\thi}, {\thi}, {\thi*\slope+\apex});
    \draw[magenta!60, thick]
      ({\tlo}, {\tlo}, {\tlo*\slope+\apex})
      -- ({\thi}, {\thi}, {\thi*\slope+\apex});

    \fill[magenta!20, opacity=0.75]
      ({\tlo+\off}, {\tlo-\off}, {\tlo*\slope})
      -- ({\thi+\off}, {\thi-\off}, {\thi*\slope})
      -- ({\thi}, {\thi}, {\thi*\slope+\apex})
      -- ({\tlo}, {\tlo}, {\tlo*\slope+\apex})
      -- cycle;
    \draw[magenta!60, thick]
      ({\tlo+\off}, {\tlo-\off}, {\tlo*\slope})
      -- ({\tlo}, {\tlo}, {\tlo*\slope+\apex});
    \draw[magenta!60, thick]
      ({\thi+\off}, {\thi-\off}, {\thi*\slope})
      -- ({\thi}, {\thi}, {\thi*\slope+\apex});
    \draw[magenta!60, thick]
      ({\tlo}, {\tlo}, {\tlo*\slope+\apex})
      -- ({\thi}, {\thi}, {\thi*\slope+\apex});

    % --- MAGENTA LINES (tilted: near end at z=0, far end rises) ---
    \draw[thick, magenta!60]
      ({\tlo+\off}, {\tlo-\off}, {\tlo*\slope})
      -- ({\thi+\off}, {\thi-\off}, {\thi*\slope});
    \draw[thick, magenta!60]
      ({\tlo-\off}, {\tlo+\off}, {\tlo*\slope})
      -- ({\thi-\off}, {\thi+\off}, {\thi*\slope});

    \draw[thick, ->] (0,0,-2.25) -- (0,0,4)   node[anchor=south]{$\mu_0$};
    \draw[thick, ->] (0,-6.5,0)  -- (0,6.5,0)  node[anchor=south]{$\mu_2$};
    \draw[thick, ->] (-3.5,0,0)  -- (4.5,0,0)  node[anchor=north]{$\mu_1$};
  \end{scope}
\end{scope}

\end{tikzpicture}

\caption{Left: Illustration of $M_{\mu,\infty}$ in the evidence aggregation example when $C_1=C_2=C$. The
highlighted area represents $\{ \mu\in M_{\mu,\infty}:\overline{I}^{(1)}(\mathbf{0};\mu)\geq0\}$. Nature's equilibrium choice of $(\mu_1,\mu_2)$ in the limit game is along the red line whenever $\rho\leq\text{min}\{ \sigma_{1}/{\sigma_{2}},\sigma_{2}/\sigma_{1}\}$. Right: A 3D illustration of the parameter space in the same example. 
\label{fig:ik.example}}
\end{figure}
Given the simple structure of $\partial\overline{I}(0,0)$, the unique
solution of the empirical analog of \eqref{eq:quadratic.program}
is $\hat{p}=\left(\hat{e},1-\hat{e}\right)^{\top}$, where 
\[
\hat{e}=\left[\frac{\hat{\sigma}_{2}^{2}-\hat{\rho}\hat{\sigma}_{1}\hat{\sigma}_{2}}{\hat{\sigma}_{1}^{2}+\hat{\sigma}_{2}^{2}-2\hat{\rho}\hat{\sigma}_{1}\hat{\sigma}_{2}}\right]_{0}^{1}.
\]
Moreover, 
\begin{align*}
\underline{I}^{(1)}(\mathbf{0};\hat{\Sigma}\hat{p})-\overline{I}^{(1)}(\mathbf{0};\hat{\Sigma}\hat{p}) & =\max\left\{ \hat{\Sigma}\hat{p}\right\} -\min\left\{ \hat{\Sigma}\hat{p}\right\} =\begin{cases}
0, & \hat{\rho}\leq\text{min}\left\{ \hat{\sigma}_{1}/\hat{\sigma}_{2},\hat{\sigma}_{2}/\hat{\sigma}_{1}\right\} ,\\
\hat{\rho}\hat{\sigma}_{1}\hat{\sigma}_{2}-\hat{\sigma}_{2}^{2}, & \hat{\rho}>\hat{\sigma}_{2}/\hat{\sigma}_{1},\\
\hat{\rho}\hat{\sigma}_{1}\hat{\sigma}_{2}-\hat{\sigma}_{1}^{2}, & \hat{\rho}>\hat{\sigma}_{1}/\hat{\sigma}_{2}.
\end{cases}
\end{align*}
Theorem \ref{thm:non-diff-asymptotic} can be directly applied. In
particular, if $\hat{\rho}\leq\min\left\{ \hat{\sigma}_{1}\big/\hat{\sigma}_{2},\hat{\sigma}_{2}\big/\hat{\sigma}_{1}\right\} $,
$\hat{e}\in[0,1]$ and information from both countries will always
be used if the inequality is strict. The asymptotically optimal decision
is either a randomized one based on $\hat{p}^{\top}\hat{\mu}$ or
a threshold rule $\mathbf{1}\{\hat{p}^{\top}\hat{\mu}\geq0\}$, depending
on the magnitude of $(\pi/2)\hat{p}^{\top}\hat{\Sigma}\hat{p}$ versus
$\overline{C}^{2}$. This characterization agrees with those in \citet[Section 3.2]{ishihara2021}
and \citet[Example 1]{yata2021} that focused on $\rho=0$.\footnote{\citet{ishihara2021} characterize the efficient weight numerically
for a class of non-randomized threshold rules. \citet{yata2021} characterizes
the optimal decision with techniques analogous to \citet{armstrong2018optimal}.
To be clear, their general results also apply to $\rho\neq0$. } The cases for $\hat{\rho}>\hat{\sigma}_{2}/\hat{\sigma}_{1}$ or
$\hat{\rho}>\hat{\sigma}_{1}/\hat{\sigma}_{2}$ are more involved.
For example, suppose $\hat{\rho}>\hat{\sigma}_{2}/\hat{\sigma}_{1}$.
Then, $\hat{p}=(0,1)^{\top}$. If $\left(\pi/2\right)^{1/2}\hat{\sigma}_{2}<\overline{C}$,
an asymptotic optimal rule is a randomized one based on $\hat{\mu}_{2}$
only. Let $\hat{x}$ solve $\max_{x\in[0,\infty)}(\overline{C}+x)\Phi\left(-x/\hat{\sigma}_{2}\right)$.
If $\left(\pi/2\right)^{1/2}\hat{\sigma}_{2}\geq\overline{C}$ and
$\hat{x}\leq(2\overline{C}\hat{\sigma}_{2})/(\hat{\rho}\hat{\sigma}_{1}-\hat{\sigma}_{2})$,
our optimal decision is simply $\mathbf{1}\left\{ \hat{\mu}_{2}\geq0\right\} $.
Otherwise, case (iii) of Theorem \ref{thm:non-diff} would apply.
The optimal weight of aggregating information from the two countries
must be computed by finding $\hat{\mu}^{o}=\arg\max_{\{\mu\in M_{\mu,\infty}:\min\left\{ \mu_{1},\mu_{2}\right\} \geq0\}}r(\mu;\hat{\Sigma})$,
and the optimal decision $\mathbf{1}\bigl\{\left(\mathbf{\hat{\mu}}^{o}\right)^{\top}\hat{\Sigma}^{-1}\hat{\mu}\geq0\bigr\}$
may use information from both countries again.

\section{Conclusion}\label{sec:conclude}
In this paper, we propose a new asymptotic framework to study treatment assignment problems with partially identified parameters. Our proposal has two novel features: (i) we consider a sequence of drifting parameters such that the extent of partial identification vanishes at the same rate as sampling uncertainty; (ii) we localize the reduced-form parameter at what we call a \emph{global least favorable point}, which we argue is a natural counterpart of the  \emph{hardest case} considered by \cite{HiranoPorter2009} for point-identified settings. Our approach characterizes the limit treatment choice problem in a Gaussian location shift model with a limit local identified set; to show this, we not only establish a rigorous link between finite-sample and limit problems but also generalize earlier analyses of the latter. Similarly, we show how to handle a large class of non-centrosymmetric models. Finally, our approach is computationally inexpensive, providing a simple recipe for applied researchers to find approximately optimal rules for a wide range of problems with partially identified parameters. 

Several extensions and open questions remain. First, it may be possible to solve the limit game with directional differentiability of the bounds even with centrosymmetry dropped (e.g., by utilizing Lemma \ref{lem:CoMo}(i)), although we suspect that this will involve more technicalities; we leave it for future research. Second, we also do not consider reduced-form parameters that are semiparametric in nature. It remains to be investigated to what extent our results for parametric models are relevant for their semiparametric counterparts.
Third, it may also be interesting to investigate whether one can combine our approach with recent ones by \citet{BenTal}, \citet{eisenhauer}, or \citet{andrews2025certified}, who take the confidence sets for identified quantities as constraining  the state space and solve the decision problem subject to those restrictions. 

\appendix

\section{Proofs of Main Results}\label{sec:app.1}

\begin{lemma}\label{lem:CoMo}
Let Assumptions \ref{asm:1} and \ref{asm:2} hold. Then the following
statements are true:
\begin{itemize}
\item[(i)] If $\overline{I}(\cdotp)$ and $\underline{I}(\cdotp)$ are directionally
differentiable at $\mu_{0}$ for all $\mathbf{v}\in\mathbb{R}^{d_{\mu}}$,
then 
\[
\underline{I}^{(1)}(\mu_{0};\mathbf{v})\geq\overline{I}^{(1)}(\mu_{0};\mathbf{v})\cdotp\kappa_{0},\forall\mathbf{v}\in\mathbb{R}^{d_{\mu}}.
\]
\item[(ii)] If $\overline{I}(\cdotp)$ and $\underline{I}(\cdotp)$ are differentiable
at $\mu_{0}$, we have $\underline{I}^{(1)}(\mu_{0})=\overline{I}^{(1)}(\mu_{0})\cdotp\kappa_{0}$. 
\end{itemize}
\end{lemma} 
\begin{proof}
Under Assumption
\ref{asm:2}, the value of $\overline{R}(\cdot)$ is strictly positive at $\mu_0$, which also 
cannot be increased by deviating from $\mu_{0}$. As $\underline{I}(\cdotp)$ and $\overline{I}(\cdotp)$
are directionally differentiable, so is $\overline{R}(\cdotp)$. Then, we must have 
\[
\overline{R}^{(1)}(\mu_{0};\mathbf{v})\leq0,\forall\mathbf{v}\in\mathbb{R}^{d_{\mu}},
\]
which after algebra, must imply (note $\overline{I}(\mu_{0})>0>\underline{I}(\mu_{0})$
and $\overline{I}(\mu_{0})-\underline{I}(\mu_{0})>0$)
\[
\overline{I}^{(1)}(\mu_{0};\mathbf{v})\left(\underline{I}(\mu_{0})\right)^{2}\leq\underline{I}^{(1)}(\mu_{0};\mathbf{v})\left(\overline{I}(\mu_{0})\right)^{2}.
\]
If $\underline{I}(\cdotp)$ and $\overline{I}(\cdotp)$ are differentiable,
so is $\overline{R}(\cdotp)$. It follows then 
$\overline{R}^{(1)}(\mu_{0})=0$,
implying 
\[
\overline{I}^{(1)}(\mu_{0})\left(\underline{I}(\mu_{0})\right)^{2}=\underline{I}^{(1)}(\mu_{0})\left(\overline{I}(\mu_{0})\right)^{2}.\] 
\end{proof}

\subsection{Proof of Theorem \ref{thm: general}}
Let $d_F$  be a sequence of rules (we omit the index $n$ for notational simplicity) that matches with $d^*_{\infty}$ in the sense of \eqref{eq:matching.rule}.

\textbf{Step 1}: We show 
$\sup_{J}\liminf_{n\rightarrow\infty}\sup_{\theta_n\in\Theta_{n}(J)}\sqrt{n}R(d_{F},\theta_n)\geq R^{*}$.
For each $\theta_n=\theta_0+\theta_h/\sqrt{n}$, we have 
$m(\theta_n)=\gamma_{0}+h/\sqrt{n}$ and $U(\theta_n)=\mu^{*}/\sqrt{n}$, where $h\in M_h$ and $\mu^{*}\in I_{\infty}(h)$. Thus, we can write 
\begin{align*}
\sqrt{n}R(d_{F},\theta_n) & =\sqrt{n}\frac{\mu^{*}}{\sqrt{n}}\left(\mathbf{1}\left\{ \frac{\mu^{*}}{\sqrt{n}}\geq0\right\} -\mathbb{E}_{\gamma_{0}+\frac{h}{\sqrt{n}}}[d_{F}]\right)
 =\mu^{*}\left(\mathbf{1}\left\{ \mu^{*}\geq0\right\} -\mathbb{E}_{\gamma_{0}+\frac{h}{\sqrt{n}}}[d_{F}]\right)\\
 & :=R\left(d_{F},\gamma_{0}+\frac{h}{\sqrt{n}},\mu^{*}\right).
\end{align*}
Note for each finite subset $J$ of $M_h$ and all $n$ sufficiently large,   we have $\left((\gamma_{0}+h/\sqrt{n})^{\top},\mu^{*}/\sqrt{n}\right)\in \Theta$ for all $h\in J$ and $\mu^{*}\in I_{\infty}(h)$. Thus, for all $n$ sufficiently large, 
\[\sup_{\theta_n\in\Theta_{n}(J)}\sqrt{n}R(d_{F},\theta_n)=\sup_{h\in J,\mu^{*}\in I_{\infty}(h)}R\left(d_{F},\gamma_{0}+\frac{h}{\sqrt{n}},\mu^{*}\right).\]
For $d^*_{\infty}$ in the limit game, write
\[
R_{\infty}(d^*_{\infty},h,\mu^{*}):=\mu^{*}\left(\mathbf{1}\left\{ \mu^{*}\geq0\right\} -\mathbb{E}[d^*_{\infty}]\right).
\]
Since  $d_{F}$ is matched with $d^*_{\infty}$ in the sense of 
\eqref{eq:matching.rule}, we have that, as $n\rightarrow\infty$, for each $h\in M_{h}$
and each $\mu^{*}\in I_{\infty}(h)$, 
\[
R\left(d_{F},\gamma_{0}+\frac{h}{\sqrt{n}},\mu^{*}\right)\rightarrow R_{\infty}(d^*_{\infty},h,\mu^{*}).
\]
Conclude that, for each finite subset $J$ of $M_{h}$, we have
\begin{align*}
\liminf_{n\rightarrow\infty}\sup_{\theta_n\in\Theta_{n}(J)}\sqrt{n}R(d_{F},\theta_n) & =\liminf_{n\rightarrow\infty}\sup_{h\in J,\mu^{*}\in I_{\infty}(h)}R\left(d_{F},\gamma_{0}+\frac{h}{\sqrt{n}},\mu^{*}\right)\\
 & \geq\sup_{h\in J,\mu^{*}\in I_{\infty}(h)}\liminf_{n\rightarrow\infty}R\left(d_{F},\gamma_{0}+\frac{h}{\sqrt{n}},\mu^{*}\right)\\
 & =\sup_{h\in J,\mu^{*}\in I_{\infty}(h)}\lim_{n\rightarrow\infty}R\left(d_{F},\gamma_{0}+\frac{h}{\sqrt{n}},\mu^{*}\right)\\
 & =\sup_{h\in J,\mu^{*}\in I_{\infty}(h)}R_{\infty}(d^*_{\infty},h,\mu^{*}).
\end{align*}
Therefore, 
\begin{align*}
\sup_{J}\liminf_{n\rightarrow\infty}\sup_{\theta_n\in\Theta_{n}(J)}\sqrt{n}R(d_{F},\theta_n) & \geq\sup_{J}\sup_{h\in J,\mu^{*}\in I_{\infty}(h)}R_{\infty}(d^*_{\infty},h,\mu^{*})=\sup_{h\in M_{h},\mu^{*}\in I_{\infty}(h)}R_{\infty}(d^*_{\infty},h,\mu^{*})=R^{*},
\end{align*}
where the last equality follows as $d^*_{\infty}$ is a solution of \eqref{eq:limit.game.1}.

\textbf{Step 2}: We show 
$\sup_{J}\liminf_{n\rightarrow\infty}\sup_{\theta_n\in\Theta_{n}(J)}\sqrt{n}R(d_{F},\theta_n)\leq R^{*}$. Fix a $\theta_n\in\Theta_n(J)$. We have
$m(\theta_n)=\gamma_{0}+\frac{h}{\sqrt{n}}$ and $U(\theta_n)=\frac{\mu^{*}}{\sqrt{n}}$, where $h\in J$ and $\mu^{*}\in I_{\infty}(h)$. Then, $\sqrt{n}R(d_{F},\theta_n)  =\mu^{*}\left(\mathbf{1}\left\{ \mu^{*}\geq0\right\} -\mathbb{E}_{\gamma_{0}+\frac{h}{\sqrt{n}}}[d_{F}]\right)$ and $R_{\infty}(d^*_{\infty},h,\mu^{*})=\mu^{*}\left(\mathbf{1}\left\{ \mu^{*}\geq0\right\} -\mathbb{E}[d^*_{\infty}]\right)$. Therefore, 
\begin{align*}
 & \left|\sqrt{n}R(d_{F},\theta_n)-R_{\infty}(d^*_{\infty},h,\mu^{*})\right|\\
= & \left|\mu^{*}\right|\left|\mathbb{E}_{\gamma_{0}+\frac{h}{\sqrt{n}}}[d_{F}]-\mathbb{E}[d^*_{\infty}]\right|\\
\leq & \max\left\{ \left|\overline{I}_{\infty}(h)\right|,\left|\underline{I}_{\infty}(h)\right|\right\} \left|\mathbb{E}_{\gamma_{0}+\frac{h}{\sqrt{n}}}[d_{F}]-\mathbb{E}[d^*_{\infty}]\right|.
\end{align*}
Note $d_{F}$ is matched with $d^*_{\infty}$ in the sense of \eqref{eq:matching.rule}.
Thus, for each $h\in J$ and each $\varepsilon>0$, there exists
some $N_{h}$ such that, for all $n>N_{h}$, 
\[
\left|\sqrt{n}R(d_{F},\theta_n)-R_{\infty}(d^*_{\infty},h,\mu^{*})\right|<\varepsilon
\]
for all $\mu^{*}\in I_{\infty}(h)$, implying that for each $h\in J$
and all $\mu^{*}\in I_{\infty}(h)$:
\begin{equation}
\sqrt{n}R(d_{F},\theta_n)\leq R_{\infty}(d^*_{\infty},h,\mu^{*})+\varepsilon\label{eq:pf.lem.asy.2}
\end{equation}
for $n>N_{h}$. Since $J$ is finite, we can pick $N=\max_{h\in J}\left\{ N_{h}\right\} $
so that, when  $n>N$, (\ref{eq:pf.lem.asy.2}) holds for all $h\in J$
and $\mu^{*}\in I_{\infty}(h)$. That implies for $n>N$, we have
\[
\sup_{\theta_n\in\Theta_{n}(J)}\sqrt{n}R(d_{F},\theta_n)\leq\sup_{h\in J,\mu^{*}\in I_{\infty}(h)}R_{\infty}(d^*_{\infty},h,\mu^{*})+\varepsilon,
\]
which in turn, means
\[
\liminf_{n\rightarrow\infty}\sup_{\theta_n\in\Theta_{n}(J)}\sqrt{n}R(d_{F},\theta_n)\leq\sup_{h\in J,\mu^{*}\in I_{\infty}(h)}R_{\infty}(d^*_{\infty},h,\mu^{*})+\varepsilon.
\]
As $\varepsilon>0$ is arbitrary, conclude that 
\[
\liminf_{n\rightarrow\infty}\sup_{\theta_n\in\Theta_{n}(J)}\sqrt{n}R(d_{F},\theta_n)\leq\sup_{h\in J,\mu^{*}\in I_{\infty}(h)}R_{\infty}(d^*_{\infty},h,\mu^{*}),
\]
and 
\begin{align*}
\sup_{J}\liminf_{n\rightarrow\infty}\sup_{\theta_n\in\Theta_{n}(J)}\sqrt{n}R(d_{F},\theta_n) & \leq\sup_{J}\sup_{h\in J,\mu^{*}\in I_{\infty}(h)}R_{\infty}(d^*_{\infty},h,\mu^{*})
 =\sup_{h\in M_{h},\mu^{*}\in I_{\infty}(h)}R_{\infty}(d^*_{\infty},h,\mu^{*}) =R^{*}.
\end{align*}

\subsection{Proof of Theorem \ref{thm:full-diff}}
We focus on the case $\kappa_{0}<1$, in which case note $\overline{C}>-\underline{C}>0$.
The proof for case $\kappa_{0}\geq1$ is analogous. Denote by $\mathbf{R}_{\upsilon}^{*}$ the value of the limit game among rules
that depend on $\Delta$ only via $\hat{\upsilon}=\left(\overline{I}^{(1)}(\mu_{0})\right)^{\top}G_{0}\Delta\sim\mathcal{N}\left(\upsilon,\sigma^{2}\right)$,
i.e., the value of \eqref{eq:one.dim.limit.game.smooth}.

\subsubsection*{Proof of Statement (i)}

\textbf{Step 1}: Note when $\sigma<\overline{\sigma}$, $d_{\infty,RT}$
is well defined, and $\mathbb{E}[d_{\infty,RT}(\hat{\upsilon})]=\Phi\left(\frac{\upsilon-t^{*}}{\overline{\sigma}}\right)$. Therefore, the worst-case expected regret of $d_{\infty,RT}$ in game \eqref{eq:one.dim.limit.game.smooth} writes  
\begin{align*}
\mathbf{R}_{RT}^{*} & :=\sup_{\upsilon\in\mathbb{R},\left(\kappa_{0}-1\right)\upsilon\leq\overline{C}-\underline{C},\upsilon^{*}\in I(\upsilon)}\upsilon^{*}\left(\mathbf{1}\left\{ \upsilon^{*}\geq0\right\} -\mathbb{E}[d_{RT}(\hat{\upsilon})]\right)\\
 & =\max\left\{ \sup_{\upsilon\in[-\overline{C},\infty)}R_{2}(t^{*},\upsilon;\overline{\sigma}),\sup_{\upsilon\in\left[\frac{\overline{C}-\underline{C}}{\kappa_{0}-1},-\frac{\underline{C}}{\kappa_{0}}\right]}R_{1}(t^{*},\upsilon;\overline{\sigma})\right\},
\end{align*}
where the second equality uses $\max\left\{ \frac{\overline{C}-\underline{C}}{\kappa_{0}-1},-\overline{C}\right\} =-\overline{C}$
when $\kappa_0<1$ (note $\kappa_{0}=\underline{C}^{2}/{\overline{C}^{2}}$). By construction, 
\[
\frac{\partial R_{2}(t^{*},\upsilon;\overline{\sigma})}{\partial\upsilon}\mid_{\upsilon=0}=0,\frac{\partial R_{1}(t^{*},\upsilon;\overline{\sigma})}{\partial\upsilon}\mid_{\upsilon=0}=0.
\]
Furthermore, Lemma \ref{lem:full.diff.lem.1} implies that $R_{2}(t^{*},\cdotp;\overline{\sigma})$
is first increasing then decreasing. Therefore, 
\[
\sup_{\upsilon\in[-\overline{C},\infty)}R_{2}(t^{*},\upsilon;\overline{\sigma})=R_{2}(t^{*},0;\overline{\sigma})=\overline{C}\Phi\left(\frac{t^{*}}{\overline{\sigma}}\right)=\frac{-\underline{C}\overline{C}}{\overline{C}-\underline{C}}.
\]
Analogously, we have
\[
\sup_{\upsilon\in\left[\frac{\overline{C}-\underline{C}}{\kappa_{0}-1},-\frac{\underline{C}}{\kappa_{0}}\right]}R_{1}(t^{*},\upsilon;\overline{\sigma})=R_{1}(t^{*},0;\overline{\sigma})=\frac{-\underline{C}\overline{C}}{\overline{C}-\underline{C}}.
\]
Conclude that $\mathbf{R}_{RT}^{*}=\frac{-\underline{C}\overline{C}}{\overline{C}-\underline{C}}$,
and $R^{*}\leq\mathbf{R}_{\upsilon}^{*}\leq\mathbf{R}_{RT}^{*}=\frac{-\underline{C}\overline{C}}{\overline{C}-\underline{C}}.$

\noindent\textbf{Step 2}: We show $R^{*}\geq\frac{-\underline{C}\overline{C}}{\overline{C}-\underline{C}}$
whenever $\sigma<\overline{\sigma}$. This is because 
\begin{align*}
R^{*} & \geq R_{\mathbf{0}}^{*}:=\min_{d(\Delta)}\sup_{h=\mathbf{0},\mu^{*}\in\left[\underline{C},\overline{C}\right]}\mu^{*}\left[\mathbf{1}\left\{ \mu^{*}\geq0\right\} -\mathbb{E}[d(\Delta)]\right].
\end{align*}
To find the value of $R_{\mathbf{0}}^{*}$, let $x_{d}:=\mathbb{E}[d(\Delta)]$
where $\Delta\sim \mathcal{N}(\mathbf{0},\mathbf{I}_{0}^{-1})$. Then,
\begin{align*}
 & \sup_{h=\mathbf{0},\mu^{*}\in\left[\underline{C},\overline{C}\right]}\mu^{*}\left[\mathbf{1}\left\{ \mu^{*}\geq0\right\} -\mathbb{E}[d(\Delta)]\right]
=  \max\left\{ \overline{C}\left(1-x_{d}\right),-\underline{C}x_{d}\right\} .
\end{align*}
Thus, $R_{\mathbf{0}}^{*}$ is achieved at any rule such that
$\overline{C}\left(1-x_{d}\right)=-\underline{C}x_{d}$, i.e., $x_{d}=\frac{\overline{C}}{\overline{C}-\underline{C}}$,
implying $R_{\mathbf{0}}^{*}=\frac{-\underline{C}\overline{C}}{\overline{C}-\underline{C}}$.

\noindent\textbf{Step 3}: Conclude that, based on Steps 1 and 2 above, we have $R^{*}=\frac{-\underline{C}\overline{C}}{\overline{C}-\underline{C}}$, and that  $d_{\infty,RT}$ achieves the value of $R^*$ and is MMR optimal for the limiting game \eqref{eq:limit.game.1}.

\subsubsection*{Proof of Statement (ii)}

\textbf{Step 1}: Consider the following class of threshold rules: $d_{t}(\hat{\upsilon}):=\mathbf{1}\left\{ \hat{\upsilon}\geq t\right\}$,  $t\in\mathbb{R}$. Analogous to proof of statement (i), in case of $\kappa_0<1$, the
worst case expected regret of $d_{t}(\hat{\upsilon})$ writes
\begin{align*}
 & \max\left\{ \max_{\left\{ \upsilon\in\mathbb{R}:\left(\kappa_{0}-1\right)\upsilon\leq\overline{C}-\underline{C},\upsilon+\overline{C}\geq0\right\} }R_{2}(t,\upsilon;\sigma),\max_{\left\{ \upsilon\in\mathbb{R}:\left(\kappa_{0}-1\right)\upsilon\leq\overline{C}-\underline{C},\kappa_{0}\upsilon+\underline{C}\leq0\right\} }R_{1}(t,\upsilon;\sigma) \right\}\\
 =&\max\left\{ \max_{\upsilon\in\left[-\overline{C},\infty\right)}R_{2}(t,\upsilon;\sigma),\max_{\upsilon\in\left[\frac{\overline{C}-\underline{C}}{\kappa_{0}-1},-\frac{\underline{C}}{\kappa_{0}}\right]}R_{1}(t,\upsilon;\sigma)\right\}.
\end{align*}
Lemma \ref{lem:full.diff.lem.1} shows that 
\begin{align*}
\overline{\upsilon}_{2}(t;\sigma) & :=\arg\max_{\upsilon\in\left[-\overline{C},\infty\right)}R_{2}(t,\upsilon;\sigma)
\end{align*}
and 
\[
\overline{\upsilon}_{1}(t;\sigma):=\arg\max_{\upsilon\in\left[\frac{\overline{C}-\underline{C}}{\kappa_{0}-1},-\frac{\underline{C}}{\kappa_{0}}\right]}R_{1}(t,\upsilon;\sigma)
\]
are single-valued continuous functions in both $t$ and $\sigma$,
and there exists some unique $t\in\mathbb{R}$ such that 
$R_{2}(t,\overline{\upsilon}_{2}(t;\sigma);\sigma)  =R_{1}(t,\overline{\upsilon}_{1}(t;\sigma);\sigma)$, 
which also implicitly defines a single-valued continuous function
$t_{0}(\sigma):\mathbb{R}^{+}\rightarrow\mathbb{R}$. 

\noindent\textbf{Step 2}: Consider the case when $\sigma=\overline{\sigma}$. Algebra
shows  $t_{0}(\overline{\sigma})=t^{*},\overline{\upsilon}_{2}(t_{0}(\overline{\sigma});\overline{\sigma})=\overline{\upsilon}_{1}(t_{0}(\overline{\sigma});\overline{\sigma})=0$. Then, completely analogous to proof of statement (i), we can establish that
$\mathbf{1}\left\{ \hat{\upsilon}\geq t^{*}\right\} $ is MMR optimal. 

\noindent\textbf{Step 3}: For the case when $\sigma>\overline{\sigma}$, Lemma \ref{lem:full.diff.lem.2} establishes
that $\overline{\upsilon}_{2}(t_{0}(\sigma);\sigma)>\overline{\upsilon}_{1}(t_{0}(\sigma);\sigma)$. To simplify notation, let $\beta:=G_0^{\top}\overline{I}^{(1)}(\mu_0)$.
Then, for some $p\in(0,1)$, consider a prior $\pi(p)$ in the space
of $(h,\mu^{*})$ that randomizes between
\begin{align*}
\left(I_{0}^{-1}\beta\frac{\overline{\upsilon}_{2}(t_{0}(\sigma);\sigma)}{\sigma^{2}},\beta^{\top}I_{0}^{-1}\beta\frac{\overline{\upsilon}_{2}(t_{0}(\sigma);\sigma)}{\sigma^{2}}+\overline{C}\right),\text{w/ prob }p & ;\\
\left(I_{0}^{-1}\beta\frac{\overline{\upsilon}_{1}(t_{0}(\sigma);\sigma)}{\sigma^{2}},\kappa_{0}\beta^{\top}I_{0}^{-1}\beta\frac{\overline{\upsilon}_{1}(t_{0}(\sigma);\sigma)}{\sigma^{2}}+\underline{C}\right),\text{w/ prob }1-p & .
\end{align*}
Note $\beta^{\top}I_{0}^{-1}\beta=\sigma^{2}$, $\overline{\upsilon}_{2}(t_{0}(\sigma);\sigma)+\overline{C}>0$
and $\kappa_{0}\overline{\upsilon}_{1}(t_{0}(\sigma);\sigma)+\underline{C}<0$.
It follows that, after algebra, the corresponding Bayes optimal rule
is $\mathbf{1}\left\{ \beta^{\top}\Delta\geq t(p)\right\}$,
where note $\beta^{\top}\Delta=\hat{\upsilon}$, and
\begin{align*}
t(p) & :=\frac{\sigma^{2}}{\overline{\upsilon}_{2}(t_{0}(\sigma);\sigma)-\overline{\upsilon}_{1}(t_{0}(\sigma);\sigma)}\log\left(\frac{-(1-p)\left(\kappa_{0}\overline{\upsilon}_{1}(t_{0}(\sigma);\sigma)+\underline{C}\right)}{p\left(\overline{\upsilon}_{2}(t_{0}(\sigma);\sigma)+\overline{C}\right)}\right)\\
 & +\frac{\overline{\upsilon}_{2}(t_{0}(\sigma);\sigma)+\overline{\upsilon}_{1}(t_{0}(\sigma);\sigma)}{2}.
\end{align*}
As $p$ ranges in $\left(0,1\right)$, $\frac{(1-p)}{p}\in(0,\infty)$.
Thus, there exists some $p_{0}\in(0,1)$ such that $t(p_{0})=t_{0}$,
implying that $\mathbf{1}\left\{ \hat{\upsilon}\geq t_{0}\right\} $
is Bayes optimal with respect to prior $\pi(p_{0})$. 

\noindent \textbf{Step 4}: We have shown that $\mathbf{1}\left\{ \hat{\upsilon}\geq t_{0}\right\} $
is Bayes optimal with respect to prior $\pi(p_{0})$ with a Bayes
value
\[
r(p_{0})=R_{2}(t_{0},\overline{\upsilon}_{2}(t_{0}(\sigma);\sigma);\sigma)=R_{1}(t_{0},\overline{\upsilon}_{1}(t_{0}(\sigma);\sigma);\sigma).
\]
Furthermore, Step 1 implies that the worst case expected regret of $\mathbf{1}\left\{ \hat{\upsilon}\geq t_{0}\right\}$ equals $r(p_0)$. Thus, we have established an equilibrium; $\mathbf{1}\left\{ \hat{\upsilon}\geq t_{0}\right\} $
is MMR optimal and $\pi(p_{0})$ is least favorable.

\subsection{Proof of Theorem \ref{thm:full-diff-asymptotic}}
Lemma \ref{lem:matching.diff} shows that $d_{F}$ is
matched with $d_{\infty,RT}$ in the sense of \eqref{eq:matching.rule} if $\sigma<\overline{\sigma}$ and  matched with 
$\mathbf{1}\left\{ \hat{\upsilon}\geq t_{0}\right\}$ 
if  $\sigma>\overline{\sigma}$. Lemma \ref{lem:full.diff.lem.3} shows that  $d_{F}$ is still matched with 
$\mathbf{1}\left\{ \hat{\upsilon}\geq t_{0}\right\}$ even when   $\sigma=\overline{\sigma}$. Then, the conclusion follows directly from invoking Theorem \ref{thm: general}.

\subsection{Proof of Theorem \ref{thm:non-diff}}

Under Assumptions \ref{asm:1}, \ref{asm:2}, \ref{asm:diff} and \ref{asm:centro}, recall the following properties: (i) $M_{\mu}$ is convex and centrosymmetric; (ii) $\overline{I}(\mu)=-\underline{I}(-\mu)$ for all $\mu\in M_{\mu}$, and in particular, $\overline{C}=-\underline{C}$; (iii) $\overline{I}(\cdotp)$ is concave, and $\underline{I}(\cdotp)$ is
convex;  (iv) $\overline{I}^{(1)}(\mathbf{0};\cdotp)$ is  concave and positively homogeneous; (v)  $\overline{I}^{(1)}(\mathbf{0};\mathbf{v})=-\underline{I}^{(1)}(\mathbf{0};-\mathbf{v})$
for any $\mathbf{v}\in\mathbb{R}^{d_{\mu}}$; 
(vi) $\overline{I}^{(1)}(\mathbf{0};\mathbf{v})\leq-\overline{I}^{(1)}(\mathbf{0};-\mathbf{v})=\underline{I}^{(1)}(\mathbf{0};\mathbf{v})$ for any $\mathbf{v}\in\mathbb{R}^{d_{\mu}}$. Moreover, Lemma \ref{lem:global.lfp} implies that  $\mu_0=\mathbf{0}$  is indeed a global least favorable point.  To simplify notation, let $\varDelta:=G_{0}\Delta\sim\mathcal{N}(\mu,\Sigma_0)$. Denote by $\mathbf{R}^{*}$ the value of \eqref{pf:R.star}, i.e., the value of \eqref{eq:limit.game.1}   among rules that depend on $\Delta$ only via $\varDelta$. Clearly, $R^*\leq\mathbf{R}^*$.

\subsubsection*{Proof of Statement (i)}
\textbf{Step 1}: 
We show $R^*\leq\frac{\overline{C}}{2}$ in case (i). 
Note $\overline{p}\in\partial \overline{I}(\mathbf{0})$, so we have
\[
\overline{I}^{(1)}(\mathbf{0};\mu)\leq\overline{p}^{\top}\mu,\text{ for all }\mu.
\]
As $\overline{I}^{(1)}(\mathbf{0};\mathbf{v})=-\underline{I}^{(1)}(\mathbf{0};-\mathbf{v})$
for any $\mathbf{v}\in\mathbb{R}^{d_{\mu}}$, we also have
\[
\underline{I}^{(1)}(\mathbf{0};\mu)\geq\overline{p}^{\top}\mu,\text{ for all }\mu.
\]
Conclude that $\mathbf{R}^{*}$ is upper bounded by the value of of the following ``one
dimensional linear embedding MMR game'': 
\begin{align}\label{eq:linear.embedding}
\min_{d}\sup_{\mu\in \mathbb{R}^{d_\mu},\mu^{*}\in\left[\underline{C}+\overline{p}^{\top}\mu,\overline{C}+\overline{p}^{\top}\mu\right]}\mu^{*}\left[\mathbf{1}\left\{ \mu^{*}\geq0\right\} -\mathbb{E}[d(\overline{p}^{\top}\varDelta)]\right],
\end{align}
where $d$ is a decision rule based on $\overline{p}^{\top}\varDelta\sim\mathcal{N}\left(\overline{p}^{\top}\mu,\overline{p}^{\top}\Sigma_{0}\overline{p}\right)$ only. 
The solution and value of \eqref{eq:linear.embedding} are completely known. When $\left\{\frac{\pi}{2} \overline{p}^{\top}\Sigma_{0}\overline{p}\right\} ^{1/2}<\overline{C}$, $d_{\infty,RT}$ solves \eqref{eq:linear.embedding} and achieves a value  of 
$\frac{\overline{C}}{2}$. Thus, $R^*\leq\frac{\overline{C}}{2}$. 

\noindent \textbf{Step 2}: We show $R^*\geq\frac{\overline{C}}{2}$ in case (i). 
This is because
\begin{align*}
R^{*} & \geq\min_{d(\Delta)}\sup_{h=\mathbf{0},\mu^{*}\in\left[\underline{C},\overline{C}\right]}\mu^{*}\left[\mathbf{1}\left\{ \mu^{*}\geq0\right\} -\mathbb{E}[d(\Delta)]\right]=\frac{\overline{C}}{2}.
\end{align*}

\noindent \textbf{Step 3}: Conclude that  $R^{*}=\frac{\overline{C}}{2}$. Moreover, as  $d_{\infty,RT}$  achieves a value of $\frac{\overline{C}}{2}$ in the
upper bound linear embedding game, and the worst case regret of $d_{\infty,RT}$
in the original game \eqref{eq:limit.game.1} cannot be larger than that of the upper bound linear
embedding game, we conclude that $d_{\infty,RT}$ must achieve the
value of $\frac{\overline{C}}{2}$ and is MMR optimal for \eqref{eq:limit.game.1}.

\subsubsection*{Proof of Statement (ii)}
\textbf{Step 1}: The existence and uniqueness of $\overline{x}$ is provided in Lemma \ref{lem:non-diff-continuous}. Analogous to Step 1 for the proof of statement (i),
when $\left\{ \frac{\pi}{2}\overline{p}^{\top}\Sigma_{0}\overline{p}\right\} ^{1/2}\geq\overline{C}$,
we have 
\begin{equation}\label{eq:pf.thm.non.diff.ii}
R^{*}\leq\mathbf{R}^{*}\leq(\overline{C}+\overline{x})\Phi\left(-\frac{\overline{x}}{\left(\overline{p}^{\top}\Sigma_{0}\overline{p}\right)^{1/2}}\right),   
\end{equation}
where the RHS of \eqref{eq:pf.thm.non.diff.ii}
is the value of the linear embedding game \eqref{eq:linear.embedding} and $\overline{x}>0$ if   $\left\{ \frac{\pi}{2}\overline{p}^{\top}\Sigma_{0}\overline{p}\right\} ^{1/2}>\overline{C}$.

\noindent \textbf{Step 2}: Let $\overline{\mathbf{h}}:=\mathbf{I}_{0}^{-1}G_{0}^{\top}\overline{p}$,
and consider $h$ taking values along the direction of $\overline{\mathbf{h}}$,
i.e., $h=t\overline{\mathbf{h}}$ for some $t\in\mathbb{R}$. We show
that $t$ takes values in an interval $\left[-\overline{t},\overline{t}\right]$,
where 
\[
\overline{t}=\frac{2\overline{C}}{\underline{I}^{(1)}(\mathbf{0};\Sigma_{0}\overline{p})-\overline{I}^{(1)}(\mathbf{0};\Sigma_{0}\overline{p})}>0,
\]
with the understanding that the interval becomes $(-\infty,\infty)$
if and only if $\underline{I}^{(1)}(\mathbf{0};\Sigma_{0}\overline{p})=\overline{I}^{(1)}(\mathbf{0};\Sigma_{0}\overline{p})$. To see this, note when $h=t\overline{\mathbf{h}}$,
the lower and upper bounds of the identified set for $\mu^{*}$ writes
(note $G_{0}\overline{\mathbf{h}}=\Sigma_{0}\overline{p}$)
\begin{align*}
\underline{I}_{\infty}(t\overline{\mathbf{h}})= & \underline{C}+\underline{I}^{(1)}(\mathbf{0};tG_{0}\overline{\mathbf{h}})
=  \begin{cases}
\underline{C}+t\underline{I}^{(1)}(\mathbf{0};\Sigma_{0}\overline{p}), & t\geq0,\\
\underline{C}-t\underline{I}^{(1)}(\mathbf{0};-\Sigma_{0}\overline{p}), & t<0,
\end{cases}
\end{align*}
and
\begin{align*}
\overline{I}_{\infty}(t\overline{\mathbf{h}}) & =\overline{C}+\overline{I}^{(1)}(\mathbf{0};tG_{0}\overline{\mathbf{h}}) =\begin{cases}
\overline{C}+t\overline{I}^{(1)}(\mathbf{0};\Sigma_{0}\overline{p}) & t\geq0,\\
\overline{C}-t\overline{I}^{(1)}(\mathbf{0};-\Sigma_{0}\overline{p}) & t<0.
\end{cases}
\end{align*}
A nonnegative $t$ is feasible whenever $\underline{I}_{\infty}(t\overline{\mathbf{h}})\leq\overline{I}_{\infty}(t\overline{\mathbf{h}})$, i.e., 

\[
\underline{C}+t\underline{I}^{(1)}(\mathbf{0};\Sigma_{0}\overline{p})\leq \overline{C}+t\overline{I}^{(1)}(\mathbf{0};\Sigma_{0}\overline{p}),
\]
where recall $\overline{C}=-\underline{C}$ and  
$\overline{I}^{(1)}(\mathbf{0};\Sigma_{0}\overline{p})\leq\underline{I}^{(1)}(\mathbf{0};\Sigma_{0}\overline{p})$. Also, Lemma \ref{lem:key} shows that $\overline{I}^{(1)}(\mathbf{0};\Sigma_{0}\overline{p})>0$. Therefore, we find the upper bound of the values that $t$ can take is $\overline{t}$ after simple algebra. Then, the lower bound of the values that  $t$ can take must be $-\overline{t}$ due to centrosymmetry. 

\noindent \textbf{Step 3}: We show $R^{*}\geq(\overline{C}+\overline{x})\Phi\left(-\frac{\overline{x}}{\left(\overline{p}^{\top}\Sigma_{0}\overline{p}\right)^{1/2}}\right)$. Consider two cases. If $\overline{x}=0$, invoking  Step 2 of proof for statement (i) immediately yields the conclusion. If $0<\overline{x}\leq\bar{t}\overline{I}^{(1)}(\mathbf{0};\Sigma_{0}\overline{p}),$
we pick 
\[
t^{*}:=\frac{\overline{x}}{\overline{I}^{(1)}(\mathbf{0};\Sigma_{0}\overline{p})}\in\left(0,\bar{t}\right].
\]
Consider a prior $\pi$ in the space of $\left(h,\mu^{*}\right)$
that randomizes evenly between 
\[
\left\{ \left(t^{*}\overline{\mathbf{h}},\overline{C}+t^{*}\overline{I}^{(1)}(\mathbf{0};G_{0}\overline{\mathbf{h}})\right),\left(-t^{*}\overline{\mathbf{h}},\underline{C}+\underline{I}^{(1)}(\mathbf{0};-t^{*}G_{0}\overline{\mathbf{h}})\right)\right\} ,
\]
where note 
\begin{align*}
\overline{C}+t^{*}\overline{I}^{(1)}(\mathbf{0};G_{0}\overline{\mathbf{h}}) & =\overline{C}+t^{*}\overline{I}^{(1)}(\mathbf{0};\Sigma_{0}\overline{p})>0,\quad
\underline{C}+\underline{I}^{(1)}(\mathbf{0};-t^{*}G_{0}\overline{\mathbf{h}}))  =-\left(\overline{C}+t^{*}\overline{I}^{(1)}(\mathbf{0};G_{0}\overline{\mathbf{h}})\right)<0.
\end{align*}
The Bayes optimal rule with respect to $\pi$ is 
\[
\mathbf{1}\left\{ t^{*}\overline{\mathbf{h}}^{\top}\mathbf{I}_{0}\Delta\geq0\right\} =\mathbf{1}\left\{ \overline{p}^{\top}G_{0}\Delta\geq0\right\} ,
\]
with a corresponding Bayes value 
\begin{align*}
r^{*} & :=\left(\overline{C}+t^{*}\overline{I}^{(1)}(\mathbf{0};G_{0}\overline{\mathbf{h}})\right)\Phi\left(-t^{*}\left(\overline{\mathbf{h}}^{\top}I_{0}\overline{\mathbf{h}}\right)^{1/2}\right)\\
 & =\left(\overline{C}+t^{*}\overline{I}^{(1)}(\mathbf{0};\Sigma_{0}\overline{p})\right)\Phi\left(-t^{*}\left(\overline{p}^{\top}\Sigma_{0}\overline{p}\right)^{1/2}\right)\\
 & =\left(\overline{C}+t^{*}\overline{I}^{(1)}(\mathbf{0};\Sigma_{0}\overline{p})\right)\Phi\left(-\frac{t^{*}\overline{I}^{(1)}(\mathbf{0};\Sigma_{0}\overline{p})}{\left(\overline{p}^{\top}\Sigma_{0}\overline{p}\right)^{1/2}}\right)\\
 & =\left(\overline{C}+\overline{x}\right)\Phi\left(-\frac{\overline{x}}{\left(\overline{p}^{\top}\Sigma_{0}\overline{p}\right)^{1/2}}\right),
\end{align*}
where the second equality follows from plugging in $\overline{\mathbf{h}}=\mathbf{I}_{0}^{-1}G_{0}^{\top}\overline{p}$,
the third equality follows from Lemma \ref{lem:key}, and the final
equality follows from $t^{*}\overline{I}^{(1)}(\mathbf{0};\Sigma_{0}\overline{p})=\overline{x}$
by construction. The conclusion of this step follows because $R^{*}\geq r^{*}$. 

\noindent \textbf{Step 4}: Based on Steps 1-3 above, we conclude that: $R^{*}=\left(\overline{C}+\overline{x}\right)\Phi\left(-\frac{\overline{x}}{\left(\overline{p}^{\top}\Sigma_{0}\overline{p}\right)^{1/2}}\right)\geq\overline{C}/2$, with equality if and only if $\overline{x}=0$, 
and $\mathbf{1}\left\{ \overline{p}^{\top}G_{0}\Delta\geq0\right\} $
is indeed MMR optimal for the limit game \eqref{eq:limit.game.1}.

\subsubsection*{Proof of Statement (iii)}
\begin{proof}
First, the existence and uniqueness of $\mu^{o}$ is established in Lemma \ref{lem:non-diff-continuous}. Then, Lemma \ref{lem:thm.non.diff.case.3.1} shows that in case (iii), a solution of $\eqref{pf:R.star}$, i.e., a MMR optimal rule that depends on $\Delta$ only via $\varDelta$ is 
 $\mathbf{1}\left\{ \left(\mu^{o}\right)^{\top}\Sigma_{0}^{-1}\varDelta\geq0\right\}$, 
where $\mu^{o}$  solves \eqref{eq:non.benign}.   Next, given $\mu^{o}$, let  $h^{o}:=\mathbf{I}_0^{-1}G_0^{\top}\Sigma_0^{-1}\mu^{o}$. By construction,  $G_{0}h^{o}=\mu^{o}$.  In the space of $(h,\mu^*)$, consider a  prior
that randomizes evenly between 
\[
\left\{ \left(h^{o},\overline{C}+\overline{I}^{(1)}(\mathbf{0};G_{0}h^{o})\right),\left(-h^{o},\underline{C}+\underline{I}^{(1)}(\mathbf{0};-G_{0}h^{o})\right)\right\},
\]
for which the Bayes optimal rule is $\mathbf{1}\left\{ \left(h^{o}\right)^{\top}\mathbf{I}_{0}\Delta\geq0\right\} =\mathbf{1}\left\{ (\mu^{o})^{\top}\Sigma_{0}^{-1}\varDelta\geq0\right\}$, with the corresponding Bayes value 
$\left(\overline{C}+\overline{I}^{(1)}(\mathbf{0};G_{0}h^{o})\right)\Phi\left(-\sqrt{\left(h^{o}\right)^{\top}\mathbf{I}_{0}h^{o}}\right)$, which equals $r(\mu^{o};\Sigma_0)$ after simple algebra. Conclude that $R^{*}=r(\mu^{o};\Sigma_0)$, and $\mathbf{1}\left\{ \left(\mu^{o}\right)^{\top}\Sigma_{0}^{-1}\varDelta\geq0\right\}$ achieves  $R^*$.
\end{proof}

\subsection{Proof of Theorem \ref{thm:non-diff-asymptotic}}
Lemma \ref{lem:non-diff-asymptotic-master} shows that $d_F$ is matched with the right limiting optimal rule under conditions of Theorem \ref{thm:non-diff-asymptotic}. The conclusion follows directly by invoking Theorem \ref{thm: general}.

\section{Other Forms of  Shrinking Identified Set}\label{sec:global.ID}

In our main text, we constructed a particular sequence of shrinking identified set in the form of \eqref{eq:local.id.U} and \eqref{eq:local.id.U.2}, leading to a limiting identified set \eqref{eq:limit.ID} that only incorporates the local geometry of the original identified set $I(\cdot)$ around $\mu_0$.  There are other reasonable ways to construct a shrinking identified set that may lead to a different geometry in the limit, which may retain more global features of the original  identified set. We provide an alternative (arguably more general) approach to constructing the shrinking identified set below and illustrate the differences by revisiting the example in Section \ref{sec:IK}.\footnote{We thank Tim Armstrong for comments that motivated this section.} 

Following the setup in Section \ref{sec:frame}, consider the point-identified parameter following a sequence $m(\theta_n)=\gamma_0+\frac{h}{\sqrt{n}}$ for $h\in\mathbb{R}^{d_m}$. As $n\rightarrow\infty$,  the identified set of $U(\theta_n)$ will also be approximately
\begin{align}\label{eq:global.1}
I\left(\mu\left(\gamma_{0}+\frac{h}{\sqrt{n}}\right)\right) & \approx\left[\underline{I}\left(\mu_{0}+\frac{G_{0}h}{\sqrt{n}}\right),\overline{I}\left(\mu_{0}+\frac{G_{0}h}{\sqrt{n}}\right)\right].
\end{align}
Note our \eqref{eq:local.id.U} just applies first-order approximation to the bounds on the RHS of \eqref{eq:global.1}. Instead of applying \eqref{eq:local.id.U}-\eqref{eq:local.id.U.2}, consider a general drifting identified set $I_n(\cdot)=[\underline{I}_n(\cdot),\overline{I}_n(\cdot)]$ indexed by sample size $n$,  such that for all local deviations $t/\sqrt{n}$ from $\mu_0$, we have $\underline{I}_n\left(\mu_{0}+t/\sqrt{n}\right)\rightarrow0$, $\overline{I}_n\left(\mu_{0}+t/\sqrt{n}\right)\rightarrow 0$,  and 
\begin{equation}\label{eq:global.1}
\sqrt{n}\underline{I}_n\left(\mu_{0}+t/\sqrt{n}\right)=\underline{I}_{\infty}(\mu_0+t), \sqrt{n}\overline{I}_n\left(\mu_{0}+t/\sqrt{n}\right)= \overline{I}_{\infty}(\mu_0+t),
\end{equation}
where both $\underline{I}_{\infty}(\mu_0+t)$ and  $\overline{I}_{\infty}(\mu_0+t)$ are finite numbers.  Then, we construct the drifting identified set as
\begin{align}\label{eq:global.2}
U(\theta_n)\in&\left[\underline{I}_n\left(\mu_{0}+G_{0}h/\sqrt{n}\right),\overline{I}_n\left(\mu_{0}+G_{0}h/\sqrt{n}\right)\right],
\end{align}
implying 
\begin{align*}\label{eq:global.3}
\sqrt{n}U(\theta_n)\in&\left[\sqrt{n}\underline{I}_n\left(\mu_{0}+G_{0}h/\sqrt{n}\right),\sqrt{n}\overline{I}_n\left(\mu_{0}+G_{0}h/\sqrt{n}\right)\right]\\
= &\left[\underline{I}_\infty\left(\mu_{0}+G_{0}h\right),\overline{I}_\infty\left(\mu_{0}+G_{0}h\right)\right].
\end{align*}
It follows that we can still characterize the limit decision problem as \eqref{eq:limit.game.1}, except that the limit identified set is re-defined as
\begin{equation}\label{eq:global.3}
I_{\infty}(h):=\left[\underline{I}_\infty\left(\mu_{0}+G_{0}h\right),\overline{I}_\infty\left(\mu_{0}+G_{0}h\right)\right],    
\end{equation}
and $M_h$ is re-defined analogously as long as \eqref{eq:global.3} is not empty. Depending on how $I_n(\cdot)$ is constructed, \eqref{eq:global.3}  may retain some global features of the original identified set. However, the limiting decision problem also becomes more difficult to solve  due to the more complicated geometry of the limit identified set in general.

One case that allows tractability is when we impose Assumption \ref{asm:centro}. Then, note  $\mu_0=\mathbf{0}$ and the limit identified set simplifies to 
\begin{equation}\label{eq:global.4}
I_{\infty}(h):=\left[\underline{I}_\infty\left(G_{0}h\right),\overline{I}_\infty\left(G_{0}h\right)\right].  
\end{equation}
The limit parameter space can still be convex and centrosymmetric, but with a limit identified set more general than that considered in our form \eqref{eq:limit.ID}. As a result, our theorems in Section \ref{sec:results} do not apply. However, one can still apply results in \cite{yata2021} to solve the limit problem, based on which an asymptotically optimal decision rule can be constructed.

To give a concrete example, consider the evidence aggregation application in Section \ref{sec:IK} with $0<C_1<C_2$, and $\mu_0=(0,0)$. Note the original identified set follows \eqref{eq:IK.original.bound}.  We can alternatively consider the following drifting identified set:
\begin{equation}\label{eq:ik.alternative.drift}
\underline{I}_{n}(\mu)=\max\left\{ \mu_1-C_{1,n},\mu_2-C_{2,n}\right\},
\overline{I}_{n}(\mu)=\min\left\{ \mu_1+C_{1,n},\mu_2+C_{2,n}\right\},   
\end{equation}
where  $C_{1,n}\rightarrow0$ and $C_{2,n}\rightarrow0$ as $n\rightarrow\infty$,  
$\sqrt{n}C_{1,n}=C_{1,\infty}$  and $\sqrt{n}C_{2,n}=C_{2,\infty}$ for $0<C_{1,\infty}<C_{2,\infty}<\infty$. Applying \eqref{eq:global.1}-\eqref{eq:global.2} to \eqref{eq:ik.alternative.drift}, we have that, for all $t=(t_1,t_2)$,
\[
\underline{I}_{n}(t/\sqrt{n})=\max\left\{ t_1/\sqrt{n}-C_{1,n},t_2/\sqrt{n}-C_{2,n}\right\},
\overline{I}_{n}(t/\sqrt{n})=\min\left\{ t_1/\sqrt{n}+C_{1,n},t_1/\sqrt{n}+C_{2,n}\right\},
\]
yielding the following limit identified set (cf. \eqref{eq:ik.limit.bound.1}):
\begin{equation}\label{eq:ik.global.limit.set}
\underline{I}_\infty\left(\mu\right)=\max\left\{ \mu_1-C_{1,\infty},\mu_2-C_{2,\infty}\right\},\overline{I}_\infty\left(\mu\right)=\min\{\mu_1+C_{1,\infty},\mu_2+C_{2,\infty}\}.
\end{equation}
Note how the geometry of \eqref{eq:ik.global.limit.set} is more analogous to that of the original identified set. The associated decision problem can be solved by applying results in \cite{yata2021}. Moreover, note the limit identified set \eqref{eq:ik.limit.bound.1} corresponds to a case of \eqref{eq:ik.alternative.drift} when $C_{1,n}\rightarrow0$, $C_{2,n}\rightarrow0$ (or $C_{2,n}$ is fixed) as  $n\rightarrow\infty$, but with $\sqrt{n}C_{1,n}=C_{1,\infty}<\infty$  and $\sqrt{n}C_{2,n}\rightarrow\infty$.   

\bibliographystyle{ecta}
\bibliography{references}

\newpage

\part*{Online Appendix }\label{sec:OA}

\section{Remedial Asymptotics When Assumption \ref{asm:2} is Violated}\label{sec:non.standard}
Consider the application in Section \ref{sec:contaminate} with an unknown rate of contamination in the population. Then, the reduced-form parameter becomes $\Gamma:=(\gamma_1,\gamma_0,p_1,p_0)$, so that   $m(\theta)=\Gamma$ and $U(\theta)=\mu^*$. The identified set in \eqref{eq:HM-bound-centro} is still valid; it depends  on $\Gamma$  only via $\mu(\Gamma)=(1-p_1)\gamma_1-(1-p_0)\gamma_0-(p_0-p_1)/2$ and $k(\Gamma)=(p_1+p_0)/2$. A global least favorable point would correspond to any $\Gamma$ such that $k(\Gamma)=1$, i.e.,  $p_1=p_0=1$ lies on the boundary of the parameter space and  is irregular. Then, \eqref{eq:HM-bound-centro} becomes $[-1,1]$, completely uninformative. 

The root cause of such irregularity is Nature's freedom to pick $k(\Gamma)$ so that the length of the identified set can be as large as possible. Our proposal is to slightly change the action space of  Nature to somewhat restore regularity. We propose to localize at any $\Gamma_0$ such that $\mu(\Gamma_0)=0$ and $k(\Gamma_0)=k_0$ for some $k_0\in(0,1)$. This can be interpreted as a ``second-best'' least favorable point: if $k(\Gamma_0)$ is fixed, then setting  $\mu(\Gamma_0)=0$ is indeed the most challenging case for the population game of \eqref{eq:population.mmr}. Next, we still consider a local parameter sequence for the reduced-form such that $m(\theta_n)=\Gamma_0+\frac{h}{\sqrt{n}}$. For the payoff-relevant parameter, we construct  a local parameter sequence $U(\theta_n)=\frac{\mu^*}{\sqrt{n}}$, 
where 
\begin{equation}\label{eq:id.non.standard}
\mu^*\in[G_0h-k_0,G_0h+k_0],   
\end{equation}
and $G_0$ is still defined as the gradient of $\mu(\Gamma)$ at the point of localization $\Gamma_0$. The limiting identified set \eqref{eq:id.non.standard} can be motivated as restricting  the action space of Nature: Suppose, once the point of localization is picked (so $k_0$ is fixed), Nature can only pick values of $U(\theta)$ in a drifting identified set 
\[\left[\mu\left(\Gamma_0+\frac{h}{\sqrt{n}}\right)-\frac{k_0}{\sqrt{n}}, \mu\left(\Gamma_0+\frac{h}{\sqrt{n}}\right)+\frac{k_0}{\sqrt{n}}\right].\]
Then, the identified set for $\sqrt{n}U(\theta)$ in the limit  would match \eqref{eq:id.non.standard}. In the above construction, the catch is that $k_0$ is completely fixed, implying for each sample size, Nature cannot arbitrarily enlarge the length of the identified set. 
% In other words, the local parameter $h$ can not affect the identified set via varying its width. 
With such modification, we can see that the limiting game becomes a standard one: the DM observes a Gaussian signal $\Delta\sim\mathcal{N}(h,\mathbf{I}_0^{-1})$ and the identified set for the limiting partially identified parameter is \eqref{eq:id.non.standard}. Applying Theorem \ref{thm:full-diff}, the MMR optimal rule in the limit game is $\mathbf{1}\{G_0\Delta\geq0\}$ if $\sigma:=(G_0\mathbf{I}^{-1}_0G_0^{\top})^{1/2}\geq\sqrt{2/\pi}k_0$ and to take fractional decision $\Phi\left(G_0\Delta/(\frac{2}{\pi}k_0^2-\sigma^2)^{1/2}\right)$ otherwise. To find a feasible matching rule, note in this case, even if $k_0$ is unknown, we can replace it with a consistent estimator, e.g., $k(\hat{\Gamma})$, analogous to the treatment of an unknown $\sigma$ in general. The asymptotic validity of such feasible rules can still be established by appropriately modifying the proofs of Theorem \ref{thm:full-diff-asymptotic}. It remains to see how general this approach can be to deal with scenarios when Assumption \ref{asm:2} fails, which we leave for future research.

\section{Additional Results: Full Differentiability}

\begin{lem}
\label{lem:full.diff.lem.1} Suppose $\kappa_0<1$. Under conditions of Theorem \ref{thm:full-diff}, 
\[
\overline{\upsilon}_{2}(t;\sigma):=\arg\max_{\upsilon\in\left[-\overline{C},\infty\right)}R_{2}(t,\upsilon;\sigma)
\]
 and 
\[
\overline{\upsilon}_{1}(t;\sigma):=\arg\max_{\upsilon\in\left[\frac{\overline{C}-\underline{C}}{\kappa_{0}-1},-\frac{\underline{C}}{\kappa_{0}}\right]}R_{1}(t,\upsilon;\sigma).
\]
are single-valued continuous functions in both $t$ and $\sigma$.
Furthermore, for each $\sigma>0$, there exists some unique $t\in\mathbb{R}$
such that 
\begin{align}
R_{2}(t,\overline{\upsilon}_{2}(t;\sigma);\sigma) & =R_{1}(t,\overline{\upsilon}_{1}(t;\sigma);\sigma),\label{eq:t.unique-1}
\end{align}
which implicitly defines a single-valued continuous function $t_{0}(\sigma):\mathbb{R}^{+}\rightarrow\mathbb{R}$
from the values of $\sigma$ to a unique $t\in\mathbb{R}$.
\end{lem}

\begin{proof}
Observe the following:
\begin{itemize}
\item[1)] For each $t\in\mathbb{R}$ and $\sigma>0$, $R_{2}(t,\upsilon;\sigma)$
is increasing then decreasing, and first concave then convex as $\upsilon$
ranges from $-\overline{C}$ to $+\infty$;
\item[2)] Analogously, for each $t\in\mathbb{R}$ and $\sigma>0$, as $\upsilon$
ranges from $-\infty$ to $-\frac{\underline{C}}{\kappa_{0}}$, $R_{1}(t,\upsilon;\sigma)$
is increasing then decreasing, and first convex then concave.
\end{itemize} 
As a result, $\max_{\upsilon\in\left[-\overline{C},\infty\right)}R_{2}(t,\upsilon;\sigma)$
is uniquely characterized by the FOC
\[
\frac{\phi\left(\frac{\upsilon-t}{\sigma}\right)}{1-\Phi\left(\frac{\upsilon-t}{\sigma}\right)}=\frac{\sigma}{(\upsilon+\overline{C})},
\]
which implicitly defines a single-valued function 
\begin{align*}
\overline{\upsilon}_{2}(t;\sigma) & :=\arg\max_{\upsilon\in\left[-\overline{C},\infty\right)}R_{2}(t,\upsilon;\sigma).
\end{align*}
Implicit function theorem implies that $\overline{\upsilon}_{2}(t;\sigma)$
is differentiable (and thus continuous) in both $t$ and $\sigma$.
Analogously, 
\[
\max_{\upsilon\in\left[\frac{\overline{C}-\underline{C}}{\kappa_{0}-1},-\frac{\underline{C}}{\kappa_{0}}\right]}R_{1}(t,\upsilon;\sigma)
\]
is uniquely characterized by either FOC 
\[
\frac{\phi\left(\frac{t-\upsilon}{\sigma}\right)}{1-\Phi\left(\frac{t-\upsilon}{\sigma}\right)}=-\frac{\kappa_{0}\sigma}{\kappa_{0}\upsilon+\underline{C}},
\]
or is at the boundary point $\frac{\overline{C}-\underline{C}}{\kappa_{0}-1}$,
which also defines a single-valued function:
\[
\overline{\upsilon}_{1}(t;\sigma):=\arg\max_{\upsilon\in\left[\frac{\overline{C}-\underline{C}}{\kappa_{0}-1},-\frac{\underline{C}}{\kappa_{0}}\right]}R_{1}(t,\upsilon;\sigma).
\]
By Berge's theorem of maximization, $\overline{\upsilon}_{1}(t;\sigma)$
is continuous in both $t$ and $\sigma$ as well. Therefore, conclude
that (cf., Lemma 1 of \cite{tetenov2012statistical}) for each $\sigma>0$:

\begin{itemize}
\item[1)] $R_{2}(t,\overline{\upsilon}_{2}(t;\sigma);\sigma)$ is continuous and strictly increasing
in $t$, $R_{2}(t,\overline{\upsilon}_{2}(t;\sigma);\sigma)\rightarrow0$ as $t\rightarrow-\infty$, and $R_{2}(t,\overline{\upsilon}_{2}(t;\sigma);\sigma)\rightarrow+\infty$ as $t\rightarrow+\infty$;
\item[2)] $R_{1}(t,\overline{\upsilon}_{1}(t;\sigma);\sigma)$ is continuous and decreasing
in $t$, $R_{1}(t,\overline{\upsilon}_{1}(t;\sigma);\sigma)$ approaches some finite positive constant as $t\rightarrow-\infty$, and  $R_{1}(t,\overline{\upsilon}_{1}(t;\sigma);\sigma)\rightarrow 0$ as $t\rightarrow+\infty$.
\end{itemize}
Conclude that there exists some unique $t\in\mathbb{R}$ such that
\[
R_{2}(t,\overline{\upsilon}_{2}(t;\sigma);\sigma)=R_{1}(t,\overline{\upsilon}_{1}(t;\sigma);\sigma),
\]
which implicitly defines a single-valued continuous function $t_{0}(\sigma):\mathbb{R}^{+}\rightarrow\mathbb{R}$
from the values of $\sigma$ to a unique $t\in\mathbb{R}$.
\end{proof}
\begin{lem}
\label{lem:full.diff.lem.2} Under conditions of Lemma  \ref{lem:full.diff.lem.1}, suppose $\sigma>\overline{\sigma}$. Then,  $\overline{\upsilon}_{2}(t_{0}(\sigma);\sigma)>\overline{\upsilon}_{1}(t_{0}(\sigma);\sigma)$.
\end{lem}

\begin{proof}
Let $f(\sigma):=\overline{\upsilon}_{2}(t_{0}(\sigma);\sigma)-\overline{\upsilon}_{1}(t_{0}(\sigma);\sigma)$.
Note $f(\sigma)$ is a continuous function because Lemma \ref{lem:full.diff.lem.1}
established that both $\overline{\upsilon}_{2}(\cdotp;\cdotp)$ and
$\overline{\upsilon}_{1}(\cdotp;\cdotp)$ are continuous, and that
$t_{0}(\cdotp)$ is also continuous.

\noindent\textbf{Step 1}: Consider  $\sigma=\overline{\sigma}$, in which case  $f(\overline{\sigma})=0$,
as we can verify that
$t_{0}(\overline{\sigma})=t^{*},\overline{\upsilon}_{2}(t_{0}(\overline{\sigma});\overline{\sigma})=\overline{\upsilon}_{1}(t_{0}(\overline{\sigma});\overline{\sigma})=0$.

\noindent \textbf{Step 2}: Consider $\sigma$ sufficiently large, in which case we have $\overline{\upsilon}_{2}(t_{0}(\sigma);\sigma)>\overline{\upsilon}_{1}(t_{0}(\sigma);\sigma)$.
To see this, note for $\sigma$ sufficiently large, 
\begin{align*}
\frac{\partial R_{2}(0,\upsilon;\sigma)}{\partial\upsilon} & \mid_{\upsilon=0}=\Phi\left(0\right)-\frac{\overline{C}}{\sigma}\phi\left(0\right)>0,
\frac{\partial R_{1}(0,\upsilon;\sigma)}{\partial\upsilon}  \mid_{\upsilon=0}=-\left[\kappa_{0}\Phi\left(0\right)+\frac{\underline{C}}{\sigma}\phi\left(0\right)\right]<0,
\end{align*}
implying $\overline{\upsilon}_{2}(0;\sigma)>0$
and $\overline{\upsilon}_{1}(0;\sigma)<0$ if $\sigma$ is sufficiently
large. Furthermore, note for each $\upsilon>0$: 
\begin{align*}
R_{2}(0,\upsilon;\sigma) & =(\upsilon+\overline{C})\Phi\left(\frac{-\upsilon}{\sigma}\right),
R_{1}(0,-\upsilon;\sigma) =(\kappa_{0}\upsilon-\underline{C})\Phi\left(\frac{-\upsilon}{\sigma}\right).
\end{align*}
As $\kappa_{0}<1$ and $\overline{C}>-\underline{C}$, conclude that
$R_{2}(0,\upsilon;\sigma)>R_{1}(0,-\upsilon;\sigma)$. This implies
that, whenever $\sigma$ is sufficiently large,
\[
R_{2}(0,\overline{\upsilon}_{2}(0;\sigma);\sigma)\geq R_{2}(0,-\overline{\upsilon}_{1}(0;\sigma);\sigma)>R_{1}(0,\overline{\upsilon}_{1}(0;\sigma);\sigma),
\]
which, in turn, means $t_{0}(\sigma)<0$ for $\sigma$ sufficiently
large. Now, note (using $t_{0}(\sigma)<0$ and $-\kappa_{0}\overline{C}+\underline{C}<0$)
\begin{align*}
 & \frac{\partial R_{1}(t_{0}(\sigma),\upsilon;\sigma)}{\partial\upsilon}\mid_{\upsilon=-\overline{C}} \\
 =&-\kappa_{0}\Phi\left(\frac{-\overline{C}-t_{0}(\sigma)}{\sigma}\right)-\frac{(-\kappa_{0}\overline{C}+\underline{C})}{\sigma}\phi\left(\frac{-\overline{C}-t_{0}(\sigma)}{\sigma}\right)\\
 < &-\kappa_{0}\Phi\left(\frac{-\overline{C}}{\sigma}\right)+\frac{(\kappa_{0}\overline{C}-\underline{C})}{\sigma}\phi\left(0\right).
\end{align*}
Conclude that for $\sigma$ sufficiently large, we also have $\frac{\partial R_{1}(t_{0}(\sigma),\upsilon;\sigma)}{\partial\upsilon}\mid_{\upsilon=-\overline{C}}<0$.
Therefore, for $\sigma$ sufficiently large, we have 
$\overline{\upsilon}_{1}(t_{0}(\sigma);\sigma)<-\overline{C}\leq\overline{\upsilon}_{2}(t_{0}(\sigma);\sigma)$,
implying $f(\sigma)>0$ for $\sigma$ sufficiently large. 

\textbf{Step 3}: Suppose, on the contrary, it were
the case that $f(\sigma)\leq0$ for some $\sigma>\bar{\sigma}$. Since $f(\sigma)$ is continuous, 
there must exist some $\widetilde{\sigma}>\overline{\sigma}$ such
that 
\[
\overline{\upsilon}_{1}(t_{0}(\widetilde{\sigma});\widetilde{\sigma})=\overline{\upsilon}_{2}(t_{0}(\widetilde{\sigma});\widetilde{\sigma}).
\]
In this case, let $\tilde{x}:=\overline{\upsilon}_{2}(t_{0}(\widetilde{\sigma});\widetilde{\sigma})=\overline{\upsilon}_{1}(t_{0}(\widetilde{\sigma});\widetilde{\sigma})$.
Then,  analogous to Steps 3-4 in the proof of  Theorem \ref{thm:full-diff}(ii), we can establish an equilibrium for the game (\ref{eq:one.dim.limit.game.smooth})
in which a least favorable prior randomizes between $\left(\tilde{x},\tilde{x}+\overline{C}\right)$
and $\left(\tilde{x},\kappa_{0}\tilde{x}+\underline{C}\right)$, with the corresponding  value of the game equal 
\[
-\frac{\left(\tilde{x}+\overline{C}\right)\left(\kappa_{0}\tilde{x}+\underline{C}\right)}{(1-\kappa_{0})\tilde{x}+\overline{C}-\underline{C}}.
\]
Since randomizing between $\left(\tilde{x},\tilde{x}+\overline{C}\right)$
and $\left(\tilde{x},\kappa_{0}\tilde{x}+\underline{C}\right)$ is least
favorable, it must be the case that 
\[
-\frac{\left(\tilde{x}+\overline{C}\right)\left(\kappa_{0}\tilde{x}+\underline{C}\right)}{(1-\kappa_{0})\tilde{x}+\overline{C}-\underline{C}}\geq-\frac{\left(x+\overline{C}\right)\left(\kappa_{0}x+\underline{C}\right)}{(1-\kappa_{0})x+\overline{C}-\underline{C}}
\]
for all $x$ in a neighborhood of $\tilde{x}$, which implies a FOC
that pins down $\tilde{x}=0$, i.e., $\overline{\upsilon}_{2}(t_{0}(\widetilde{\sigma});\widetilde{\sigma})=\overline{\upsilon}_{1}(t_{0}(\widetilde{\sigma});\widetilde{\sigma})=0$.
This means that $\left(t_{0}(\widetilde{\sigma}),\widetilde{\sigma}\right)$
would need to satisfy the following three conditions for the equilibrium:
\begin{itemize}
\item[a)] $\mathbb{E}_{\hat{\upsilon}\sim \mathcal{N}(0,\widetilde{\sigma}^{2})}\left[\mathbf{1}\left\{ \hat{\upsilon}\geq t_{0}(\widetilde{\sigma})\right\} \right]=\Phi\left(-\frac{t_{0}(\widetilde{\sigma})}{\widetilde{\sigma}}\right)=\frac{\overline{C}}{\overline{C}-\underline{C}}$,

\item[b)] $\frac{\partial}{\partial\upsilon}R_{2}(t_{0}(\widetilde{\sigma}),\upsilon;\widetilde{\sigma})\mid_{\upsilon=0}=\Phi\left(\frac{t_{0}(\widetilde{\sigma})}{\widetilde{\sigma}}\right)-\frac{\overline{C}}{\widetilde{\sigma}}\phi\left(\frac{t_{0}(\widetilde{\sigma})}{\widetilde{\sigma}}\right)=0$,

\item[c)] $\frac{\partial}{\partial\upsilon}R_{1}(t_{0}(\widetilde{\sigma}),\upsilon;\widetilde{\sigma})\mid_{\upsilon=0}=-\left(\kappa_{0}\Phi\left(\frac{-t_{0}(\widetilde{\sigma})}{\widetilde{\sigma}}\right)+\frac{\underline{C}}{\widetilde{\sigma}}\phi\left(\frac{-t_{0}(\widetilde{\sigma})}{\widetilde{\sigma}}\right)\right)=0$,
\end{itemize}
which, after some algebra, implies
\[
\widetilde{\sigma}=\frac{-\overline{C}\left(\overline{C}-\underline{C}\right)}{\underline{C}}\phi\left(\Phi^{-1}\left(\frac{\overline{C}}{\overline{C}-\underline{C}}\right)\right)=\overline{\sigma}.
\]
But this contradicts the assumption that $\widetilde{\sigma}>\overline{\sigma}$.
\end{proof}

\begin{lem}\label{lem:matching.diff}
Suppose the conditions of Theorem \ref{thm:full-diff-asymptotic} hold. If $\sigma<\overline{\sigma}$, $d_{F}$ is
matched with $d_{\infty,RT}$ in the sense of \eqref{eq:matching.rule};  If
$\sigma>\overline{\sigma}$, $d_{F}$ is matched with $\mathbf{1}\left\{ \hat{\upsilon}\geq t_{0}\right\}$. 
% Lemma \ref{lem:DQM}
\end{lem}

\begin{proof}
Let $P_{h}^{n}$ be the probability measure when
$\gamma=\gamma_{0}+\frac{h}{\sqrt{n}}$. Denote by $\overset{P_{h}^{n}}{\rightarrow}$
convergence in probability under $P_{h}^{n}$.

\noindent\textbf{Step 1}: Preparation. Under stated assumptions,
note the following results hold:
\begin{itemize}
\item[1)] Le Cam's first lemma implies that $\hat{\sigma}\overset{P_{h}^{n}}{\rightarrow}\sigma$
for all $h$. 

\item[2)]  Since Lemma \ref{lem:full.diff.lem.1}
establishes that $t_{0}(\cdotp)$ is a continuous function, 
conclude that 
\[
\hat{t}=t_{0}(\hat{\sigma})\overset{P_{h}^{n}}{\rightarrow}t_{0}(\sigma)=t_{0}
\]
for all $h$ as well. 

\item[3)] Note $\gamma_0$ is such that $\mu(\gamma_0)=\mu_0$.  Then,  Le Cam's third lemma implies 
\[
\sqrt{n}(\mu(\hat{\gamma})-\mu_0)\overset{h}{\rightsquigarrow}\mathcal{N}(G_{0}h,G_{0}I_{0}^{-1}G_{0}^{\top}).\]

\item[4)] It follows by continuous mapping theorem that 
\[
\sqrt{n}\left(\overline{I}^{(1)}(\mu_{0})\right)^{\top}(\mu(\hat{\gamma})-\mu_0)\overset{h}{\rightsquigarrow}\mathcal{N}\left(\left(\overline{I}^{(1)}(\mu_{0})\right)^{\top}G_{0}h,\sigma^{2}\right),
\]
\[
\mathbf{1}\left\{ \sqrt{n}\left(\overline{I}^{(1)}(\mu_{0})\right)^{\top}(\mu(\hat{\gamma})-\mu_0)\geq\hat{t}\right\} \overset{h}{\rightsquigarrow}\mathbf{1}\left\{ \mathcal{N}\left(\left(\overline{I}^{(1)}(\mu_{0})\right)^{\top}G_{0}h,\sigma^{2}\right)\geq t_{0}\right\} ,
\]
and if $\sigma<\overline{\sigma}$, 
\[
\frac{\sqrt{n}\left(\overline{I}^{(1)}(\mu_{0})\right)^{\top}(\mu(\hat{\gamma})-\mu_0)-t^{*}}{\sqrt{\overline{\sigma}^{2}-\hat{\sigma}^{2}}}\overset{h}{\rightsquigarrow}\mathcal{N}\left(\frac{\left(\overline{I}^{(1)}(\mu_{0})\right)^{\top}G_{0}h-t^{*}}{\sqrt{\overline{\sigma}^{2}-\sigma^{2}}},\frac{\sigma^{2}}{\overline{\sigma}^{2}-\sigma^{2}}\right).
\]
\end{itemize}

\noindent \textbf{Step 2}: We show that if
$\sigma>\overline{\sigma}$, $d_{F}$ is matched with $\mathbf{1}\left\{ \hat{\upsilon}\geq t_{0}\right\}$. By Slutsky's theorem, we have
$\mathbf{1}\left\{ \hat{\sigma}\geq\overline{\sigma}\right\} \overset{P_{h}^{n}}{\rightarrow}1$ and that $\mathbf{1}\left\{ \hat{\sigma}<\overline{\sigma}\right\} \overset{P_{h}^{n}}{\rightarrow}0$ as $n\rightarrow\infty$. 
Step 1 already establishes 
\[
d_{F,\hat{t}}\overset{h}{\rightsquigarrow}\mathbf{1}\left\{ \mathcal{N}\left(\left(\overline{I}^{(1)}(\mu_{0})\right)^{\top}G_{0}h,\sigma^{2}\right)\geq t_{0}\right\} .
\]
As $d_{F,RT}\leq1$, conclude that in case of $\sigma>\overline{\sigma}$, we have
\begin{align*}
d_{F} & =\mathbf{1}\left\{ \hat{\sigma}\geq\overline{\sigma}\right\} \cdot d_{F,\hat{t}}+\mathbf{1}\left\{ \hat{\sigma}<\overline{\sigma}\right\} \cdot d_{F,RT} \overset{h}{\rightsquigarrow}\mathbf{1}\left\{ \mathcal{N}\left(\left(\overline{I}^{(1)}(\mu_{0})\right)^{\top}G_{0}h,\sigma^{2}\right)\geq t_{0}\right\} .
\end{align*}
Furthermore, note $d_{F}$ is uniformly integrable. Thus, we have
\[
\mathbb{E}\left[d_{F}\right]{\rightarrow}\mathbb{E}\left[\mathbf{1}\left\{ \mathcal{N}\left(\left(\overline{I}^{(1)}(\mu_{0})\right)^{\top}G_{0}h,\sigma^{2}\right)\geq t_{0}\right\} \right]
\]
for each $h$ as $n\rightarrow\infty$. Therefore, if $\sigma>\overline{\sigma}$, $d_{F}$
is matched with rule
\begin{equation}
\mathbf{1}\left\{ \mathcal{N}\left(\left(\overline{I}^{(1)}(\mu_{0})\right)^{\top}G_{0}h,\sigma^{2}\right)\geq t_{0}\right\} \label{eq:full.diff.asymptotics}
\end{equation}
in the sense of Lemma \ref{lem:DQM}. Note $\hat{\upsilon}\sim N\left(\left(\overline{I}^{(1)}(\mu_{0})\right)^{\top}G_{0}h,\sigma^{2}\right)$,
implying that $\mathbf{1}\left\{ \hat{\upsilon}\geq t_{0}\right\}$
shares the same distribution as \eqref{eq:full.diff.asymptotics}.
Conclude that $d_{F}$ is matched with rule $\mathbf{1}\left\{ \hat{\upsilon}\geq t_{0}\right\}$ if $\sigma>\overline{\sigma}$.
The proof for $\sigma<\overline{\sigma}$ is completely analogous.
\end{proof}

\begin{lem}
\label{lem:full.diff.lem.3} Suppose the conditions of Theorem \ref{thm:full-diff-asymptotic} hold. If $\sigma=\overline{\sigma}$, $d_{F}$ is matched with $\mathbf{1}\left\{ \hat{\upsilon}\geq t_{0}\right\} $
in the sense of \eqref{eq:matching.rule}.
\end{lem}

\begin{proof}
When $\sigma=\overline{\sigma}$, we need to show that, for all $h$,
we have
$\mathbb{E}[d_{F}]\rightarrow\mathbb{E}\left[\mathbf{1}\left\{ \hat{\upsilon}\geq t_{0}\right\} \right]$ as $n\rightarrow\infty$,
where $\hat{\upsilon}\sim\mathcal{N}\left(\left(\overline{I}^{(1)}(\mu_{0})\right)^{\top}G_{0}h,\sigma^{2}\right)$.
As $d_{F}=\mathbf{1}\left\{ \hat{\sigma}\geq\overline{\sigma}\right\} \cdot d_{F,\hat{t}}+\mathbf{1}\left\{ \hat{\sigma}<\overline{\sigma}\right\} \cdot d_{F,RT}$ is uniformly integrable, it suffices to show that, for
every $h$, 
\begin{align*}
  d_{F}
\overset{h}{\rightsquigarrow} & X_{0}:=\mathbf{1}\left\{ \mathcal{N}\left(\left(\overline{I}^{(1)}(\mu_{0})\right)^{\top}G_{0}h,\sigma^{2}\right)\geq t_{0}\right\} .
\end{align*}
\textbf{Step 1}: The distribution function of $X_{0}$ is 
\[
X_{0}=\begin{cases}
0, & \text{w/ prob.}\text{ }\Phi\left(\frac{t_{0}-\left(\overline{I}^{(1)}(\mu_{0})\right)^{\top}G_{0}h}{\sigma}\right),\\
1, & \text{w/ prob.}\text{ }\Phi\left(\frac{\left(\overline{I}^{(1)}(\mu_{0})\right)^{\top}G_{0}h-t_{0}}{\sigma}\right).
\end{cases}
\]
Thus, by definition of convergence in distribution, it suffices to
show that, for each $h$, we have, as $n\rightarrow\infty$,
\begin{align*}
Pr\left\{ d_{F}\leq x\right\} \rightarrow & \Phi\left(\frac{t_{0}-\left(\overline{I}^{(1)}(\mu_{0})\right)^{\top}G_{0}h}{\sigma}\right)
\end{align*}
for every $x\in(0,1)$. Note the desired convergence trivially holds
when $x<0$ and $x>1$.

\noindent\textbf{Step 2}: Pick any $x\in(0,1)$, we can write 
\begin{align*}
Pr\left\{ d_{F}\leq x\right\} = & \mathbb{E}\left[\mathbf{1}\left\{ d_{F}\leq x\right\} \right]\\
= & \mathbb{E}\left[\mathbf{1}\left\{ d_{F}\leq x\right\} \mathbf{1}\left\{ \hat{\sigma}\geq\overline{\sigma}\right\} \right]+\mathbb{E}\left[\mathbf{1}\left\{ d_{F}\leq x\right\} \mathbf{1}\left\{ \hat{\sigma}<\overline{\sigma}\right\} \right]\\
= & \mathbb{E}\left[\mathbf{1}\left\{ d_{F,\hat{t}}\leq x\right\} \mathbf{1}\left\{ \hat{\sigma}\geq\overline{\sigma}\right\} \right]+\mathbb{E}\left[\mathbf{1}\left\{ d_{F,RT}\leq x\right\} \mathbf{1}\left\{ \hat{\sigma}<\overline{\sigma}\right\} \right]\\
= & A_{n}(h)+B_{n}(h),
\end{align*}
where 
\begin{align*}
A_{n}(h):= & \mathbb{E}\left[\mathbf{1}\left\{ d_{F,\hat{t}}\leq x\right\} \right],\\
B_{n}(h):= & \mathbb{E}\left[\mathbf{1}\left\{ d_{F,RT}\leq x\right\} \mathbf{1}\left\{ \hat{\sigma}<\overline{\sigma}\right\} \right]-\mathbb{E}\left[\mathbf{1}\left\{ d_{F,\hat{t}}\leq x\right\} \mathbf{1}\left\{ \hat{\sigma}<\overline{\sigma}\right\} \right]
\end{align*}
Note,
\begin{align*}
A_{n} & (h)=Pr\left\{ d_{F,\hat{t}}\leq x\right\} =Pr\left\{ d_{F,\hat{t}}=0\right\} =Pr\left\{ \sqrt{n}\left(\overline{I}^{(1)}(\mu_{0})\right)^{\top}(\mu(\hat{\gamma})-\mu_0)<\hat{t}\right\} \\
 & \rightarrow\Phi\left(\frac{t_{0}-\left(\overline{I}^{(1)}(\mu_{0})\right)^{\top}G_{0}h}{\sigma}\right),
\end{align*}
as $n\rightarrow\infty$ for all $h$, since Step 1 of the proof for
Lemma \ref{lem:matching.diff} established that 
\[
\mathbf{1}\left\{ \sqrt{n}\left(\overline{I}^{(1)}(\mu_{0})\right)^{\top}(\mu(\hat{\gamma})-\mu_0)\geq\hat{t}\right\} \overset{h}{\rightsquigarrow}X_{0}.
\]
\textbf{Step 3}: Therefore, it suffices to show that, for each $h$,
$B_{n}(h)\rightarrow0$ as $n\rightarrow\infty$. To this end, write
\begin{align*}
g_{n,1}(\hat{\gamma},\hat{\sigma}) & :=\sqrt{n}\left(\overline{I}^{(1)}(\mu_{0})\right)^{\top}(\mu(\hat{\gamma})-\mu_{0})-t^{*}-\Phi^{-1}\left(x\right)\mathbf{1}{\left\{ \hat{\sigma}<\overline{\sigma}\right\}}\sqrt{\overline{\sigma}^{2}-\hat{\sigma}^{2}},\\
g_{n,2}(\hat{\gamma},\hat{t}) & :=\sqrt{n}\left(\overline{I}^{(1)}(\mu_{0})\right)^{\top}(\mu(\hat{\gamma})-\mu_{0})-\hat{t},\\
G & :=\mathcal{N}\left(\left(\overline{I}^{(1)}(\mu_{0})\right)^{\top}G_{0}h,\sigma^{2}\right)-t_{0}.
\end{align*}
Note $\Phi^{-1}\left(x\right)\mathbf{1}{\left\{ \hat{\sigma}<\overline{\sigma}\right\}}\sqrt{\overline{\sigma}^{2}-\hat{\sigma}^{2}}\overset{P_{h}^{n}}{\rightarrow}0,\hat{t}\overset{P_{h}^{n}}{\rightarrow}t_{0}$ and $t^{*}=t_{0}$ when $\overline{\sigma}=\sigma$. It follows by
continuous mapping theorem that $g_{n,1}(\hat{\gamma},\hat{\sigma})\overset{h}{\rightsquigarrow}G$,
$g_{n,2}(\hat{\gamma},\hat{t})\overset{h}{\rightsquigarrow}G$, and $g_{n,1}(\hat{\gamma},\hat{\sigma})-g_{n,2}(\hat{\gamma},\hat{t})\overset{P_{h}^{n}}{\rightarrow}0$.
Therefore,
\begin{align*}
\vert B_{n}(h)\vert= & \vert\mathbb{E}\left[\mathbf{1}\left\{ d_{F,RT}\leq x\right\} \mathbf{1}\left\{ \hat{\sigma}<\overline{\sigma}\right\} \right]-\mathbb{E}\left[\mathbf{1}\left\{ d_{F,\hat{t}}\leq x\right\} \mathbf{1}\left\{ \hat{\sigma}<\overline{\sigma}\right\} \right]\vert\\
= & \vert\mathbb{E}\left[\mathbf{1}\left\{ g_{n,1}(\hat{\gamma},\hat{\sigma})\leq0\right\} \mathbf{1}\left\{ \hat{\sigma}<\overline{\sigma}\right\} \right]-\mathbb{E}\left[\mathbf{1}\left\{ g_{n,2}(\hat{\gamma},\hat{t})<0\right\} \mathbf{1}\left\{ \hat{\sigma}<\overline{\sigma}\right\} \right]\vert\\
\leq & \mathbb{E}\left|\mathbf{1}\left\{ g_{n,1}(\hat{\gamma},\hat{\sigma})\leq0\right\} -\mathbf{1}\left\{ g_{n,2}(\hat{\gamma},\hat{t})<0\right\} \right|\\
= & Pr\left\{ \mathbf{1}\left\{ g_{n,1}(\hat{\gamma},\hat{\sigma})\leq0\right\} \neq\mathbf{1}\left\{ g_{n,2}(\hat{\gamma},\hat{t})<0\right\} \right\}. 
\end{align*}
Thus, it suffices to show that, as $n\rightarrow\infty$,
\[
Pr\left\{ \mathbf{1}\left\{ g_{n,1}(\hat{\gamma},\hat{\sigma})\leq0\right\} \neq\mathbf{1}\left\{ g_{n,2}(\hat{\gamma},\hat{t})<0\right\} \right\} \rightarrow0.
\]
To this end, pick any $\varepsilon>0$. Note 
\[
Pr\left\{ \mathbf{1}\left\{ g_{n,1}(\hat{\gamma},\hat{\sigma})\leq0\right\} \neq\mathbf{1}\left\{ g_{n,2}(\hat{\gamma},\hat{t})<0\right\} \right\} \leq B_{n,1}(h,\varepsilon)+B_{n,2}(h,\varepsilon),
\]
where 
\begin{align*}
B_{n,1}(h,\varepsilon) & :=Pr\left\{ \left|g_{n,1}(\hat{\gamma},\hat{\sigma})-g_{n,2}(\hat{\gamma},\hat{t})\right|>\varepsilon\right\} ,\quad
B_{n,2}(h,\varepsilon)  :=Pr\left\{ \left|g_{n,1}(\hat{\gamma},\hat{\sigma})\right|\leq\varepsilon\right\} .
\end{align*}
Note $B_{n,1}(h,\varepsilon)\rightarrow0$,
and due to Portmanteau's lemma, $B_{n,2}(h,\varepsilon)\rightarrow Pr\left\{ \left|G\right|\leq\varepsilon\right\} $. Then, $B_{n}(h)\rightarrow0$
follows by letting $\varepsilon\downarrow0$ and noting $Pr\left\{ G=0\right\} =0$.

\end{proof}

\section{Additional Results: Directional Differentiability}

\begin{lem}\label{lem:global.lfp}
Suppose Assumptions \ref{asm:1}-\ref{asm:centro} hold. Then, 
$\mu_{0}=\mathbf{0}$ is a global least favorable point.
\end{lem}
\begin{proof}
Note Assumptions \ref{asm:2} and \ref{asm:centro} imply that $\mathbf{0}\in int(M_\mu)$ is such that $\overline{I}(\mathbf{0})>0>\underline{I}(\mathbf{0})$
and $\overline{R}(\mathbf{0})=\frac{\overline{I}(\mathbf{0})}{2}$.
We  show that $\overline{R}(\mathbf{0})\geq\overline{R}(\mathbf{\mu})$
for all $\mu\in M_{\mu}$ such that $\mu\neq\mathbf{0}$. To see this,
note it suffices to check all $\mu\in M_{\mu}$ such that $\overline{I}(\mathbf{\mu})>0>\underline{I}(\mu)$, for which
\begin{align*}
\overline{R}(\mathbf{\mu}) & =\frac{1}{2}\frac{2}{\left(\frac{1}{\left[-\underline{I}(\mu)\right]}+\frac{1}{\bar{I}(\mu)}\right)} \leq\frac{1}{2}\left(\frac{\bar{I}(\mu)+\left[-\underline{I}(\mu)\right]}{2}\right) \leq\frac{1}{2}\left(\overline{I}\left(\frac{\mu-\mu}{2}\right)\right)=\frac{\overline{I}(\mathbf{0})}{2},
\end{align*}
where the first inequality follows from the fact that the harmonic mean
is always no larger than the arithmetic mean for positive numbers, and
the second inequality follows from concavity of $\overline{I}(\cdotp)$
and $-\underline{I}(\mu)=\overline{I}(-\mu)$.
\end{proof}

\begin{lem}
\label{lem:key}Suppose the conditions of Theorem \ref{thm:non-diff} hold. Consider
the following problem:
\begin{equation}
\max_{\mathbf{v}\in\mathbb{R}^{d_{\mu}}:\mathbf{v}^{\top}\mathbf{v}=1}\overline{I}^{(1)}(\mathbf{0};\mathbf{v})\left(\mathbf{v}^{\top}\Sigma_{0}^{-1}\mathbf{v}\right)^{-1/2}.\label{eq:key-1}
\end{equation}
A solution of (\ref{eq:key-1}) is $\overline{\mathbf{v}}=\frac{\Sigma_{0}\overline{p}}{\left\Vert \Sigma_{0}\overline{p}\right\Vert }$,
where $\overline{p}$ is the unique solution of $\min_{p\in\partial\overline{I}(\mathbf{0})}\left\{ p^{\top}\Sigma_{0}p\right\} ^{1/2}$.
Moreover, $\overline{I}^{(1)}(\mathbf{0};\overline{\mathbf{v}})\left(\mathbf{\overline{v}}^{\top}\Sigma_{0}^{-1}\mathbf{\mathbf{\overline{v}}}\right)^{-1/2}=\left\{ \overline{p}^{\top}\Sigma_{0}\overline{p}\right\} ^{1/2}>0$. 
\end{lem}
\begin{proof}
Let $f(\mathbf{v}):=\overline{I}^{(1)}(\mathbf{0};\mathbf{v})\left(\mathbf{v}^{\top}\Sigma_{0}^{-1}\mathbf{v}\right)^{-1/2}$,
which is continuous in $\left\{ \mathbf{v}\in\mathbb{R}^{d_{\mu}}:\mathbf{v}^{\top}\mathbf{v}=1\right\} $.
Thus, $\mathbf{V}^{*}:=\max_{\mathbf{v}\in\mathbb{R}^{d_{\mu}}:\mathbf{v}^{\top}\mathbf{v}=1}f(\mathbf{v})$
is finite. Moreover, note $\mathbf{v}^{\top}\Sigma_{0}^{-1}\mathbf{v}>0$
for all $\mathbf{v}^{\top}\mathbf{v}=1$, among which there exists
some $\mathbf{v}$ such that $\overline{I}^{(1)}(\mathbf{0};\mathbf{v})>0$
due to Assumption \ref{asm:centro}(ii). Thus, $\mathbf{V}^{*}>0$
is attained at some $\mathbf{v}^{\top}\mathbf{v}=1$ such that $\overline{I}^{(1)}(\mathbf{0};\mathbf{v})>0$.
Furthermore, as $f(\alpha\mu)=f(\mu)$ for all $\alpha>0$ and all
$\mu\neq\mathbf{0}$, $\mathbf{V}^{*}$ depends only on the direction
of $\mathbf{v}$ and not on its length. Therefore, $\mathbf{V}^{*}=\max_{\mathbf{v}^{\top}\Sigma_{0}^{-1}\mathbf{v}=1}\overline{I}^{(1)}(\mathbf{0};\mathbf{v}).$

First, we show that 
\begin{equation}
\max_{\mathbf{v}^{\top}\Sigma_{0}^{-1}\mathbf{v}=1}\overline{I}^{(1)}(\mathbf{0};\mathbf{v})=\max_{\mathbf{v}^{\top}\Sigma_{0}^{-1}\mathbf{v}\leq1}\overline{I}^{(1)}(\mathbf{0};\mathbf{v}).\label{eq:pf.key.equal}
\end{equation}
 Clearly, $\max_{\mathbf{v}^{\top}\Sigma_{0}^{-1}\mathbf{v}\leq1}\overline{I}^{(1)}(\mathbf{0};\mathbf{v})\geq\max_{\mathbf{v}^{\top}\Sigma_{0}^{-1}\mathbf{v}=1}\overline{I}^{(1)}(\mathbf{0};\mathbf{v})>0$.
Suppose, on the contrary, the first inequality is strict.  Then, there exists some $\widetilde{\mathbf{v}}$ such that $\left(\widetilde{\mathbf{v}}\right)^{\top}\Sigma_{0}^{-1}\widetilde{\mathbf{v}}<1$
and
\begin{equation}
\overline{I}^{(1)}(\mathbf{0};\widetilde{\mathbf{v}})>\max_{\mathbf{v}^{\top}\Sigma_{0}^{-1}\mathbf{v}=1}\overline{I}^{(1)}(\mathbf{0};\mathbf{v})>0.\label{eq:pf.key.contra}
\end{equation}
It follows that $\widetilde{\mathbf{v}}\neq0$ and $\left(\widetilde{\mathbf{v}}\right)^{\top}\Sigma_{0}^{-1}\widetilde{\mathbf{v}}>0$.
Then, let $\widetilde{\alpha}=\left(1/\left(\widetilde{\mathbf{v}}\right)^{\top}\Sigma_{0}^{-1}\widetilde{\mathbf{v}}\right)^{1/2}>1$.
We have $\left(\widetilde{\alpha}\widetilde{\mathbf{v}}\right)^{\top}\Sigma_{0}^{-1}\left(\widetilde{\alpha}\widetilde{\mathbf{v}}\right)=1$,
and 
\[
\overline{I}^{(1)}(\mathbf{0};\widetilde{\alpha}\widetilde{\mathbf{v}})=\widetilde{\alpha}\overline{I}^{(1)}(\mathbf{0};\mathbf{\widetilde{\mathbf{v}}})>\overline{I}^{(1)}(\mathbf{0};\mathbf{\widetilde{\mathbf{v}}}),
\]
which contradicts (\ref{eq:pf.key.contra}). Conclude that the equality
must hold. 

Then, note the RHS of (\ref{eq:pf.key.equal}) can be recast as a
max-min problem 
\begin{equation}
\max_{\mathbf{v}\in\mathbb{R}^{d_{\mu}}:\mathbf{v}^{\top}\Sigma_{0}^{-1}\mathbf{v}\leq1}\min_{p\in\partial \overline{I}(\mathbf{0})}p^{\top}\mathbf{v},\label{eq:key-2}
\end{equation}
where both $\left\{ \mathbf{v}\in\mathbb{R}^{d_{\mu}}:\mathbf{v}^{\top}\Sigma_{0}^{-1}\mathbf{v}\leq1\right\} $
and $\partial\overline{I}(\mathbf{0})$ are convex
and compact. Moreover, $p^{\top}\mathbf{v}$ is continuous and affine in both $p$ and $\mathbf{v}$.  It follows by Sion's minimax theorem \citep[Theorem 3]{simons1995minimax}
that 
\begin{equation}
\max_{\mathbf{v}\in\mathbb{R}^{d_{\mu}}:\mathbf{v}^{\top}\Sigma_{0}^{-1}\mathbf{v}\leq1}\min_{p\in\partial \overline{I}(\mathbf{0})}p^{\top}\mathbf{v}=\min_{p\in\partial\overline{I}(\mathbf{0})}\max_{\mathbf{v}\in\mathbb{R}^{d_{\mu}}:\mathbf{v}^{\top}\Sigma_{0}^{-1}\mathbf{v}\leq1}p^{\top}\mathbf{v}.\label{eq:key-min-max}
\end{equation}
For each $p\in\partial\overline{I}(\mathbf{0})$, Cauchy-Schwarz inequality
implies
\begin{align*}
\max_{\mathbf{v}\in\mathbb{R}^{d_{\mu}}:\mathbf{v}^{\top}\Sigma_{0}^{-1}\mathbf{v}\leq1}p^{\top}\mathbf{v} & \leq\max_{\mathbf{v}\in\mathbb{R}^{d_{\mu}}:\mathbf{v}^{\top}\Sigma_{0}^{-1}\mathbf{v}\leq1}\left(p^{\top}\Sigma_{0}p\right)^{1/2}\left(\mathbf{v}^{\top}\Sigma_{0}^{-1}\mathbf{v}\right)^{1/2}\leq\left(p^{\top}\Sigma_{0}p\right)^{1/2},
\end{align*}
where $p^{\top}\Sigma_{0}p>0$ as $\Sigma_{0}$ is positive definite
and $p\neq\mathbf{0}$. As $\frac{\Sigma_{0}p}{\left(p^{\top}\Sigma_{0}p\right)^{1/2}}$
achieves $\left(p^{\top}\Sigma_{0}p\right)^{1/2}$, conclude that
\[
\max_{\mathbf{v}\in\mathbb{R}^{d_{\mu}}:\mathbf{v}^{\top}\Sigma_{0}^{-1}\mathbf{v}\leq1}p^{\top}\mathbf{v}=\left(p^{\top}\Sigma_{0}p\right)^{1/2}.
\]
It follows that the min-max problem in (\ref{eq:key-min-max}) reduces
to 
\begin{equation}
\min_{p\in\partial\overline{I}(\mathbf{0})}\left\{ p^{\top}\Sigma_{0}p\right\} ^{1/2},\label{eq:key-4}
\end{equation}
which has a unique solution referred to as $\overline{p}$ since (\ref{eq:key-4})
is equivalent to $\min_{p\in\partial\overline{I}(\mathbf{0})}p^{\top}\Sigma_{0}p$,
which  has a strictly convex objective function. Conclude that
\[
\mathbf{V}^{*}=\left( \overline{p}^{\top}\Sigma_{0}\overline{p}\right)^{1/2}=\left(\overline{p}^{\top}\Sigma_{0}\overline{p}\right)^{-1/2}\overline{I}^{(1)}(\mathbf{0};\Sigma_{0}\overline{p}),
\]
and a solution of \eqref{eq:key-1} is $\frac{\Sigma_{0}\overline{p}}{\left\Vert \Sigma_{0}\overline{p}\right\Vert }$.
\end{proof}

\begin{lem}\label{lem:non-diff-continuous}

Under the conditions of Theorem \ref{thm:non-diff}, $\overline{p}(\Sigma_0)$, $\mu^{o}(\Sigma_0)$ and $\overline{x}(\Sigma_0)$ are all continuous (single-valued) functions of $\Sigma_0$.
\end{lem}

\begin{proof}
\textbf{Continuity of $\overline{p}(\Sigma_0)$}. 
Note $\left( p^{\top}\Sigma_{0}p\right) ^{1/2}$ is continuous in both $\Sigma_0$ and $p$, and $\partial\overline{I}(\mathbf{0})$ is nonempty and  compact. Therefore, $\overline{p}(\Sigma_0)$, defined as the solution of $\min_{p\in\partial \overline{I}(\mathbf{0})}\left\{ p^{\top}\Sigma_{0}p\right\} ^{1/2}$ is upper-hemicontinuous, nonempty and compact-valued by Berge's theorem of maximization. As Lemma \ref{lem:key} established that  $\overline{p}(\Sigma_0)$ is unique, conclude that $\overline{p}(\Sigma_0)$ is in fact  a continuous function.

\noindent \textbf{Continuity of $\mu^{o}(\Sigma_0)$}. 
\noindent \textbf{Step 1}: We show that there exists some finite $\overline{\mathbf{t}}(\Sigma_0)\geq0$
such that 
\begin{equation}\label{eq:non.diff.continuity.1}
\max_{\left\{ \mu\in M_{\mu,\infty}:\overline{I}^{(1)}(\mathbf{0};\mu)\geq0\right\} }r(\mu;\Sigma_0)=\max_{\left\{ \mu\in M_{\mu,\infty}:\overline{I}^{(1)}(\mathbf{0};\mu)\geq0;\mu^{\top}\mu\leq\mathbf{\overline{t}}^{2}(\Sigma_0)\right\} }r(\mu;\Sigma_0).
\end{equation}
For each $\mu$, we can write $\mu=t\mathbf{v}$ for some $\mathbf{v}$ such that $\mathbf{v}^{\top}\mathbf{v}=1$
is the direction of $\mu$ and some $t\geq0$ is the length of $\mu$. Therefore,
\[
r(\mu;\Sigma_0)=r(t\mathbf{v};\Sigma_0)=\left(\overline{C}+t\overline{I}^{(1)}(\mathrm{\mathbf{0}};\mathbf{v})\right)\Phi\left(-t\sqrt{\mathbf{v}^{\top}\Sigma_0^{-1}{\mathbf{v}}}\right),
\]
and
\[
\max_{\{\mu\in M_{\mu,\infty}:\overline{I}^{(1)}(\mathbf{0};\mu)\geq0\}}r(\mu;\Sigma_0)=\max_{\left\{ \mathbf{v}^{\top}\mathbf{v}=1:\overline{I}^{(1)}(\mathrm{\mathbf{0}};\mathbf{v})\geq0\right\} }\max_{\left\{ t\geq0:t\mathbf{v}\in M_{\mu,\infty}\right\} }r(t\mathbf{v};\Sigma_0).
\]
For each $\mathbf{v}^{\top}\mathbf{v}=1$ such that $\overline{I}^{(1)}(\mathrm{\mathbf{0}};\mathbf{v})\geq0$,
note $r(t\mathbf{v};\Sigma_0)$ is continuous, either decreasing, or
first increasing and then decreasing in $t\geq0$. It follows that $\max_{\left\{ t\geq0\right\} }r(t\mathbf{v};\Sigma)
$
is either uniquely characterized by a FOC or at $t=0$, which also
defines a single-valued function $\overline{t}(\mathbf{v};\Sigma_0)$
continuous in both $\mathbf{v}$ and $\Sigma_0$. As the set $\left\{ \mathbf{v}^{\top}\mathbf{v}=1:\overline{I}^{(1)}(\mathrm{\mathbf{0}};\mathbf{v})\geq0\right\} $
is compact, there exists a finite-valued $\mathbf{\overline{t}}(\Sigma_0)$
such that 
\begin{equation}\label{eq:t.max}
\mathbf{\overline{t}}(\Sigma_0)=\sup_{\left\{ \mathbf{v}^{\top}\mathbf{v}=1:\overline{I}^{(1)}(\mathrm{\mathbf{0}};\mathbf{v})\geq0\right\} }\overline{t}(\mathbf{v};\Sigma_0).   
\end{equation}
Moreover, $\mathbf{\overline{t}}(\Sigma_0)$ is also continuous in $\Sigma_0$. These imply that the inner maximization problem must be achieved at some finitely valued $t\leq\overline{\mathbf{t}}(\Sigma_0)$ for each $\mathbf{v}^{\top}\mathbf{v}=1$ such that $\overline{I}^{(1)}(\mathrm{\mathbf{0}};\mathbf{v})\geq0$, which, in turn, implies that \eqref{eq:non.diff.continuity.1} must hold. 

\noindent \textbf{Step 2}: Conclusion from Step 1 implies that a solution of \eqref{eq:non.benign} can be found by restricting to the nonempty and compact set 
\begin{equation}\label{eq:non.diff.C.set}
\mathcal{C}(\Sigma_0):=\left\{ \mu\in M_{\mu,\infty}:\overline{I}^{(1)}(\mathbf{0};\mu)\geq0;\mu^{\top}\mu\leq\mathbf{\overline{t}}^{2}(\Sigma_0)\right\},
\end{equation}
which verifies the existence of $\mu^{o}$. Lemma
\ref{lem:non.diff.quasiconcave} shows that $r(\cdotp;\Sigma_0)$ is strictly quasiconcave in $\left\{ \mu\in M_{\mu,\infty}:\overline{I}^{(1)}(\mathbf{0};\mu)\geq0\right\}$, which is convex itself. This verifies the uniqueness of $\mu^{o}$. It follows that ${\mu}^{o}(\Sigma_0)=\arg\max_{\mu\in\mathcal{C}(\Sigma_0)}r(\mu;\Sigma_0)$. 

\noindent \textbf{Step 3:} We show $\mu^{o}(\Sigma_0)$ is continuous.  By Lemma \ref{lem:non-diff.correspondence},  $\mathcal{C}(\Sigma_0)$
is a continuous correspondence. Invoking Berge's theorem of maximization yields that
${\mu}^{o}(\Sigma_0)$ is upper-hemicontinuous. As Step 2 established that $\mu^{o}(\Sigma_0)$ is  single-valued, conclude that ${\mu}^{o}(\Sigma_0)$
is in fact a continuous function. 

\noindent \textbf{Continuity of $\overline{x}(\Sigma_0)$}. Note we showed that $\overline{p}=\overline{p}(\Sigma_0)$ is continuous in $\Sigma_0$. Then, the continuity of $\overline{x}(\Sigma_0)$ follows by applying Berge's theorem of maximization after compactifying $[0,\infty)$ with arguments analogous to the proof of continuity of $\mu^{o}(\Sigma_0)$ above. 
\end{proof}

\begin{lem}\label{lem:thm.non.diff.case.3.1}
In case (iii) of Theorem \ref{thm:non-diff}, a solution of (\ref{pf:R.star})
is $\mathbf{1}\left\{ \left(\mu^{o}\right)^{\top}\Sigma_{0}^{-1}\varDelta\geq0\right\}$, 
where $\mu^{o}$ solves \eqref{eq:non.benign}. 
\end{lem}
\begin{proof}
Fix some scalars $x$ and $t$ such that
\begin{equation}\label{eq:x}
x\in\left(0,\bar{t}\overline{I}^{(1)}(\mathbf{0};\Sigma_{0}\overline{p})\right), \quad t>\overline{\mathbf{t}}(\Sigma_0),
\end{equation}
where  $\overline{\mathbf{t}}(\Sigma_0)$ is defined in \eqref{eq:t.max}.  Denote by
\[
S:=\left\{ \mu\in M_{\mu,\infty}\mid\overline{I}^{(1)}(\mathbf{0};\mu)\geq x,\mu^{\top}\mu\leq t^2 \right\}.   
\]
Due to centrosymmetry, for each $\mu\in S$, $-\mu\in M_{\mu,\infty}$. Consider
solving the following MMR problem in a restricted parameter space when either $\mu$ or $-\mu$ can take values in $S$ (which we write as $\pm\mu\in S$):
\begin{equation}
\min_{d(\varDelta)}\sup_{\pm\mu\in S,\mu^{*}\in I_{\infty}(\mu)}\mu^{*}\left[\mathbf{1}\left\{ \mu^{*}\geq0\right\} -\mathbb{E}[d(\varDelta)]\right].\label{eq:prop3.1}
\end{equation}
Below, we take several steps to show that: (1) \eqref{eq:prop3.1} has a solution $\mathbf{1}\left\{ \left(\mu^{\star}\right)^{\top}\Sigma_{0}^{-1}\varDelta\geq0\right\}$, where
$\mu^{\star}\in S$ solves a fixed-point program;  (2) $\mathbf{1}\left\{ \left(\mu^{\star}\right)^{\top}\Sigma_{0}^{-1}\varDelta\geq0\right\} $
is also the MMR solution of (\ref{pf:R.star}); and (3) $\mu^\star$ can be found by solving   \eqref{eq:non.benign}. 

\noindent \textbf{Step 1}:  For each $\tilde{\mu}\in S$, consider a prior
$\pi_{\tilde{\mu}}$ that randomizes evenly between 
\[
\left\{ \left(\tilde{\mu},\overline{I}^{(1)}(\mathbf{0};\tilde{\mu})+\overline{C}\right),(-\tilde{\mu},\underline{I}^{(1)}(\mathbf{0};-\tilde{\mu})+\underline{C})\right\} ,
\]
where note $\underline{I}^{(1)}(\mathbf{0};-\tilde{\mu})+\underline{C}=-\overline{I}^{(1)}(\mathbf{0};\tilde{\mu})-\overline{C}$, and these two points belong to the parameter
space in (\ref{eq:prop3.1}). The Bayes optimal rule with respect
to $\pi_{\tilde{\mu}}$ is $\mathbf{1}\left\{ \tilde{\mu}^{\top}\Sigma_{0}^{-1}\varDelta\geq0\right\} $,
with a corresponding Bayes value 
$r(\tilde{\mu};\Sigma_0)$.
Furthermore, the best response of Nature with respect to $\mathbf{1}\left\{ \tilde{\mu}^{\top}\Sigma_{0}^{-1}\varDelta\geq0\right\} $
in the restricted parameter space is 
\begin{align*}
 & \sup_{\pm\mu\in S,\mu^{*}\in I_{\infty}(\mu),}\mu^{*}\left[\mathbf{1}\left\{ \mu^{*}\geq0\right\} -\Phi\left(\frac{\tilde{\mu}^{\top}\Sigma_{0}^{-1}\mu}{\left(\tilde{\mu}^{\top}\Sigma_{0}^{-1}\tilde{\mu}\right)^{1/2}}\right)\right]\\
= & \sup_{\mu\in S}\left(\overline{C}+\overline{I}^{(1)}(\mathbf{0};\mu)\right)\Phi\left(-\frac{\tilde{\mu}^{\top}\Sigma_{0}^{-1}\mu}{\left(\tilde{\mu}^{\top}\Sigma_{0}^{-1}\tilde{\mu}\right)^{1/2}}\right),
\end{align*}
where the second equality uses the centrosymmetry. Write 
\[
g(\mu,\widetilde{\mu}):=\left(\overline{C}+\overline{I}^{(1)}(\mathbf{0};\mu)\right)\Phi\left(-\frac{\tilde{\mu}^{\top}\Sigma_{0}^{-1}\mu}{\left(\tilde{\mu}^{\top}\Sigma_{0}^{-1}\tilde{\mu}\right)^{1/2}}\right).
\]
Consider the correspondence $\phi:S\rightrightarrows S$ defined as
\begin{equation}\label{eq:correspondence}
\arg\max_{\mu\in S}g(\mu,\widetilde{\mu}).   
\end{equation}
\noindent \textbf{Step 2}: We show that $\phi$ has a fixed point.  To see this, note:

\begin{itemize}
\item[(a)] $S$ is convex and compact by construction. 

\item[(b)]  Since $g(\cdotp;\cdotp)$ is continuous, and $S$ is compact and
nonempty, it follows by Berge's maximum theorem that the correspondence
$\phi(\cdotp)$ is upper hemicontinuous with nonempty and compact
values. It follows then $\phi$
 has a closed graph.

\item[(c)] Note $\phi(\cdotp)$ is nonempty. Furthermore, quasiconcavity
of $g(\cdotp,\widetilde{\mu})$ (Lemma \ref{lem:quasi.concavity})
implies that $\phi(\cdotp)$ is convex-valued. 
\end{itemize}
With claims (a)-(c) above, we verified that all conditions of Kakutani's fixed point
theorem hold. Conclude that $\phi:S\rightrightarrows S$ has
a fixed point $\mu^{\star}$, and \eqref{eq:prop3.1} has a solution $\mathbf{1}\left\{ \left(\mu^{\star}\right)^{\top}\Sigma_{0}^{-1}\varDelta\geq0\right\}$, where
$\mu^{\star}\in S$.

\noindent \textbf{Step 3}:  Conclusion from Step 2 implies that a solution
for (\ref{eq:prop3.1}) is $\mathbf{1}\left\{ \left(\mu^{\star}\right)^{\top}\Sigma_{0}^{-1}\varDelta\geq0\right\}$, which we now show 
is also the MMR solution of (\ref{pf:R.star}), i.e.,
\begin{equation}\label{eq:local.is.global}
\sup_{\mu\in M_{\mu,\infty},\overline{C}+\overline{I}^{(1)}(\mathbf{0};\mu)\geq0}g(\mu,\mu^{\star})=\sup_{\mu\in S}g(\mu,\mu^{\star})=g(\mu^\star,\mu^\star).  
\end{equation}
\noindent Suppose not. Then, there exists some $\widehat{\mu}$ not
in $S$ but in $\{\mu\in M_{\mu,\infty},\overline{C}+\overline{I}^{(1)}(\mathbf{0};\mu)\geq0\}$
such that $g(\widehat{\mu};\mu^{\star})>g(\mu^{\star};\mu^{\star})$. Then, for $\lambda\in(0,1)$, consider 
\[
\widehat{\mu}_{\lambda}:=\lambda\widehat{\mu}+(1-\lambda)\mu^{\star}\in\{\mu\in M_{\mu,\infty},\overline{C}+\overline{I}^{(1)}(\mathbf{0};\mu)\geq0\}.
\]
Algebra shows 
\begin{align*}
g(\widehat{\mu}_{\lambda};\mu^{\star}) & =\left(\overline{C}+\overline{I}^{(1)}(\mathbf{0};\widehat{\mu}_{\lambda})\right)\Phi\left(-\frac{\left(\mu^{\star}\right)^{\top}\Sigma_{0}^{-1}\widehat{\mu}_{\lambda}}{\left(\left(\mu^{\star}\right)^{\top}\Sigma_{0}^{-1}\mu^{\star}\right)^{1/2}}\right)\\
 & =\left(\overline{C}+\overline{I}^{(1)}(\mathbf{0};\lambda\widehat{\mu}+(1-\lambda)\mu^{\star})\right)\Phi\left(-\left(\left(\mu^{\star}\right)^{\top}\Sigma_{0}^{-1}\mu^{\star}\right)^{1/2}-\lambda\delta_{\widehat{\mu}}\right),
\end{align*}
where $\delta_{\widehat{\mu}}:=\frac{\left(\mu^{\star}\right)^{\top}\Sigma_{0}^{-1}\left(\widehat{\mu}-\mu^{\star}\right)}{\left(\left(\mu^{\star}\right)^{\top}\Sigma_{0}^{-1}\mu^{\star}\right)^{1/2}}$. Consider two cases:

\begin{enumerate}
\item If $\widehat{\mu}$ is such that $\delta_{\widehat{\mu}}=0$, then
notice $g(\widehat{\mu}_{\lambda};\mu^{\star})$ is concave in $\lambda$
(as $\overline{I}^{(1)}(\mathbf{0};\cdotp)$ is concave), implying
\begin{align*}
g(\widehat{\mu}_{\lambda};\mu^{\star}) & =\left(\overline{C}+\overline{I}^{(1)}(\mathbf{0};\lambda\widehat{\mu}+(1-\lambda)\mu^{\star})\right)\Phi\left(-\left(\left(\mu^{\star}\right)^{\top}\Sigma_{0}^{-1}\mu^{\star}\right)^{1/2}\right)\\
 & \geq\lambda g(\widehat{\mu};\mu^{\star})+(1-\lambda)g(\mu^{\star};\mu^{\star})\\
 & >g(\mu^{\star};\mu^{\star}),
\end{align*}

\noindent for all $\lambda\in(0,1)$ since we posited $g(\widehat{\mu};\mu^{\star})>g(\mu^{\star};\mu^{\star})$. This forms a contradiction. Indeed,
Lemma \ref{lem:boundary} shows that we must have $\overline{I}^{(1)}(\mathbf{0};\mu^{\star})>x$
and $(\mu^{\star})^{\top}\mu^{\star}<t^{2}$. Therefore, for $\lambda>0$
sufficiently close to zero, we have $\widehat{\mu}_{\lambda}\in S$.
As $\mu^{\star}$ maximizes $g(\cdotp;\mu^{\star})$ in $S$, we must have
that
\begin{equation}\label{eq:local.max}
g(\widehat{\mu}_{\lambda};\mu^{\star})\leq g(\mu^{\star};\mu^{\star})   
\end{equation}
for $\lambda>0$ sufficiently small.

\item  Suppose that $\widehat{\mu}$ is such that $\delta_{\widehat{\mu}}\neq0$.
Then, analogous to the proof of Lemma \ref{lem:non.diff.quasiconcave}, we can show that $g(\widehat{\mu}_{\lambda};\mu^{\star})$
is strictly quasiconcave in $\lambda$. Thus, we have
\[
g(\widehat{\mu}_{\lambda};\mu^{\star})>\min\left\{ g(\mu^{\star};\mu^{\star}),g(\widehat{\mu};\mu^{\star})\right\} =g(\mu^{\star};\mu^{\star})
\]
for all $\lambda\in(0,1)$. This analogously forms a contradiction
with \eqref{eq:local.max}. 
\end{enumerate}
Hence, \eqref{eq:local.is.global} must  hold, and we conclude that
$\mathbf{1}\left\{ \left(\mu^{\star}\right)^{\top}\Sigma_{0}^{-1}\varDelta\geq0\right\} $
is indeed an MMR optimal rule for \eqref{pf:R.star} and  $\mathbf{R}^*=r(\mu^{\star};\Sigma_0)$. 

\noindent\textbf{Step 4}: Note conclusion from Step 3 immediately implies that the prior that randomizes evenly between
\[
\left\{ \left(\mu^{\star},\overline{I}^{(1)}(\mathbf{0};\mu^{\star})+\overline{C}\right),\left(-\mu^{\star},\underline{I}^{(1)}(\mathbf{0};-\mu^{\star})+\underline{C}\right)\right\} 
\]
is least favorable, which must also solve  \eqref{eq:non.benign}, as we note $r(\mu;\Sigma_{0})$ is also the Bayes value for any two-point symmetric prior that randomizes evenly between
\[\left\{ \left(\mu,\overline{I}^{(1)}(\mathbf{0};\mu)+\overline{C}\right),\left(-\mu,\underline{I}^{(1)}(\mathbf{0};-\mu)+\underline{C}\right)\right\}\]
for any $\mu\in M_{\mu,\infty}$ such that $\overline{I}^{(1)}(\mathbf{0};\mu)\geq0$.
\end{proof}

\begin{lem}
\label{lem:boundary}Under conditions of Lemma \ref{lem:thm.non.diff.case.3.1},
the fixed point $\mu^{\star}$ of the correspondence defined in \eqref{eq:correspondence} is such that $\overline{I}^{(1)}(\mathbf{0};\mu^{\star})>x$ and $(\mu^{\star})^{\top}\mu^{\star}<t^2$,  where $x$ and $t$ are defined in \eqref{eq:x}.
\end{lem}
\begin{proof}
\textbf{Step 1}: we show $\overline{I}^{(1)}(\mathbf{0};\mu^{\star})>x$.  
The Bayes value for any two-point prior that randomizes evenly
between
\[
\left\{ \left(\widetilde{\mu},\overline{I}^{(1)}(\mathbf{0};\widetilde{\mu})+\overline{C}\right),\left(-\widetilde{\mu},\underline{I}^{(1)}(\mathbf{0};-\widetilde{\mu})+\underline{C}\right)\right\} ,
\]
where $\widetilde{\mu}\in S$, is $r(\widetilde{\mu};\Sigma_0)$. 
Step 2 in the proof of Lemma \ref{lem:thm.non.diff.case.3.1} implies  that 
\begin{equation}\label{pf:fixed.point.condition}
\mu^{\star}\in\arg\max_{\widetilde{\mu}\in S}r(\widetilde{\mu};\Sigma_0).
\end{equation}
We show that for any $\mu$ such that $\overline{I}^{(1)}(\mathbf{0};\mu)=x$,
\[
\mu\notin\arg\max_{\widetilde{\mu}\in S}r(\widetilde{\mu};\Sigma_0).
\]
To see this, consider 
\[
\max r(\widetilde{\mu};\Sigma_0)\quad s.t.\quad\overline{I}^{(1)}(\mathbf{0};\widetilde{\mu})=x,
\]
 which can be solved by considering 
\begin{align}\label{eq:min.mu.tilde}
\min & \sqrt{\tilde{\mu}^{\top}\Sigma_{0}^{-1}\tilde{\mu}}\quad
s.t. \quad \overline{I}^{(1)}(\mathbf{0};\tilde{\mu})=x.
\end{align}
As we can write $\tilde{\mu}=\widetilde{t}\widetilde{\mathbf{v}}$,
where $\widetilde{\mathbf{v}}$ denotes the direction of $\tilde{\mu}$,
and $\widetilde{t}$ is the length of $\widetilde{\mu}$, solving \eqref{eq:min.mu.tilde} is equivalent to solving
\begin{align}\label{eq:max.v.tilde}
\max_{\widetilde{\mathbf{v}}} \overline{I}^{(1)}(\mathbf{0};\widetilde{\mathbf{v}})\left(\widetilde{\mathbf{v}}^{\top}\Sigma_{0}^{-1}\widetilde{\mathbf{v}}\right)^{-1/2},
\end{align}
among all directions $\widetilde{\mathbf{v}}$ such that the corresponding length of $\widetilde{\mu}={\frac{x}{\overline{I}^{(1)}(\mathbf{0};\widetilde{\mathbf{v}})}}$ does not exceed the largest possible length that $\widetilde{\mu}$ can take along its direction. By Lemma  \ref{lem:key}, a unique solution  of \eqref{eq:max.v.tilde} among all directions is  $\overline{\mathbf{v}}$ (proportional to $\Sigma_0\overline{p}$). Along this direction, note indeed $\frac{x}{\overline{I}^{(1)}(\mathbf{0};\overline{\mathbf{v}})}$ does not exceed its largest possible length $\overline{t}\Vert\Sigma_0\overline{p}\Vert$  by construction. Conclude that
\[
\max_{\overline{I}^{(1)}(\mathbf{0};\tilde{\mu})=x}r(\widetilde{\mu};\Sigma_0)=r\left(\overline{\mathbf{v}}\frac{x}{\overline{I}^{(1)}(\mathbf{0};\overline{\mathbf{v}})};\Sigma_0\right)=\left(\overline{C}+x\right)\Phi\left(-\frac{x}{\left(\overline{p}^{\top}\Sigma_{0}\overline{p}\right)^{1/2}}\right)
\]
after lengthy algebra. However, note $f(\cdot):=(\overline{C}+\cdot)\Phi(-(\cdot)/(\overline{p}^{\top}\Sigma_{0}\overline{p})^{1/2})$ is strictly increasing in $(0,\overline{t}\overline{I}^{(1)}(\mathbf{0};\Sigma_0\overline{p{}})]$ under the conditions of this lemma and $x<\overline{t}\overline{I}^{(1)}(\mathbf{0};\Sigma_0\overline{p{}})$ by construction. Conclude that  $\mu^{\star}$ cannot be attained
at any point such that $\overline{I}^{(1)}(\mathbf{0};\mu)=x$.

\noindent \textbf{Step 2}: we show $(\mu^{\star})^{\top}\mu^{\star}<t^2$. Recall the fixed point $\mu^{\star}$ satisfies $\mu^{\star}\in\arg\max_{\widetilde{\mu}\in S}r(\widetilde{\mu};\Sigma_{0})$
(cf. \eqref{pf:fixed.point.condition}), where $S\subseteq\{\mu\in M_{\mu,\infty}:\overline{I}^{(1)}(\mathbf{0};\mu)\geq0\}$.
Meanwhile, proof of Lemma \ref{lem:non-diff-continuous} establishes
that $\mu^{o}=\arg\max_{\left\{ \mu\in M_{\mu,\infty}:\overline{I}^{(1)}(\mathbf{0};\mu)\geq0\right\} }r(\mu;\Sigma_{0})$
has a length not exceeding $\overline{\mathbf{t}}(\Sigma_{0})$. Suppose, on the contrary, $(\mu^{\star})^{\top}\mu^{\star}=t^{2}$, where
$t>\overline{\mathbf{t}}(\Sigma_{0})$. Then, we must have $\mu^{o}\notin S$, i.e., $\mu^{o}$ must be such that $0\leq\overline{I}^{(1)}(\mathbf{0};\mu^{o})<x$.
However, by analogous arguments used in Step 1, we can show that
$\mu^{o}$ cannot be attained at any point such that $0\leq\overline{I}^{(1)}(\mathbf{0};\mu)<x$ under conditions of this lemma,  a contradiction.
\end{proof}

\begin{lem}\label{lem:non-diff.correspondence}
Under conditions of Theorem \ref{thm:non-diff}, $\mathcal{C}(\Sigma_0)$  defined in \eqref{eq:non.diff.C.set} is a continuous correspondence.
\end{lem} 
\begin{proof}
\textbf{Step 1}: we show $\mathcal{C}(\Sigma_0)$ is upper hemi-continuous. Note
$\mathcal{C}(\Sigma_0)=\mathcal{C}_{1}\cap\mathcal{C}_{2}(\Sigma_0)$
where $\mathcal{C}_{1}:=\left\{ \mu\in M_{\mu,\infty}:\overline{I}^{(1)}(\mathbf{0};\mu)\geq0\right\} $
is closed, and 
\[
\mathcal{C}_{2}(\Sigma_0)=\left\{ \mu\in M_{\mu,\infty}:\mu^{\top}\mu\leq\mathbf{\overline{t}}^{2}(\Sigma_0)\right\} 
\]
is compact. Therefore, it suffices to show that $\mathcal{C}_{2}(\cdotp)$
has a closed graph. To this end, pick any sequences $\left\{ \Sigma_{n}\right\} ,\left\{ \mu_{n}\right\} $
such that $\Sigma_{n}$ is positive definite, and $\mu_{n}\in M_{\mu,\infty}$.
Suppose $\Sigma_{n}\rightarrow\Sigma_0$, $\mu_{n}\rightarrow\mu$ for
$\Sigma_0$ positive definite, and $\mu\in M_{\mu,\infty}$. Moreover, suppose
$\mu_{n}\in\mathcal{C}_{2}(\Sigma_{n})$ for each $n$, i.e., 
$\mu_{n}^{\top}\mu_{n}\leq\mathbf{\overline{t}}^{2}(\Sigma_{n})$. As $\mathbf{\overline{t}}^{2}(\cdotp)$ is continuous, we have $\mathbf{\overline{t}}^{2}(\Sigma_{n})\rightarrow\mathbf{\overline{t}}^{2}(\Sigma_0)$.
Moreover, $\mu_{n}^{\top}\mu_{n}\rightarrow\mu^{\top}\mu$. It follows
then we have $\mu^{\top}\mu\leq\mathbf{\overline{t}}^{2}(\Sigma_0)$,
i.e., $\mu\in\mathcal{C}_{2}(\Sigma_0)$. Conclude that $\mathcal{C}_{2}(\cdotp)$
has a closed graph. 

\noindent \textbf{Step 2}: We show that $\mathcal{C}(\Sigma_0)$ is lower hemi-continuous. Let
$\Sigma_{n}\rightarrow\Sigma_0$. Fix each $\mu\in\mathcal{C}(\Sigma_0)$. Consider the following two cases.
\begin{itemize}
\item[(1)] Suppose $\mu^{\top}\mu<\mathbf{\overline{t}}^{2}(\Sigma_0)$. As $\mathbf{\overline{t}}(\Sigma_{n})\rightarrow\mathbf{\overline{t}}(\Sigma_0)$,
we have $\mu^{\top}\mu<\mathbf{\overline{t}}^{2}(\Sigma_{n})$ for
$n$ sufficiently large. Note $\overline{I}^{(1)}(\mathbf{0};\mu)\geq0$
as well. Therefore, by setting $\mu_{n}=\mu$ for all $n$, we have
$\mu_{n}\in\mathcal{C}(\Sigma_{n})$ for $n$ sufficiently large and
$\mu_{n}\rightarrow\mu$ by construction. 
\item[(2)] Suppose $\mu^{\top}\mu=\mathbf{\overline{t}}^{2}(\Sigma_0)$.  In this
case, let $a_{n}:=\min\left\{ 1,\frac{\mathbf{\overline{t}}(\Sigma_{n})}{\sqrt{\mu^{\top}\mu}}\right\} $
and $\mu_{n}=a_{n}\mu$ (if $\mu=\mathbf{0}$, set $a_{n}=1$). As $\mathbf{\overline{t}}(\Sigma_{n})\rightarrow\mathbf{\overline{t}}(\Sigma_0)$
and $\mu^{\top}\mu=\mathbf{\overline{t}}^{2}(\Sigma_0)$, we have $a_{n}\rightarrow 1$
and $\mu_{n}\rightarrow\mu$ by construction. It suffices to show that $\mu_{n}\in\mathcal{C}(\Sigma_{n})$.
To this end, note $\mu_{n}=a_{n}\mu+(1-a_{n})\cdotp\mathbf{0}$. As
$\mathbf{0}\in M_{\mu,\infty}$, $\mu\in M_{\mu,\infty}$ and $M_{\mu,\infty}$ is convex,
we have $\mu_{n}\in M_{\mu,\infty}$. Moreover, $\overline{I}^{(1)}(\mathbf{0};\mu_{n})=a_{n}\overline{I}^{(1)}(\mathbf{0};\mu)\geq0$,
and 
\[
\mu_{n}^{\top}\mu_{n}=\min\left\{ 1,\frac{\mathbf{\overline{t}}^{2}(\Sigma_{n})}{\mu^{\top}\mu}\right\} \mu^{\top}\mu\leq\mathbf{\overline{t}}^{2}(\Sigma_{n}).
\]
Conclude that $\mu_{n}\in\mathcal{C}(\Sigma_{n})$ for each $n$. 
\end{itemize}
Based on the above two cases, we conclude that $\mathcal{C}(\Sigma)$
is lower hemi-continuous.
\end{proof}

\begin{lem}\label{lem:non.diff.quasiconcave}
 Under conditions of Theorem \ref{thm:non-diff}, $r(\cdotp;\Sigma_0)$
is strictly quasiconcave in \[\left\{ \mu\in M_{\mu,\infty}:\overline{I}^{(1)}(\mathbf{0};\mu)\geq0\right\}\]
for each positive definite $\Sigma_0$.    
\end{lem}

\begin{proof}
Note $\left\{ \mu\in M_{\mu,\infty}:\overline{I}^{(1)}(\mathbf{0};\mu)\geq0\right\} $
is convex.
Also, for each $\mu\in M_{\mu,\infty}$ such that $\overline{I}^{(1)}(\mathrm{\mathbf{0}};\mu)\geq0$,
we have 
\[
\log r(\mu;\Sigma_0)=\log\left(\overline{C}+\overline{I}^{(1)}(\mathrm{\mathbf{0}};\mu)\right)+\log\Phi\left(-\sqrt{\mu^{\top}\Sigma_0^{-1}\mu}\right).
\]
The first component of the above is concave. Below, we show the second
component is strictly concave. And the conclusion of the lemma follows
because a strictly monotone transformation of a strictly concave function
is strictly quasiconcave.  To this end, write $h(z)=\log\Phi(z)$, and $g(\mu)=\sqrt{\mu^{\top}\Sigma_0^{-1}\mu}$.
Note $h^{(1)}(z)=\frac{\phi(z)}{\Phi(z)}>0$, 
\begin{align*}
h^{(2)}(z) & =\frac{-z\phi(z)\Phi(z)-\phi^{2}(z)}{\Phi^{2}(z)}=-zh^{(1)}(z)-\left(h^{(1)}(z)\right)^{2}\\
 & =-h^{(1)}(z)\left(z+h^{(1)}(z)\right)<0
\end{align*}
for all $z\in\mathbb{R}$. Therefore, $h(\cdotp)$ is strictly concave
in $\mathbb{R}$. Moreover, note $g(\mu)$ is  convex in $\mathbb{R}$.
Conclude that $h(-g(\mu))$ is concave. To show that the function
is in fact strictly concave, pick any $\mu_{1}\neq\mu_{2}$ such that
$\overline{I}^{(1)}(\mathrm{\mathbf{0}};\mu_{1})\geq0$ and $\overline{I}^{(1)}(\mathrm{\mathbf{0}};\mu_{2})\geq0$. We consider three cases.

\noindent\textbf{Case 1}: Suppose one of $\mu_{1}$ and $\mu_{2}$ equals $\mathbf{0}$. Without
loss of generality, let $\mu_{1}=\mathbf{0}$ and $\mu_{2}\neq\mathbf{0}$. As $g$ is convex and $h$ is strictly increasing, we have
\begin{equation}\label{pf:strict.concave.1}
 h(-g(\alpha\mu_{1}+(1-\alpha)\mu_{2}))\geq h(\alpha\left(-g(\mu_{1})\right)+(1-\alpha)\left(-g(\mu_{2})\right)),   
\end{equation}
where also note $g(\mu_{1})\neq g(\mu_{2})$. Since $h$ is also strictly concave, conclude that 
\begin{equation}\label{pf:strict.concave.2}
h\left(\alpha\left(-g(\mu_{1})\right)+(1-\alpha)\left(-g(\mu_{2})\right)\right)>\alpha h\left(-g(\mu_{1})\right)+(1-\alpha)h\left(-g(\mu_{2})\right),
\end{equation}
establishing the desired strict inequality. 

\noindent \textbf{Case 2}: Suppose $\mu_{1}\neq\mu_{2}\neq\mathbf{0}$ but there exists some $\alpha_{0}\in(0,1)$ such that $\alpha_{0}\mu_{1}+(1-\alpha_{0})\mu_{2}=\mathbf{0}$. Then, $\mu_{1}=-\lambda_{0}\mu_{2}$ for $\lambda_{0}:=(1-\alpha_{0})/\alpha_{0}\in(0,\infty)$. If $\lambda_{0}\neq1$, then note $g(\mu_{1})\neq g(\mu_{2})$ still holds, implying that arguments via \eqref{pf:strict.concave.1} and \eqref{pf:strict.concave.2}  still go through. If $\lambda_{0}=1$, we have, for any $\alpha\in(0,1)$, $\left|1-2\alpha\right|<1$. Thus, 
\begin{equation}\label{pf:strict.concave.3}
-g(\alpha\mu_{1}+(1-\alpha)\mu_{2})	=-g((1-2\alpha)\mu_{2})=-\left|1-2\alpha\right|g(\mu_{2})>-g(\mu_{2}).    
\end{equation}
Moreover, since $\mu_{1}=-\mu_{2}$, $g(\mu_{1})=g(\mu_{2})$, implying
\begin{equation}\label{pf:strict.concave.4}
\alpha(-g(\mu_{1}))+(1-\alpha)(-g(\mu_{2}))=-g(\mu_{2}).    
\end{equation}
\eqref{pf:strict.concave.3} and \eqref{pf:strict.concave.4}  jointly imply $-g(\alpha\mu_{1}+(1-\alpha)\mu_{2})>\alpha(-g(\mu_{1}))+(1-\alpha)(-g(\mu_{2}))$. As a result,  the desired strict inequality still holds due to strict monotonicity of $h$. 

\noindent \textbf{Case 3}: Suppose $\mu_{1}\neq\mu_{2}\neq\mathbf{0}$ and their convex combination does not pass through $\mathbf{0}$. Then, it suffices to check the second order condition. To this end, note
$h^{(1)}(z)$ is the inverse mills ratio, with the property that $\ensuremath{h^{(1)}(z)+z>0}$
for all $\ensuremath{z\in\mathbb{R}}$. Algebra shows that, for any
$\mu\neq\mathbf{0}$, 
\begin{align*}
\frac{\partial h\left(-g(\mu)\right)}{\partial\mu} & =-h^{(1)}\left(-g(\mu)\right)\frac{\Sigma_0^{-1}\mu}{g(\mu)},\\
\frac{\partial^{2}h\left(-g(\mu)\right)}{\partial\mu\partial\mu^{\top}} & =h^{(2)}\left(-g(\mu)\right)\frac{\Sigma_0^{-1}\mu\mu^{\top}\Sigma_0^{-1}}{g^{2}(\mu)}+h^{(1)}\left(-g(\mu)\right)\left(\frac{\Sigma_0^{-1}\mu\mu^{\top}\Sigma_0^{-1}}{g^{3}(\mu)}-\frac{\Sigma_0^{-1}}{g(\mu)}\right)\\
 & =-h^{(1)}\left(-g(\mu)\right)\frac{\Sigma_0^{-1}}{g(\mu)}+\Sigma_0^{-1}\mu\mu^{\top}\Sigma_0^{-1}\left(\frac{h^{(2)}\left(-g(\mu)\right)}{g^{2}(\mu)}+\frac{h^{(1)}\left(-g(\mu)\right)}{g^{3}(\mu)}\right)\\
 & =-h^{(1)}\left(-g(\mu)\right)\frac{\Sigma_0^{-1}}{g(\mu)}+\Sigma_0^{-1}\mu\mu^{\top}\Sigma_0^{-1}\frac{h^{(1)}(-g(\mu))}{g(\mu)}\left(\frac{g(\mu)-h^{(1)}(-g(\mu))}{g(\mu)}+\frac{1}{g^{2}(\mu)}\right).
\end{align*}
As $g(\mu)>0$, $h^{(1)}\left(-g(\mu)\right)>0$, the definiteness
of the above matrix is determined by 
\begin{align*}
-\Sigma_0^{-1}+\Sigma_0^{-1}\mu\mu^{\top}\Sigma_0^{-1}\left(\frac{-g(\mu)\left(h^{(1)}(-g(\mu))-g(\mu)\right)+1}{g^{2}(\mu)}\right)= & A+B,
\end{align*}
where 
\begin{align*}
A & =\underset{<0}{\underbrace{-\frac{\left(h^{(1)}(-g(\mu))-g(\mu)\right)}{g(\mu)}}}\left(\Sigma_0^{-1}\mu\mu^{\top}\Sigma_0^{-1}\right),\\
B & =\left(\frac{\Sigma_0^{-1}\mu\mu^{\top}\Sigma_0^{-1}}{g^{2}(\mu)}-\Sigma_0^{-1}\right).
\end{align*}
Note $\Sigma_0^{-1}\mu\mu^{\top}\Sigma_0^{-1}$ is positive semidefinite,
with 
\[
a^{\top}\Sigma_0^{-1}\mu\mu^{\top}\Sigma_0^{-1}a=0
\]
for some $a\neq\mathbf{0}$ only when $a^{\top}\Sigma_0^{-1}\mu=0$,
i.e., $a$ is orthogonal to $\Sigma_0^{-1}\mu$; Moreover, $B$ is negative
semidefinite due to Cauchy Schwarz inequality, with $a^{\top}Ba=0$
for some $a\neq\mathbf{0}$ only when $a$ is proportional to $\mu$.
Therefore, for any $a\neq\mathbf{0}$, at most one of $a^{\top}Aa$
and $a^{\top}Ba$ can equal zero. Conclude that $A+B$ is in fact
negative definite for all $\mu\neq\mathbf{0}$. 
\end{proof}

\begin{lem}\label{lem:non.diff.quasi.concave.best.respond}
\label{lem:quasi.concavity} For each $\widetilde{\mu}\in S$ and each positive definite $\Sigma_0$, $g(\mu,\widetilde{\mu})$
is quasiconcave in 
\[\left\{ \mu\in M_{\mu,\infty}:\overline{I}^{(1)}(\mathbf{0};\mu)\geq0\right\}.\]
\end{lem}
\begin{proof}
For each $\widetilde{\mu}\in S$, consider
\begin{align}
\log\left(g(\mu,\widetilde{\mu})\right) & =\log\left(\overline{C}+\overline{I}^{(1)}(\mathbf{0};\mu)\right)+\log\left[\Phi\left(-\frac{\tilde{\mu}^{\top}\Sigma_{0}^{-1}\mu}{\left(\tilde{\mu}^{\top}\Sigma_{0}^{-1}\tilde{\mu}\right)^{1/2}}\right)\right].\label{eq:log.concave}
\end{align}
Then, the proof is analogous to Lemma \ref{lem:non.diff.quasiconcave} by noting that both components on the RHS of \eqref{eq:log.concave} are  concave. 
\end{proof}

\begin{lem}\label{lem:non-diff-asymptotic-master}
Under conditions of Theorem \ref{thm:non-diff-asymptotic}, the following statements hold true.

\begin{itemize}
\item[(i)] When $\left(\frac{\pi}{2}\overline{p}^{\top}\Sigma_{0}\overline{p}\right)^{1/2}<\overline{C}$, $d_{F}$ is matched with $d_{\infty,RT}$ in the sense of \eqref{eq:matching.rule};
\item[(ii)] When $\left(\frac{\pi}{2}\overline{p}^{\top}\Sigma_{0}\overline{p}\right)^{1/2}\geq\overline{C}$ and $\overline{x}\leq\overline{t}\overline{I}^{(1)}(\mathbf{0};\Sigma_{0}\overline{p})$, $d_{F}$ is matched with $d_{\infty,\overline{p}}$;

\item[(iii)] When $\left(\frac{\pi}{2}\overline{p}^{\top}\Sigma_{0}\overline{p}\right)^{1/2}\geq\overline{C}$
and $\overline{x}>\overline{t}\overline{I}^{(1)}(\mathbf{0};\Sigma_{0}\overline{p})$,  $d_{F}$ is matched with $d_{\infty,\mu^{o}}$. 
\end{itemize}

\end{lem}

\begin{proof}

Note the following results hold under stated assumptions. 

\begin{itemize}
\item[1)] Le Cam's first lemma implies that $\hat{\Sigma}\overset{P_{h}^{n}}{\rightarrow}\Sigma_{0}$
for all $h$.
\item[2)] Lemma \ref{lem:non-diff-continuous} establishes that $\overline{p}(\cdotp)$, $\overline{x}(\cdotp)$  and $\mu^{o}(\cdotp)$ are all continuous functions. Conclude that 
$\hat{p}=\overline{p}(\hat{\Sigma})\overset{P_{h}^{n}}{\rightarrow}\overline{p}(\Sigma_{0})=\overline{p}$, 
 $\hat{x}=\overline{x}(\hat{\Sigma})\overset{P_{h}^{n}}{\rightarrow}\overline{x}(\Sigma_{0})$
and $\hat{\mu}^{o}=\mu^{o}(\hat{\Sigma})\overset{P_{h}^{n}}{\rightarrow}\mu^{o}(\Sigma_0)=\mu^{o}$
for all $h$.  
\item[3)] As $\overline{I}^{(1)}(\mathbf{0};\cdotp)$ and $\underline{I}^{(1)}(\mathbf{0};\cdotp)$
are continuous and $\overline{I}^{(1)}(\mathbf{0};\mathbf{v})\leq\underline{I}^{(1)}(\mathbf{0};\mathbf{v})$
for all $\mathbf{v}$, we have
\begin{align*}
 & \hat{x}-\frac{2\overline{C}\overline{I}^{(1)}(\mathbf{0};\hat{\Sigma}\hat{p})}{\underline{I}^{(1)}(\mathbf{0};\hat{\Sigma}\hat{p})-\overline{I}^{(1)}(\mathbf{0};\hat{\Sigma}\hat{p})}
\overset{P_{h}^{n}}{\rightarrow}  \overline{x}-\frac{2\overline{C}\overline{I}^{(1)}(\mathbf{0};\Sigma_{0}\overline{p})}{\underline{I}^{(1)}(\mathbf{0};\Sigma_{0}\overline{p})-\overline{I}^{(1)}(\mathbf{0};\Sigma_{0}\overline{p})},
\end{align*}
with  $\overline{x}-\frac{2\overline{C}\overline{I}^{(1)}(\mathbf{0};\Sigma_{0}\overline{p})}{\underline{I}^{(1)}(\mathbf{0};\Sigma_{0}\overline{p})-\overline{I}^{(1)}(\mathbf{0};\Sigma_{0}\overline{p})}=-\infty$
in case $\underline{I}^{(1)}(\mathbf{0};\Sigma_{0}\overline{p})-\overline{I}^{(1)}(\mathbf{0};\Sigma_{0}\overline{p})=0$. 
\item[4)] Since $\mu(\gamma_{0})=\mu_{0}=\mathbf{0}$, Le Cam's third lemma
implies $\sqrt{n}\mu(\hat{\gamma})\overset{h}{\rightsquigarrow}\mathcal{N}(G_{0}h,\Sigma_{0})$.
\item[5)]  With the above results, conclude that,  by continuous mapping theorem, we have
\begin{align*}
d_{F,\hat{p}} & \overset{h}{\rightsquigarrow}\mathbf{1}\left\{ \overline{p}^{\top}\mathcal{N}(G_{0}h,\Sigma_{0})\geq0\right\} ,
\end{align*}
and if $\left(\frac{\pi}{2}\overline{p}^{\top}\Sigma_{0}\overline{p}\right)^{1/2}<\overline{C}$,
\begin{align*}
d_{F,RT} & \overset{h}{\rightsquigarrow}\Phi\left(\frac{\overline{p}^{\top}\mathcal{N}(G_{0}h,\Sigma_{0})}{\sqrt{\frac{2}{\pi}\overline{C}^{2}-\overline{p}^{\top}\Sigma_{0}\overline{p}}}\right),
\end{align*}
and if $\mathbf{\mu}^{o}\neq\mathbf{0},$
\[
d_{F,\hat{\mu}^{o}}\overset{h}{\rightsquigarrow}\mathbf{1}\left\{ \left(\mathbf{\mu}^{o}\right)^{\top}\Sigma_{0}^{-1}\mathcal{N}(G_{0}h,\Sigma_{0})\geq0\right\} .
\]
\end{itemize}

\noindent \textbf{Case 1}: Suppose  $\left(\frac{\pi}{2}\overline{p}^{\top}\Sigma_{0}\overline{p}\right)^{1/2}<\overline{C}$. Then, note 
$\mathbf{1}\left\{ \left(\frac{\pi}{2}\hat{p}^{\top}\hat{\Sigma}\hat{p}\right)^{1/2}<\overline{C}\right\} \overset{P_{h}^{n}}{\rightarrow}1$,
\begin{align*}
\mathbf{1}\left\{ \left(\frac{\pi}{2}\hat{p}^{\top}\hat{\Sigma}\hat{p}\right)^{1/2}\geq\overline{C},\overline{x}(\hat{\Sigma})\leq\frac{2\overline{C}\overline{I}^{(1)}(\mathbf{0};\hat{\Sigma}\hat{p})}{\underline{I}^{(1)}(\mathbf{0};\hat{\Sigma}\hat{p})-\overline{I}^{(1)}(\mathbf{0};\hat{\Sigma}\hat{p})}\right\} \overset{P_{h}^{n}}{\rightarrow}0,\\
\mathbf{1}\left\{ \left(\frac{\pi}{2}\hat{p}^{\top}\hat{\Sigma}\hat{p}\right)^{1/2}\geq\overline{C},\overline{x}(\hat{\Sigma})>\frac{2\overline{C}\overline{I}^{(1)}(\mathbf{0};\hat{\Sigma}\hat{p})}{\underline{I}^{(1)}(\mathbf{0};\hat{\Sigma}\hat{p})-\overline{I}^{(1)}(\mathbf{0};\hat{\Sigma}\hat{p})}\right\} \overset{P_{h}^{n}}{\rightarrow}0.
\end{align*}
As $d_{F,\hat{p}}\leq1$, $d_{F,\hat{\mu}^{o}}\leq1$, conclude that, by continuous mapping theorem, 
$d_{F}\overset{h}{\rightsquigarrow}\Phi\left(\frac{\overline{p}^{\top}\mathcal{N}(G_{0}h,\Sigma_{0})}{\sqrt{\frac{2}{\pi}\overline{C}^{2}-\overline{p}^{\top}\Sigma_{0}\overline{p}}}\right)$
as $n\rightarrow\infty$. Then, uniform integrability implies that
\[
\mathbb{E}[d_{F}]\rightarrow\mathbb{E}\left[\Phi\left(\frac{\overline{p}^{\top}\mathcal{N}(G_{0}h,\Sigma_{0})}{\sqrt{\frac{2}{\pi}\overline{C}^{2}-\overline{p}^{\top}\Sigma_{0}\overline{p}}}\right)\right].
\]
As $\varDelta\sim\mathcal{N}(G_{0}h,\Sigma_{0})$, conclude that $d_{F}$
is matched with $d_{\infty,RT}$ in the sense of \eqref{eq:matching.rule}.

\noindent \textbf{Case 2}: Suppose $\left(\frac{\pi}{2}\overline{p}^{\top}\Sigma_{0}\overline{p}\right)^{1/2}\geq\overline{C}$
and $\overline{x}\leq\overline{t}\overline{I}^{(1)}(\mathbf{0};\Sigma_{0}\overline{p})$, which itself contains three separate cases. If $\left(\frac{\pi}{2}\overline{p}^{\top}\Sigma_{0}\overline{p}\right)^{1/2}>\overline{C}$
and $\overline{x}<\overline{t}\overline{I}^{(1)}(\mathbf{0};\Sigma_{0}\overline{p})$.
Then, note, $\mathbf{1}\left\{ \left(\frac{\pi}{2}\hat{p}^{\top}\hat{\Sigma}\hat{p}\right)^{1/2}<\overline{C}\right\} \overset{P_{h}^{n}}{\rightarrow}0$, 
and 
\begin{align*}
&\mathbf{1}\left\{ \left(\frac{\pi}{2}\hat{p}^{\top}\hat{\Sigma}\hat{p}\right)^{1/2}\geq\overline{C},\overline{x}(\hat{\Sigma})\leq\frac{2\overline{C}\overline{I}^{(1)}(\mathbf{0};\hat{\Sigma}\hat{p})}{\underline{I}^{(1)}(\mathbf{0};\hat{\Sigma}\hat{p})-\overline{I}^{(1)}(\mathbf{0};\hat{\Sigma}\hat{p})}\right\} \overset{P_{h}^{n}}{\rightarrow}1,\\
 & \mathbf{1}\left\{ \left(\frac{\pi}{2}\hat{p}^{\top}\hat{\Sigma}\hat{p}\right)^{1/2}\geq\overline{C},\overline{x}(\hat{\Sigma})>\frac{2\overline{C}\overline{I}^{(1)}(\mathbf{0};\hat{\Sigma}\hat{p})}{\underline{I}^{(1)}(\mathbf{0};\hat{\Sigma}\hat{p})-\overline{I}^{(1)}(\mathbf{0};\hat{\Sigma}\hat{p})}\right\} \overset{P_{h}^{n}}{\rightarrow}0.
\end{align*}
It follows by analogous arguments that $d_{F}\overset{h}{\rightsquigarrow}\mathbf{1}\left\{ \overline{p}^{\top}\mathcal{N}(G_{0}h,\Sigma_{0})\geq0\right\} $
and $d_{F}$ is matched with $d_{\infty,\overline{p}}$. Otherwise, we are in two edge cases: $\left(\frac{\pi}{2}\overline{p}^{\top}\Sigma_{0}\overline{p}\right)^{1/2}=\overline{C}$,
in which we must have $0=\overline{x}<\overline{t}\overline{I}^{(1)}(\mathbf{0};\Sigma_{0}\overline{p})$; and $\left(\frac{\pi}{2}\overline{p}^{\top}\Sigma_{0}\overline{p}\right)^{1/2}>\overline{C}$
and $\overline{x}=\overline{t}\overline{I}^{(1)}(\mathbf{0};\Sigma_{0}\overline{p})$. The matching of $d_F$ and $d_{\infty,\overline{p}}$ in these two edge cases are proved in Lemma \ref{lem:non-diff-matching}.

\noindent\textbf{Case 3}:  $\left(\frac{\pi}{2}\overline{p}^{\top}\Sigma_{0}\overline{p}\right)^{1/2}\geq\overline{C}$
and $\overline{x}>\overline{t}\overline{I}^{(1)}(\mathbf{0};\Sigma_{0}\overline{p})$.
In this case, note it must hold $\left(\frac{\pi}{2}\overline{p}^{\top}\Sigma_{0}\overline{p}\right)^{1/2}>\overline{C}$. Applying completely analogous arguments yields that  $d_{F}\overset{h}{\rightsquigarrow}\mathbf{1}\left\{ \left(\mathbf{\mu}^{o}\right)^{\top}\Sigma_{0}^{-1}\mathcal{N}(G_{0}h,\Sigma_{0})\geq0\right\} $
and $d_{F}$ is matched with $d_{\infty,\mu^{o}}$ .
\end{proof}

\begin{lem}\label{lem:non-diff-matching}
Under conditions of Theorem \ref{thm:non-diff-asymptotic}, the following statements hold true. 
\begin{itemize}
\item[(i)] If $\left(\frac{\pi}{2}\overline{p}^{\top}\Sigma_{0}\overline{p}\right)^{1/2}=\overline{C}$,
$d_{F}$ is matched with $d_{\infty,\overline{p}}$ in the sense of
\eqref{eq:matching.rule};

\item[(ii)] If $\left(\frac{\pi}{2}\overline{p}^{\top}\Sigma_{0}\overline{p}\right)^{1/2}>\overline{C}$
and $\overline{x}=\overline{t}\overline{I}^{(1)}(\mathbf{0};\Sigma_{0}\overline{p})$,
$d_{F}$ is matched with $d_{\infty,\overline{p}}$.
\end{itemize} 
\end{lem}

\begin{proof}
\textbf{Statement (i)}.
Suppose $\left(\frac{\pi}{2}\overline{p}^{\top}\Sigma_{0}\overline{p}\right)^{1/2}=\overline{C}$.
As $d_{F}$ is uniformly integrable, it suffices to show that $d_{F}\overset{h}{\rightsquigarrow}X_{0,\overline{p}}:= \mathbf{1}\left\{ \overline{p}^{\top}\mathcal{N}(G_{0}h,\Sigma_{0})\geq0\right\}$. 
Then, analogous to the proof of Lemma \ref{lem:full.diff.lem.3}, it suffices to show that,
for each $h$, we have, as $n\rightarrow\infty$, 
\begin{align*}
Pr\left\{ d_{F}\leq z\right\} \rightarrow & \Phi\left(\frac{-\left(\overline{p}\right)^{\top}G_{0}h}{\left(\left(\overline{p}\right)^{\top}\Sigma_{0}\overline{p}\right)^{1/2}}\right)
\end{align*}
for each $z\in(0,1)$. Now, pick each $z\in(0,1)$, algebra shows 
\begin{align*}
Pr\left\{ d_{F}\leq z\right\}  & =\mathbb{E}\left[\mathbf{1}\left\{ d_{F}\leq z\right\} \right]\\
 & =\mathbb{E}\left[\mathbf{1}\left\{ d_{F,\hat{p}}\leq z\right\} \right]+S_{1,n}(h)+S_{2,n}(h)+S_{3,n}(h),
\end{align*}
where 
\begin{align*}
S_{1,n}(h) & =\mathbb{E}\left[\mathbf{1}\left\{ \left(\frac{\pi}{2}\hat{p}^{\top}\hat{\Sigma}\hat{p}\right)^{1/2}<\overline{C}\right\} \left(\mathbf{1}\left\{ d_{F,RT}\leq z\right\} -\mathbf{1}\left\{ d_{F,\hat{p}}\leq z\right\} \right)\right]\\
S_{2,n}(h) & =-\mathbb{E}\left[\mathbf{1}\left\{ \left(\frac{\pi}{2}\hat{p}^{\top}\hat{\Sigma}\hat{p}\right)^{1/2}\geq\overline{C}\right\} \mathbf{1}\left\{ \overline{x}(\hat{\Sigma})>\frac{2\overline{C}\overline{I}^{(1)}(\mathbf{0};\hat{\Sigma}\hat{p})}{\underline{I}^{(1)}(\mathbf{0};\hat{\Sigma}\hat{p})-\overline{I}^{(1)}(\mathbf{0};\hat{\Sigma}\hat{p})}\right\} \mathbf{1}\left\{ d_{F,\hat{p}}\leq z\right\} \right],\\
S_{3,n}(h) & =\mathbb{E}\left[\mathbf{1}\left\{ \left(\frac{\pi}{2}\hat{p}^{\top}\hat{\Sigma}\hat{p}\right)^{1/2}\geq\overline{C}\right\} \mathbf{1}\left\{ \overline{x}(\hat{\Sigma})>\frac{2\overline{C}\overline{I}^{(1)}(\mathbf{0};\hat{\Sigma}\hat{p})}{\underline{I}^{(1)}(\mathbf{0};\hat{\Sigma}\hat{p})-\overline{I}^{(1)}(\mathbf{0};\hat{\Sigma}\hat{p})}\right\} \mathbf{1}\left\{ d_{F,\hat{\mu}^{o}}\leq z\right\} \right].
\end{align*}
Note $\mathbb{E}\left[\mathbf{1}\left\{ d_{F,\hat{p}}\leq z\right\} \right]=\mathbb{E}\left[\mathbf{1}\left\{ \hat{p}^{\top}\mu(\hat{\gamma})<0\right\} \right]=Pr\left\{ \hat{p}^{\top}\sqrt{n}\mu(\hat{\gamma})<0\right\} $.
As $\hat{p}^{\top}\sqrt{n}\mu(\hat{\gamma})\overset{h}{\rightsquigarrow}\mathcal{N}\left(\overline{p}^{\top}G_{0}h,\overline{p}^{\top}\Sigma_{0}\overline{p}\right)$,
we conclude that $\mathbb{E}\left[\mathbf{1}\left\{ d_{F,\hat{p}}\leq z\right\} \right]\rightarrow\Phi\left(-\frac{\overline{p}^{\top}G_{0}h}{\left(\overline{p}^{\top}\Sigma_{0}\overline{p}\right)^{1/2}}\right)$
as required. Also, note $S_{1,n}(h)\rightarrow0$
by using completely analogous arguments in Steps 2-3 in the proof
of Lemma \ref{lem:full.diff.lem.3}. Moreover, note when $\left(\frac{\pi}{2}\overline{p}^{\top}\Sigma_{0}\overline{p}\right)^{1/2}=\overline{C}$,
it must be the case that $0=\overline{x}<\overline{t}\overline{I}^{(1)}(\mathbf{0};\Sigma_{0}\overline{p})$.
Hence, 
\[\mathbf{1}\left\{ \overline{x}(\hat{\Sigma})>\frac{2\overline{C}\overline{I}^{(1)}(\mathbf{0};\hat{\Sigma}\hat{p})}{\underline{I}^{(1)}(\mathbf{0};\hat{\Sigma}\hat{p})-\overline{I}^{(1)}(\mathbf{0};\hat{\Sigma}\hat{p})}\right\} \overset{P_{h}^{n}}{\rightarrow}0,\]
implying $S_{2,n}(h)\rightarrow0$ and $S_{3,n}(h)\rightarrow0$
as well.

\noindent\textbf{Statement (ii)}. Suppose $\left(\frac{\pi}{2}\overline{p}^{\top}\Sigma_{0}\overline{p}\right)^{1/2}>\overline{C}$
and $\overline{x}=\overline{t}\overline{I}^{(1)}(\mathbf{0};\Sigma_{0}\overline{p})$. The proof is completely analogous. Pick each $z\in(0,1)$. Algebra shows
\begin{align*}
\mathbb{E}\left[\mathbf{1}\left\{ d_{F}\leq z\right\} \right] & =\mathbb{E}\left[\mathbf{1}\left\{ \left(\frac{\pi}{2}\hat{p}^{\top}\hat{\Sigma}\hat{p}\right)^{1/2}\geq\overline{C}\right\} \mathbf{1}\left\{ d_{F,\hat{p}}\leq z\right\} \right]+S_{4,n}(h)+S_{5,n}(h),
\end{align*}
where 
\begin{align*}
S_{4,n}(h) & =\mathbb{E}\left[\mathbf{1}\left\{ \left(\frac{\pi}{2}\hat{p}^{\top}\hat{\Sigma}\hat{p}\right)^{1/2}<\overline{C}\right\} \mathbf{1}\left\{ d_{F,RT}\leq z\right\} \right],\\
S_{5,n}(h) & =\mathbb{E}\left[\mathbf{1}\left\{ \left(\frac{\pi}{2}\hat{p}^{\top}\hat{\Sigma}\hat{p}\right)^{1/2}\geq\overline{C}\right\} \mathbf{1}\left\{ \overline{x}(\hat{\Sigma})>\frac{2\overline{C}\overline{I}^{(1)}(\mathbf{0};\hat{\Sigma}\hat{p})}{\underline{I}^{(1)}(\mathbf{0};\hat{\Sigma}\hat{p})-\overline{I}^{(1)}(\mathbf{0};\hat{\Sigma}\hat{p})}\right\} \left(\mathbf{1}\left\{ d_{F,\hat{\mu}^{o}}\leq z\right\} -\mathbf{1}\left\{ d_{F,\hat{p}}\leq z\right\} \right)\right].
\end{align*}
Note $\mathbf{1}\left\{ \left(\frac{\pi}{2}\hat{p}^{\top}\hat{\Sigma}\hat{p}\right)^{1/2}\geq\overline{C}\right\} \overset{P_{h}^{n}}{\rightarrow}1$,
and by continuous mapping theorem
\begin{align*}
\mathbf{1}\left\{ d_{F,\hat{p}}\leq z\right\} =\mathbf{1}\left\{ \hat{p}^{\top}\mu\left(\hat{\gamma}\right)<0\right\}  & \overset{h}{\rightsquigarrow}\mathbf{1}\left\{ \overline{p}^{\top}\mathcal{N}(G_{0}h,\Sigma_{0})<0\right\}.
\end{align*}
As a result, 
\[
\mathbf{1}\left\{ \left(\frac{\pi}{2}\hat{p}^{\top}\hat{\Sigma}\hat{p}\right)^{1/2}\geq \overline{C}\right\} \mathbf{1}\left\{ d_{F,\hat{p}}\leq z\right\} \overset{h}{\rightsquigarrow}\mathbf{1}\left\{ \overline{p}^{\top}\mathcal{N}(G_{0}h,\Sigma_{0})<0\right\} ,
\]
and we have, as $n\rightarrow\infty$, due to uniform integrability, 
\begin{align*}
 & \mathbb{E}\left[\mathbf{1}\left\{ \left(\frac{\pi}{2}\hat{p}^{\top}\hat{\Sigma}\hat{p}\right)^{1/2}\geq\overline{C}\right\} \mathbf{1}\left\{ d_{F,\hat{p}}\leq z\right\} \right]
{\rightarrow}  Pr\left\{ \overline{p}^{\top}\mathcal{N}(G_{0}h,\Sigma_{0})<0\right\} 
=  \Phi\left(-\frac{\overline{p}^{\top}G_{0}h}{\left(\overline{p}^{\top}\Sigma_{0}\overline{p}\right)^{1/2}}\right).
\end{align*}
Also,  $S_{4,n}(h){\rightarrow}0$ as $n\rightarrow\infty$.
Therefore, it
remains to show $S_{5,n}(h){\rightarrow}0$.  To
this end, note in the case of $\left(\frac{\pi}{2}\overline{p}^{\top}\Sigma_{0}\overline{p}\right)^{1/2}>\overline{C}$
and $\overline{x}=\overline{t}\overline{I}^{(1)}(\mathbf{0};\Sigma_{0}\overline{p})$, we must have $0<\overline{t}<\infty$, $\hat{\mu}^{o}\overset{P_{h}^{n}}{\rightarrow}\mu^{o}=\overline{t}\Sigma_{0}\overline{p}$,
$\hat{\Sigma}^{-1}\overset{P_{h}^{n}}{\rightarrow}\Sigma_{0}^{-1}$,
and $\hat{\Sigma}^{-1}\hat{\mu}^{o}\overset{P_{h}^{n}}{\rightarrow}\overline{t}\overline{p}$.  As a result, let
\begin{align*}
X_{n} & :=\left(\overline{t}\right)^{-1}\left(\hat{\mu}^{o}\right)^{\top}\hat{\Sigma}^{-1}\sqrt{n}\mu\left(\hat{\gamma}\right),Y_{n}:=\hat{p}^{\top}\sqrt{n}\mu\left(\hat{\gamma}\right),
H  :=\mathcal{N}\left(\overline{p}^{\top}G_0h,\left(\overline{p}\right)^{\top}\Sigma_{0}\overline{p}\right).
\end{align*} 
Then, we have  $X_{n}\overset{h}{\rightsquigarrow}H, Y_{n}\overset{h}{\rightsquigarrow}H$,
and 
\[X_{n}-Y_{n}=\left(\left(\overline{t}\right)^{-1}\left(\hat{\mu}^{o}\right)^{\top}\hat{\Sigma}^{-1}-\hat{p}^{\top}\right)\sqrt{n}\mu\left(\hat{\gamma}\right)\overset{P_{h}^{n}}{\rightarrow}0.\]
Then, 
\begin{align*}
\vert S_{5,n}(h)\vert & \leq\mathbb{E}\left[\left|\mathbf{1}\left\{ d_{F,\hat{\mu}^{o}}\leq z\right\} -\mathbf{1}\left\{ d_{F,\hat{p}}\leq z\right\} \right|\right]\\
 & =\mathbb{E}\left[\left|\mathbf{1}\left\{ \left(\hat{\mu}^{o}\right)^{\top}\hat{\Sigma}^{-1}\sqrt{n}\mu\left(\hat{\gamma}\right)<0\right\} -\mathbf{1}\left\{ \hat{p}^{\top}\sqrt{n}\mu\left(\hat{\gamma}\right)<0\right\} \right|\right]\\
 & =\mathbb{E}\left[\left|\mathbf{1}\left\{ \left(\overline{t}\right)^{-1}\left(\hat{\mu}^{o}\right)^{\top}\hat{\Sigma}^{-1}\sqrt{n}\mu\left(\hat{\gamma}\right)<0\right\} -\mathbf{1}\left\{ \hat{p}^{\top}\sqrt{n}\mu\left(\hat{\gamma}\right)<0\right\} \right|\right]\\
 & \leq Pr\left\{ \mathbf{1}\left\{ X_{n}<0\right\} \neq\mathbf{1}\left\{ Y_{n}<0\right\} \right\} ,
\end{align*}
and $Pr\left\{ \mathbf{1}\left\{ X_{n}<0\right\} \neq\mathbf{1}\left\{ Y_{n}<0\right\} \right\} \rightarrow0$
as $n\rightarrow\infty$ by using a completely analogous argument
to Step 3 in the proof of Lemma \ref{lem:full.diff.lem.3}.
\end{proof}

\end{document}